\documentclass[a4paper,english,cleveref,autoref,thm-restate]{lipics-v2021}

\usepackage{esvect}
\usepackage{amsbsy}
\usepackage{bbm}
\usepackage[svgnames]{xcolor}
\usepackage{tikz}
\usepackage{tcolorbox}
\usepackage{enumitem}
\usepackage[hyperpageref]{backref}
\usetikzlibrary{arrows}
\usetikzlibrary{arrows.meta}
\usetikzlibrary{decorations.pathmorphing, decorations.pathreplacing, decorations.shapes}
\usetikzlibrary{decorations.markings}
\usetikzlibrary{calc}
\usetikzlibrary{shapes.geometric}

\pgfkeys{/tikz/.cd,
    contour distance/.store in=\ContourDistance,
    contour distance=0pt, 
    contour step/.store in=\ContourStep,
    contour step=1pt,
}

\pgfdeclaredecoration{closed contour}{initial}
{%
\state{initial}[width=\ContourStep,next state=cont] {
    \pgfmoveto{\pgfpoint{\ContourStep}{\ContourDistance}}
    \pgfcoordinate{first}{\pgfpoint{\ContourStep}{\ContourDistance}}
    \pgfpathlineto{\pgfpoint{0.3\pgflinewidth}{\ContourDistance}}
    \pgfcoordinate{lastup}{\pgfpoint{1pt}{\ContourDistance}}
    
  }
  \state{cont}[width=\ContourStep]{
     \pgfmoveto{\pgfpointanchor{lastup}{center}}
     \pgfpathlineto{\pgfpoint{\ContourStep}{\ContourDistance}}
     \pgfcoordinate{lastup}{\pgfpoint{\ContourStep}{\ContourDistance}}
  }
  \state{final}[width=\ContourStep]
  { 
    \pgfmoveto{\pgfpointanchor{lastup}{center}}
    \pgfpathlineto{\pgfpointanchor{first}{center}}
  }
}

\tikzset{
  cir/.style = {circle,draw,fill,inner sep=.7pt},
  circ/.style = {circle,draw,fill,inner sep=1.3pt},
  circg/.style = {circle,draw=lightgray,fill=lightgray,inner sep=1.3pt},
  circr/.style = {circle,draw=Crimson,fill=Crimson,inner sep=1.3pt},
  invisible/.style = {circle,draw=none,inner sep=0pt,font=\tiny},
  nonedge/.style={decorate,decoration={snake,amplitude=.3mm,segment length=1mm},draw},
}

\tcbuselibrary{skins}

\newtcolorbox{mybox}[1]{minipage boxed title*=-2cm,
enhanced,attach boxed title to top center=
{yshift=-3mm,yshifttext=-1mm},colback=Lavender!30!white,
boxed title style={size=small,colback=Lavender},coltitle=black,
center title,title={#1}}

\newcommand{\vol}{\mathrm{vol}}
\newcommand{\rk}{\mathrm{rk}}
\newcommand{\tw}{\mathrm{tw}}
\newcommand{\tin}{\mathrm{tree}\textnormal{-}\alpha}

\newtheorem{question}{Question}

\mathchardef\mhyphen="2D

\title{Polynomial-time approximation schemes for induced subgraph problems on fractionally tree-independence-number-fragile graphs} 

\titlerunning{PTASes for induced subgraph problems on fractionally $\boldsymbol{\tin}$-fragile graphs} 

\author{Esther {Galby}}{Hamburg University of Technology, Institute for Algorithms and Complexity,
Hamburg, Germany}{esther.galby@tuhh.de}{}{}

\author{Andrea {Munaro}}{Department of Mathematical, Physical and Computer Sciences, University of Parma, Parma, Italy}{andrea.munaro@unipr.it}{https://orcid.org/0000-0003-1509-8832}{}

\author{Shizhou {Yang}}{School of Mathematics and Physics, Queen’s University Belfast, Belfast, UK}{syang22@qub.ac.uk}{}{}

\authorrunning{E. Galby, A. Munaro, S. Yang} 

\Copyright{Esther Galby and Andrea Munaro and Shizhou Yang} 



\category{} 

\relatedversion{A preliminary version containing some results of this paper appeared in the Proceedings of SoCG 2023 \cite{GMY23}.}

\nolinenumbers 

\hideLIPIcs

\EventEditors{Erin W. Chambers and Joachim Gudmundsson}
\EventNoEds{2}
\EventLongTitle{39th International Symposium on Computational Geometry
(SoCG 2023)}
\EventShortTitle{SoCG 2023}
\EventAcronym{SoCG}
\EventYear{2023}
\EventDate{June 12--15, 2023}
\EventLocation{Dallas, Texas, USA}
\EventLogo{}
\SeriesVolume{258}
\ArticleNo{XX} 

\makeatletter
\let\orgdescriptionlabel\descriptionlabel
\renewcommand*{\descriptionlabel}[1]{%
  \let\orglabel\label
  \let\label\@gobble
  \phantomsection
  \edef\@currentlabel{#1\unskip}%
  \let\label\orglabel
  \orgdescriptionlabel{#1}%
}
\makeatother

\begin{document}

\maketitle

\begin{abstract}
We investigate a relaxation of the notion of fractional treewidth-fragility, namely fractional tree-independence-number-fragility. In particular, we obtain polynomial-time approximation schemes for meta-problems such as finding a maximum-weight sparse induced subgraph satisfying a given $\mathsf{CMSO}_2$ formula on fractionally tree-independence-number-fragile graph classes. Our approach unifies and extends several known polynomial-time approximation schemes on seemingly unrelated graph classes, such as classes of intersection graphs of fat objects in a fixed dimension or proper minor-closed classes. We also study the related notion of layered tree-independence number, a relaxation of layered treewidth, and its applications to exact subexponential-time algorithms.
\end{abstract}

\section{Introduction}
\label{sec:introA}
Many optimization problems involving collections of geometric objects in the $d$-dimensional space are known to admit a polynomial-time approximation scheme (PTAS). Arguably the earliest example of such behavior is the problem of finding the maximum number of pairwise non-intersecting disks or squares in a collection of unit disks or unit squares, respectively \cite{HM85}. Such subcollection is usually called an \textit{independent packing}. This result was later extended to collections of arbitrary disks and squares and, more generally, fat objects \cite{Cha03,EJS05}. The reason for the abundance of approximation schemes for geometric problems is that shifting and layering techniques can be used to reduce the problem to small subproblems that can be solved by dynamic programming. In fact, the same phenomenon occurs for graph problems, as evidenced by the seminal work of Baker~\cite{Bak94} on approximation schemes for local problems, such as \textsc{Independent Set}, on planar graphs and its generalizations first to apex-minor-free graphs \cite{Epp00} and further to graphs embeddable on a surface of bounded genus with a bounded number of crossings per edge \cite{GB07}. The notion of intersection graph allows to jump from the geometric world to the graph-theoretic one. Given a collection $\mathcal{O}$ of geometric objects in $\mathbb{R}^d$, we can consider its \textit{intersection graph}, the graph whose vertices are the objects in $\mathcal{O}$ and where two distinct vertices $O_i, O_j \in \mathcal{O}$ are adjacent if and only if $O_i \cap O_j \neq \varnothing$. An independent packing in $\mathcal{O}$ is then nothing but an independent set in the corresponding intersection graph. Notice that intersection graphs of unit disks or squares are not minor-closed, as they contain arbitrarily large cliques. Our motivating question is the following:
\begin{center}
\textit{Is there any underlying graph-theoretical reason for the existence of the seemingly unrelated PTASes for \textsc{Independent Set} mentioned above?} 
\end{center}
We provide a positive answer to this question that also allows us to further generalize to a framework of maximization problems. We also remark that the similar question of whether there is a general notion under which PTASes using Baker's technique can be obtained was already asked in \cite{GB07}. 

Baker's layering technique relies on a form of decomposition theorem for planar graphs that can be roughly summarized as follows. Given a planar graph $G$ and $k \in \mathbb{N}$, the vertex set of $G$ can be partitioned into $k$ possibly empty sets in such a way that deleting any part induces in $G$ a graph of treewidth at most $O(k)$. Moreover, such a partition together with tree decompositions of width at most $O(k)$ of the respective graphs can be found in polynomial time. A statement of this form is typically referred to as a \textit{Vertex Decomposition Theorem} (VDT) \cite{PSZ19}. VDTs are known to exist in planar graphs \cite{Bak94}, graphs of bounded-genus and apex-minor-free graphs \cite{Epp00}, and $H$-minor-free graphs \cite{DHK05,DDO04}. However, their existence is in general something too strong to ask for, as is the case of intersection graphs of unit disks or squares and hence fat objects in general. There are then two natural ways in which one can try to relax the notion of VDT. First, we can consider an approximate partition of the vertex set, where a vertex can belong to some constant number of sets. Second, we can look for a width parameter less restrictive than treewidth.

Dvo\v{r}\'{a}k~\cite{Dvo16} pursued the first direction and introduced the notion of efficient fractional treewidth-fragility. We state here an equivalent formulation from \cite{DL21}. A class of graphs $\mathcal{G}$ is \textit{efficiently fractionally treewidth-fragile} if there exists a function $f\colon \mathbb{N} \rightarrow \mathbb{N}$ and an algorithm that, for every $k \in \mathbb{N}$ and $G \in \mathcal{G}$, returns in time $\mathsf{poly}(|V(G)|)$ a collection of subsets $X_1, X_2,\ldots, X_m \subseteq V(G)$ such that each vertex of $G$ belongs to at most $m/k$ of the subsets and moreover, for $i = 1,\ldots, m$, the algorithm also returns a tree decomposition of $G - X_i$ of width at most $f(k)$. The ``efficiently'' in the definition refers to the existence of such a polynomial-time algorithm; in case such an algorithm is not known, we simply talk about fractional treewidth-fragility. Several graph classes are known to be efficiently fractionally treewidth-fragile. In fact, a class $\mathcal{G}$ is efficiently fractionally treewidth-fragile in each of the following cases (see, e.g., \cite{DL21}): $\mathcal{G}$ is subgraph-closed and has strongly sublinear separators and bounded maximum degree, $\mathcal{G}$ is proper minor-closed, or $\mathcal{G}$ consists of intersection graphs of convex objects with bounded aspect ratio in $\mathbb{R}^d$ (for fixed $d$) and the graphs in $\mathcal{G}$ have bounded clique number. Dvo\v{r}\'{a}k~\cite{Dvo16} showed that \textsc{Independent Set} admits a PTAS on every efficiently fractionally treewidth-fragile class. This result was later extended \cite{Dvo22,DL21} to a framework of maximization problems including, for example, \textsc{Max Weight Distance-$d$ Independent Set}, \textsc{Max Weight Induced Forest} and \textsc{Max Weight Induced Matching}. However, the notion of fractional treewidth-fragility falls short of capturing classes such as unit disk graphs, as it implies the existence of sublinear separators \cite{Dvo16}.  

One can then try to pursue the second direction mentioned above and further relax the notion of efficient fractional treewidth-fragility by considering width parameters \textit{more powerful} than treewidth (i.e., bounded on a larger class of graphs) and algorithmically useful. A natural candidate is \textit{tree-independence number}, a width parameter defined in terms of tree decompositions which is more powerful than treewidth (see \Cref{sec:layeredA}), introduced independently by Dallard et al.~\cite{DaMS22} and Yolov~\cite{Yol18}. Several algorithmic applications of boundedness of tree-independence number have been provided, most notably polynomial-time solvability of $(c, \psi)$-\textsc{Max Weight Induced Subgraph} (informally, the meta-problem of finding a maximum-weight induced subgraph with clique number at most $c$ satisfying a certain $\mathsf{CMSO}_2$ formula $\psi$) \cite{LMM24}, \textsc{Max Weight Independent Packing} (informally, the problem of packing pairwise disjoint and non-adjacent weighted graphs from a family $\mathcal{H}$ given in input) \cite{DaMS22}, and its distance-$d$ version, for $d$ even \cite{LMM24}. These are generalizations of problems such as \textsc{Max Weight Distance-$d$ Independent Set}, \textsc{Max Weight Induced Forest} and \textsc{Max Weight Induced Matching}. Investigating the notion of efficient fractional tree-independence-number-fragility ($\tin$-fragility for short) was recently suggested in a talk by Dvo\v{r}\'{a}k \cite{GWP22}, where it was stated that, using an argument from \cite{DGLTT22}, it is possible to show that intersection graphs of balls and cubes in $\mathbb{R}^d$ are fractionally $\tin$-fragile.   

A successful notion related to fractional treewidth-fragility is the layered treewidth of a graph \cite{DMW17}. Loosely speaking, a graph $G$ has small layered treewidth if there exist a tree decomposition and a layering of $G$ such that the intersection between any bag and layer is small. In particular, the union of any constant number of consecutive layers induces a subgraph of small treewidth and hence Baker's technique applies (provided a suitable layering can be found in polynomial time) \cite{Dvo18,Dvo20}. Besides being algorithmically interesting, this notion proved useful especially in the context of coloring-type problems (we refer to \cite{DEMWW22} for additional references). It should be mentioned that classes of bounded layered treewidth include planar graphs and, more generally, apex-minor-free graphs and graphs embeddable on a surface of bounded genus with a bounded number of crossings per edge, amongst others \cite{DEW17}. It can be shown that bounded layered treewidth implies fractional treewidth-fragility (see \Cref{sec:fragilityA}). Layered treewidth is also related to local treewidth, a notion first introduced by Eppstein~\cite{Epp00}, and in fact, on proper minor-closed classes, having bounded layered treewidth coincides with having bounded local treewidth (see, e.g., \cite{DEW17}).

\subsection{Main results}\label{sec:mainres}

In this paper, we investigate the notion of efficient fractional $\tin$-fragility, which generalizes efficient fractional treewidth-fragility and bounded tree-independence number, and show that it answers our motivating question in the positive and allows to unify and extend several known results. More precisely, we provide an approximation meta-theorem which pushes the limits of tractability well beyond the state of the art by extending existing approximation meta-theorems to the broad family of efficiently fractionally $\tin$-fragile classes. In this way, we also give a uniform and natural explanation for a large number of algorithmic results. The family of problems covered by our framework belongs to the general meta-problem introduced by Lund and Yannakakis~\cite{LY93} and called \textsc{Max Induced $\Pi$-Subgraph}: Given a graph $G$, the task is to find a maximum-size \textit{induced} subgraph of $G$ satisfying a certain fixed property $\Pi$. Loosely speaking, our framework captures the family of problems whose task is to find a \textit{maximum-weight sparse} induced subgraph satisfying some fixed \textit{hereditary} property expressible in counting monadic second-order logic ($\mathsf{CMSO}_2$). Counting monadic second order logic is a counting variant of monadic second-order logic ($\mathsf{MSO}_2$), where one is allowed to have atomic formulae expressing that the cardinality of a set is equal to $q$ modulo $p$, for some integers $p\geq2$ and $0\leq q < p$. As we will not directly work with the formalism of $\mathsf{CMSO}_2$, we refer the reader to \cite{CE} for an introduction to this subject. 

We can finally define the meta-problem called $(c, h, \psi)$-\textsc{Max Weight Induced Subgraph}, which will be the focus of our paper. For $h \in \mathbb{N}$, we say that a $\mathsf{CMSO}_2$ formula $\psi$ \textit{expresses an $h$-near-monotone property} if, for any graph $G$ and any subset $Y \subseteq V(G)$ with $G[Y] \models \psi$, there exists a system $\{R_v \subseteq Y : v \in Y\}$ of subsets of $Y$ such that $v \in R_v$ for each $v \in Y$, each vertex of $Y$ belongs to $R_v$ for at most $h$ vertices $v \in Y$, and $G[Y\setminus \bigcup_{v\in X} R_v] \models \psi$ for each $X \subseteq Y$. 
The interesting special case $h = 1$ is that of a formula $\psi$ expressing a monotone property, i.e., a formula $\psi$ such that for any graph $G$ and any subsets $X \subseteq Y \subseteq V(G)$, if $G[Y] \models \psi$, then $G[X] \models \psi$. Let $\psi$ be a fixed $\mathsf{CMSO}_2$ formula expressing an $h$-near-monotone property and let $c$ be a fixed positive integer.

\begin{mybox}{$(c, h, \psi)$-\textsc{Max Weight Induced Subgraph}}
\textbf{Input:} A graph $G$ equipped with a weight function $w\colon V(G) \rightarrow \mathbb{Q}_{+}$.\\
\textbf{Task:} Find a set $F \subseteq V(G)$ such that:
\begin{enumerate}
\item $G[F] \models \psi$,
\item $\omega(G[F]) \leq c$,
\item $F$ is of maximum weight subject to the conditions above,
\end{enumerate}
or conclude that no such set exists.
\end{mybox}

The meta-problem $(c, h, \psi)$-\textsc{Max Weight Induced Subgraph} captures several well-known problems, such as \textsc{Max Weight Independent Set}, \textsc{Max Weight Induced Matching}, \textsc{Max Weight Induced Forest} (see \Cref{sec:ptasesA} for several other examples). This framework is in fact closely related to those considered in \cite{FTV15,LMM24}, in the context of exact algorithms, and in \cite{Dvo22,DL21}, in the context of approximation algorithms. We briefly highlight the differences. We already mentioned that Lima et al.~\cite{LMM24} showed that, for each fixed $c$ and $\psi$, $(c, \psi)$-\textsc{Max Weight Induced Subgraph} (which is nothing but the problem obtained from $(c, h, \psi)$-\textsc{Max Weight Induced Subgraph} by dropping the $h$-near-monotonicity constraint) can be solved in polynomial time for graphs of bounded tree-independence number. However, as we shall see in \Cref{sec:ptasesA}, some sort of monotonicity constraint is inevitable in our case. The main difference between the $(c, h, \psi)$-\textsc{Max Weight Induced Subgraph} framework and those provided in \cite{Dvo22,DL21} is that ours does not allow to model problems defined in terms of distances. However, this is again inevitable, as the parity of the distances plays a role: For each $\varepsilon > 0$ and fixed odd $d \geq 3$, it is $\mathsf{NP}$-hard to approximate the distance-$d$ version of \textsc{Independent Set} to within a factor of $n^{1/2 - \varepsilon}$ for chordal graphs \cite{EGM14}, which coincides with the class of graphs with tree-independence number $1$ \cite{DaMS22}. To partially overcome this, we also consider the framework of \textsc{Max Weight Independent Packing} and its distance-$d$ version called \textsc{Max Weight Distance-$d$ Packing} (both formally defined in \Cref{sec:ptasesA}). Our main results can be summarized as follows (see also \Cref{reldiagram}).
\begin{enumerate}[label={\fontfamily{lmss}\selectfont\textbf{(\Alph*)}},ref=\Alph*]
\item\label{item:first} For each fixed $c, h \in \mathbb{N}$ and $\mathsf{CMSO}_2$ formula $\psi$, $(c, h, \psi)$-\textsc{Max Weight Induced Subgraph} admits a PTAS on every efficiently fractionally $\tin$-fragile class. 
\item\label{item:second} \textsc{Max Weight Independent Packing} admits a PTAS on every efficiently fractionally $\tin$-fragile class.
\item\label{item:third} Every class of intersection graphs of fat objects in $\mathbb{R}^d$, for fixed $d$, is efficiently fractionally $\tin$-fragile.
\end{enumerate}
We also introduce the notion of layered tree-independence number, which is a relaxation of layered treewidth and which provides a strengthening of fractional $\tin$-fragility (as we explain in the next section), and prove the following.
\begin{enumerate}[label={\fontfamily{lmss}\selectfont\textbf{(\Alph*)}},ref=\Alph*,resume]
\item\label{item:fourth} For each fixed even $p \in \mathbb{N}$, \textsc{Max Weight Distance-$p$ Packing} admits a PTAS on every class of bounded layered tree-independence number\footnote{Provided that a tree decomposition and a layering witnessing small layered tree-independence number can be computed efficiently.} and on every class of intersection graphs of fat objects in $\mathbb{R}^d$, for fixed $d$.
\end{enumerate}

\begin{remark} Results \ref{item:first}, \ref{item:second}, \ref{item:fourth} cannot be improved to guarantee EPTASes, unless $\mathsf{FPT} = \mathsf{W}[1]$. Indeed, Marx~\cite{Mar05} showed that \textsc{Independent Set} remains $\mathsf{W}[1]$-complete on intersection graphs of unit disks and unit squares.
\end{remark}

The main message of our work is that a doubly-relaxed version of a VDT suffices for algorithmic applications and is general enough to hold for several interesting graph classes. In particular, Result \ref{item:first} provides an approximation meta-theorem similar to those obtained in \cite{Dvo22,DL21} but applicable to the substantially broader family of efficiently fractionally $\tin$-fragile classes. It should also be noted that there exist several definitions of fatness in the literature and the one we adopt in this paper slightly generalizes that of Chan~\cite{Cha03} and is implicitly used in some of the arguments from \cite{HP17}. Informally, a collection of objects\footnote{We remark that the objects need not be convex nor similarly-sized, i.e., the ratio of the largest and smallest diameter of the objects need not be bounded by a fixed constant.} is fat according to our definition if it satisfies a sort of ``low-density property'': for each $r$, there is at most a constant number of pairwise non-intersecting objects of size at least $r$ intersecting any region of size $r$ (see \Cref{fatcomp} for the formal definition). In fact, our notion of fatness captures that of \textit{low density}, introduced by Har-Peled and Quanrud in \cite{HP17}. Therefore, Result \ref{item:first} also extends some of the local-search-based PTASes from \cite{HP17} for \textsc{Max Induced $\Pi$-Subgraph}, where $\Pi$ is a hereditary and mergeable\footnote{A property $\Pi$ is mergeable if, for any subsets of vertices $X, Y$ which are at distance at least $2$, if $X$ and $Y$ each satisfy $\Pi$, then $X \cup Y$ satisfies $\Pi$.} property, on intersection graphs of collections of low-density objects in $\mathbb{R}^d$, for fixed $d$ (or, more generally, collections of objects where every subcollection of pairwise non-intersecting objects has low density). 

The natural trade-off in extending the tractable families with respect to approximation is that fewer problems will admit a PTAS. In our case this is exemplified by the minimization problem \textsc{Feedback Vertex Set}, which admits no PTAS on unit ball graphs in $\mathbb{R}^3$, unless $\mathsf{P} = \mathsf{NP}$ \cite{FLS12}, but admits an EPTAS on disk graphs \cite{LPS23}. In fact, Lokshtanov et al.~\cite{LPS23} established a framework for designing EPTASes for a broad class of unweighted vertex-deletion problems on disk graphs including, among others, \textsc{Feedback Vertex Set} (the complement dual of \textsc{Max Induced Forest}) and $d$-\textsc{Bounded Degree Vertex Deletion} (the complement dual of \textsc{Max $d$-Dependent Set}, a problem defined in \Cref{sec:ptasesA}). Previous sporadic PTASes on this class were known only for \textsc{Vertex Cover} \cite{EJS05,vLee06}, \textsc{Dominating Set} \cite{GB10}, \textsc{Independent Set} \cite{Cha03,EJS05} and \textsc{Max Clique} \cite{BBB21}. Very recently, Dvo\v{r}\'{a}k et al.~\cite{DLP23}, generalizing \cite{LPS23}, provided EPTASes for the unweighted minimization meta-problem \textsc{(Induced) Subgraph Hitting} on graph classes with polynomial expansion (which, in the case of subgraph-closed classes, is equivalent to having strongly sublinear separators \cite{DN16}) and on intersection graphs of convex globally fat\footnote{\label{note1}Global fatness is the ``standard'' notion of fatness, typically referred to as fatness. Given $k \geq 1$, an object $O \subseteq \mathbb{R}^d$ is $k$-globally fat if there exist two $d$-dimensional balls $B_{\mbox{\tiny{in}}}$ and $B_{\mbox{\tiny{out}}}$ with radius $R_{\mbox{\tiny{in}}}$ and $R_{\mbox{\tiny{out}}}$, respectively, such that $B_{\mbox{\tiny{in}}} \subseteq O \subseteq B_{\mbox{\tiny{out}}}$ and $R_{\mbox{\tiny{out}}} \leq k \cdot R_{\mbox{\tiny{in}}}$.} objects and pseudo-disks. Results \ref{item:first}, \ref{item:second} and \ref{item:fourth} complement the results from \cite{DLP23,LPS23}, as $(c, h, \psi)$-\textsc{Max Weight Induced Subgraph} and \textsc{Max Weight Independent Packing} capture the \textit{weighted} dual of several problems addressed therein. For example, every problem which consists in computing a subset $S \subseteq V(G)$ of minimum size such that $G - S$ does not contain any graph from a fixed finite family $\mathcal{F}$ as a (induced) subgraph and $G - S$ has bounded clique number. Observe that any fixed finite family $\mathcal{F}$ containing a complete graph satisfies this requirement, and in this way one can obtain problems such as $d$-\textsc{Bounded Degree Vertex Deletion}, \textsc{$C_k$-Hitting} and \textsc{$\ell$-Component Order Connectivity}.

\subsection{Overview of the results and organization of the paper}

\textit{Fatness.} In \Cref{fatcomp}, we introduce our notion of fatness, called $c$-fatness, and compare it with three among the most general (in the sense that they apply to arbitrary objects in arbitrary dimensions) notions of fatness from the literature, which we call global fatness\footnote{See \Cref{note1}.}, local fatness and thickness (see \cite{SHO93}). We show, in particular, that all these notions are equivalent when restricting to convex objects.

\textit{Layered and local tree-independence number.} In \Cref{sec:layeredA}, we begin our study of fractional $\tin$-fragility by introducing a subclass of fractionally $\tin$-fragile graphs, namely the class of graphs with bounded layered tree-independence number. We obtain the notion of layered tree-independence number by relaxing the successful notion of layered treewidth and show that, besides graphs of bounded layered treewidth, the following classes have bounded layered tree-independence number:
\begin{itemize}  
\item Intersection graphs of similarly-sized $c$-fat objects in $\mathbb{R}^2$ (in particular, unit disk graphs); 
\item Intersection graphs of unit-width rectangles in $\mathbb{R}^2$; 
\item (Vertex and edge) intersection graphs of paths with bounded horizontal part on a grid (these classes contain some interesting families of string graphs). 
\end{itemize}
As a consequence, we show that graphs in these classes have $O(\sqrt{n})$ tree-independence number and that this is tight up to constant factors. Moreover, we observe that, for minor-closed classes, having bounded layered tree-independence number is equivalent to having bounded local tree-independence number, which in turn is equivalent to excluding an apex graph as a minor, thus extending a characterization of bounded layered treewidth stated in \cite{DEW17}. We also consider the behavior of layered tree-independence number with respect to graph powers and show that odd powers of graphs of bounded layered tree-independence number have bounded layered tree-independence number and that this does not extend to even powers. This result is crucial for the proof of Result \ref{item:fourth}.  

\textit{Fractional $\tin$-fragility.} In \Cref{sec:fragilityA}, we formally define (efficient) fractional $\tin$-fragility and show that every class of bounded layered tree-independence number is fractionally $\tin$-fragile (in fact, efficiently fractionally $\tin$-fragile, provided that a tree decomposition and a layering witnessing small layered tree-independence number can be computed efficiently). It is then natural to identify necessary conditions for fractional $\tin$-fragility. Graphs of low treewidth must have small-size balanced separators (and the converse holds in subgraph-closed classes \cite{DN19}), whereas graphs of low tree-independence number must have balanced separators with small independence number. Contrary to fractional treewidth-fragility, where sublinear-size balanced separators are needed \cite{Dvo16}, one should then expect that in the case of fractional $\tin$-fragility it is not the size of a separator that has to be small but rather its independence number. Indeed, we show that every fractionally $\tin$-fragile class has balanced separators of sublinear independence number. This result implies, unsurprisingly, that $3$-regular expanders and intersection graphs of rectangles in the plane are not fractionally $\tin$-fragile. Whether \textsc{Independent Set} admits a PTAS on intersection graphs of rectangles in the plane remains a major open problem (see, e.g., \cite{GKM22}).

\textit{Intersection graphs of fat objects.} In \Cref{fatA}, we investigate families of intersection graphs of fat objects in bounded dimensional spaces. In particular, we show that such families are efficiently fractionally $\tin$-fragile (Result \ref{item:third}) and that they are closed under taking odd powers. The latter result is used in the proof of Result \ref{item:fourth}.  

\textit{PTAS frameworks.} In \Cref{sec:ptasesA}, we formally define \textsc{Max Weight Independent Packing} and its distance-$d$ version and provide several examples of problems captured by this framework and that of $(c, h, \psi)$-\textsc{Max Weight Induced Subgraph}. Moreover, we prove Results \ref{item:first}, \ref{item:second}, \ref{item:fourth}. Given the generality of the frameworks and the broadness of the graph classes to which they are applicable, the running times obtained are typically not competitive. However, in \Cref{sec:improvedA}, we focus on a specific problem, namely \textsc{Max Weight Independent Set}, and show how tree-independence number arguments can still lead to competitive PTASes for some classes of intersection graphs. Specifically, we obtain PTASes for \textsc{Max Weight Independent Set} for intersection graphs of families of unit disks, unit-width (or, equivalently, unit-height) rectangles, and paths with bounded horizontal part on a grid, which improve or generalize results from \cite{BBCGP20,Cha04,Mat98} mentioned in the next section. These results were all obtained by applying the shifting technique: consider subgraphs whose geometric realizations are contained in narrow strips and exploit properties of such graphs for dynamic programming. We show that the common structural reason that allows fast dynamic programming algorithms on these graphs is in fact boundedness of tree-independence number.

\textit{Subexponential-time algorithms.} In \Cref{sec:subexp}, we depart from the main topic of the paper, namely approximation schemes, and note some interesting consequences of \Cref{sec:layeredA} in relation to subexponential-time algorithms that can be summarized as follows. There exists a subexponential-time algorithm for \textsc{Max Weight Distance-$d$ Packing}, for each fixed even $d\in\mathbb{N}$, on each class of bounded layered tree-independence number (provided a tree decomposition and a layering witnessing this can be computed efficiently). In particular, we obtain a $2^{O(\sqrt{n}\log{n})}$-time algorithm for \textsc{Max Weight Distance-$d$ Packing} on intersection graphs of similarly-sized $c$-fat families of objects in $\mathbb{R}^2$. This is related to the seminal work of de Berg et al.~\cite{BBK20}, who provided $2^{O(\sqrt{n})}$-time algorithms for the \textit{unweighted} version of many problems on intersection graphs of similarly-sized globally fat objects in $\mathbb{R}^d$. For weighted problems, the situation is much more obscure and de Berg and Kisfaludi{-}Bak~\cite{BK20} asked to determine the complexity of the weighted versions of problems falling in the framework of \cite{BBK20} when restricted to intersection graphs of similarly-sized fat objects in $\mathbb{R}^2$. Our $2^{O(\sqrt{n}\log{n})}$ upper bound partially answers this question, as \textsc{Max Weight Distance-$d$ Packing} captures the weighted version of several problems from \cite{BBK20}.

In \Cref{sec:remarks}, we conclude the paper with some open questions whose answers allow to identify the applicability limits of our approximation frameworks.

\begin{remark}
All our PTASes for intersection graphs of geometric objects are not robust, i.e., they require a geometric realization to be part of the input. 
\end{remark}

\subsection{Other consequences of our work}

All problems mentioned in this section are defined in \Cref{sec:ptasesA}.

\textbf{Disk graphs.} Li et al.~\cite{LSH18} provided a PTAS for \textsc{Max $\mathcal{H}_k$-Free Node Set} when restricted to disk graphs of bounded heterogeneity\footnote{The heterogeneity of a disk graph is the ratio of the maximum radius to the minimum radius of disks.} or, in other words, intersection graphs of similarly-sized disks. Moreover, they asked whether the assumption of bounded heterogeneity is necessary and what happens to the weighted version of the problem. Results \ref{item:first} and \ref{item:third} answer the first question in the negative and show that a PTAS for the weighted version can be obtained in a very general setting.

\textbf{Unit disk graphs.} Unit disk graphs are arguably one of the most well-studied graph classes in computational geometry, as they naturally model several real-world problems. Great attention has been devoted to approximation algorithms for \textsc{Max Weight Independent Set} on this class (see, e.g., \cite{HMR98,NHK04,vLee05}). To the best of our knowledge, the fastest known PTAS is a $(1-1/k)$-approximation algorithm with running time $O(k n^{4\lceil \frac{2(k-1)}{\sqrt{3}}\rceil})$ \cite{Mat98}. In \Cref{sec:ptasesA} we improve on this running time and provide a $(1-1/k)$-approximation algorithm with running time $O(\lceil 3k \rceil n^{3\lceil\frac{3k-1}{2}\rceil+3})$. We also remark that a special type of Decomposition Theorem was recently shown to hold for the class of unit disk graphs. A Contraction Decomposition Theorem (CDT) is a statement of the following form: given a graph $G$, for any $p \in \mathbb{N}$, one can partition the edge set of $G$ into $E_1,\ldots, E_p$ such that contracting the edges in each $E_i$ in $G$ yields a graph of treewidth at most $f(p)$, for some function $f\colon \mathbb{N} \rightarrow \mathbb{N}$. CDTs are useful in designing efficient approximation and parameterized algorithms and are known to hold for classes such as unit disk graphs \cite{BLLSX22} and graphs of bounded genus \cite{DHM10}. Since these classes are efficiently fractionally $\tin$-fragile, our results can be seen as providing a different type of relaxed decomposition theorem for them.

\textbf{Intersection graphs of unit-height rectangles.} As observed by Agarwal et al.~\cite{AKS98}, this class of graphs arises naturally as a model for the problem of labeling maps with labels of the same font size. Improving on \cite{HM85}, they obtained a $(1 - 1/k)$-approximation algorithm for \textsc{Max Weight Independent Set} on this class with running time $O(n^{2k-1})$. Chan~\cite{Cha04} provided a $(1-1/k)$-approximation algorithm with running time $O(n^k)$. Jana et al.~\cite{JMMR20} provided a PTAS for \textsc{Max Bipartite Subgraph} on intersection graphs of unit squares (and unit disks) and a $2$-approximation algorithm for intersection graphs of unit-height rectangles. Moreover, they asked whether the problem admits in fact a PTAS on the latter class. Result \ref{item:first}, together with the fact that unit-height rectangle graphs are efficiently fractionally $\tin$-fragile, answer this question in the positive. 

\textbf{Intersection graphs of paths on a grid.} Asinowski et al.~\cite{asinowski} introduced the class of \textit{Vertex intersection graphs of Paths on a Grid} (\textit{VPG graphs} for short). A graph $G$ is a \textit{VPG graph} if there exists a collection $\mathcal{P}$ of paths on a grid $\mathcal{G}$ such that $\mathcal{P}$ is in one-to-one correspondence with $V(G)$ and two vertices are adjacent in $G$ if and only if the corresponding paths intersect. It is not difficult to see that this class coincides with the well-known class of string graphs. If every path in $\mathcal{P}$ has at most $k$ \textit{bends}, i.e., $90$ degrees turns at a grid-point, the graph is a \textit{$B_k$-VPG graph}. Golumbic et al.~\cite{GLS09} introduced the class of \textit{Edge intersection graphs of Paths on a Grid} (\textit{EPG graphs} for short) which is defined similarly to VPG, except that two vertices are adjacent if and only if the corresponding paths share a grid-edge. It turns out that every graph is EPG \cite{GLS09} and $B_{k}$-EPG graphs have been defined similarly to $B_{k}$-VPG graphs. Approximation algorithms for \textsc{Independent Set} on VPG and EPG graphs have been deeply investigated, especially when the number of bends is a small constant (see, e.g., \cite{BD17,BCK22,FP11,LMS15}). It is an open problem whether \textsc{Independent Set} admits a PTAS on $B_1$-VPG graphs \cite{BD17,LMS15}. Concerning EPG graphs, Bessy et al.~\cite{BBCGP20} showed that the problem admits no PTAS on $B_{1}$-EPG graphs, unless $\mathsf{P} = \mathsf{NP}$, even if each path has its vertical segment or its horizontal segment of length at most $1$. On the other hand, they provided a PTAS for \textsc{Independent Set} on $B_1$-EPG graphs where the length of the horizontal part\footnote{The \textit{horizontal part} of a path is the interval corresponding to the projection of the path onto the $x$-axis.} of each path is at most a constant $c$ with running time $O^{*}(n^{\frac{3c}{\varepsilon}})$, where $O^{*}(\cdot)$ hides terms polynomial in $c$ and $1/\varepsilon$. In \Cref{sec:ptasesA} we extend this result to a PTAS for \textsc{Max Weight Independent Set} on $B_k$-EPG and $B_k$-VPG graphs with bounded horizontal part, for any fixed $k\geq 1$.

\subsection{Relationships between the main graph classes addressed in our work}

To facilitate the navigation through the main graph classes related to the paper, we depict in \Cref{reldiagram} a full picture of the relationships between these classes and conclude this section by explaining the incomparabilities therein. The terminology we adopt in the following refers to \Cref{reldiagram}. 

The class of stars shows that \textsf{bounded $\mathsf{tw}$ $\nsubseteq$ bounded degree}, \textsf{planar $\nsubseteq$ unit disks}, and \textsf{planar $\nsubseteq$ bounded degree}. The class of complete graphs shows that \textsf{unit disks $\nsubseteq$ fractionally $\mathsf{tw}$-fragile}, \textsf{unit disks $\nsubseteq$ bounded local $\mathsf{tw}$}, \textsf{bounded $\mathsf{tree}\textnormal{-}\alpha$ $\nsubseteq$ bounded local $\mathsf{tw}$}, and \textsf{bounded $\mathsf{tree}\textnormal{-}\alpha$ $\nsubseteq$ fractionally $\mathsf{tw}$-fragile}. The class of $2$-dimensional grids shows that \textsf{planar $\nsubseteq$ bounded $\mathsf{tree}\textnormal{-}\alpha$} and \textsf{unit disks $\nsubseteq$ bounded $\mathsf{tree}\textnormal{-}\alpha$}. \Cref{equivlayeredA} shows that \textsf{proper minor-closed $\nsubseteq$ bounded local $\mathsf{tree}\textnormal{-}\alpha$}. \Cref{degreefragility} shows that \textsf{bounded degree $\nsubseteq$ fractionally $\mathsf{tree}\textnormal{-}\alpha$-fragile}. \Cref{fragileunboundedlocal} shows that \textsf{disks $\nsubseteq$ bounded local $\mathsf{tree}\textnormal{-}\alpha$}. \cite[Theorem~3.4]{GB07} shows that \textsf{bounded number of crossings per edge $\nsubseteq$ proper minor-closed}. We are not aware of any class showing that \textsf{bounded $\mathsf{tw}$ $\nsubseteq$ bounded number of crossings per edge} or that \textsf{bounded $\mathsf{tw}$ $\nsubseteq$ fat objects}.

\begin{figure}
\begin{center}
\scalebox{0.65}{\begin{tikzpicture}[scale=.95]
\node[rectangle,draw] (ftf) at (7,13) {\textsf{fractionally $\mathsf{tw}$-fragile}};
\node[rectangle,draw] (ftaf) at (7,16) {\textsf{fractionally $\mathsf{tree}\textnormal{-}\alpha$-fragile}};
\node[rectangle,draw] (layeredtin) at (-1,12) {\textsf{bounded layered $\mathsf{tree}\textnormal{-}\alpha$}};
\node[rectangle,draw] (layeredtw) at (-1,9) {\textsf{bounded layered $\mathsf{tw}$}};
\node[rectangle,draw] (cross) at (-1,6) {\parbox{3cm}{\centering \textsf{bounded number of}\\ \textsf{crossings per edge}}}; 
\node[rectangle,draw] (minor) at (7,10) {\textsf{proper minor-closed}};
\node[rectangle,draw] (tin) at (3,7.5) {\textsf{bounded $\mathsf{tree}\textnormal{-}\alpha$}};
\node[rectangle,draw] (tw) at (3,2) {\textsf{bounded $\mathsf{tw}$}};
\node[rectangle,draw] (planar) at (7,3) {\textsf{planar}};
\node[rectangle,draw] (bdegree) at (-5,2) {\textsf{bounded degree}};
\node[rectangle,draw] (bltw) at (-5,11) {\textsf{bounded local $\mathsf{tw}$}};
\node[rectangle,draw] (bltin) at (-5,14) {\textsf{bounded local $\mathsf{tree}\textnormal{-}\alpha$}};
\node[rectangle,draw] (unitdisk) at (11,5) {\textsf{unit disks}};
\node[rectangle,draw] (disk) at (11,8) {\textsf{disks}};
\node[rectangle,draw] (fat) at (11,11) {\textsf{fat objects}};

\node (layeredlocal) at (-2,13.2) {\Cref{blayeredblocalA}};
\node (fatfragile) at (11.9,13) {\Cref{treealphafatA}};
\node (layeredfragile) at (1.1,14.5) {\Cref{layeredtofragileA}};
\node (layeredtwlocaltw) at (-2.6,10.3) {\cite{DMW17}};
\node (crossingslayeredtw) at (-1.5,7.5) {\cite{DEW17}};
\node (unitlayered) at (4,10) {\Cref{unitdisklayeredA}};
\node (layeredtwfragile) at (3.5,12) {\Cref{layeredtwfragile}};
\node (minorclosedfragile) at (7.5,11.5) {\cite{DDO04}};
\node (planardisks) at (8.6,4) {(see, e.g., \cite{GS93})};

\draw[->](layeredtin) ..controls (3,15)..  (ftaf); 
\draw[->](layeredtw) -- (ftf);  
\draw[->](tin) -- (layeredtin); 
\draw[->](unitdisk) -- (layeredtin);  
\draw[->](fat) ..controls (10.5,14.5).. (ftaf); 
\draw[->](layeredtw) -- (layeredtin); 
\draw[->](ftf) -- (ftaf); 
\draw[->](minor) -- (ftf); 
\draw[->](cross) -- (layeredtw);
\draw[->](planar) ..controls (-0.5,5).. (cross); 
\draw[->](unitdisk) -- (disk);
\draw[->](disk) -- (fat);
\draw[->](bdegree) -- (bltw);
\draw[->](bltw) -- (bltin);
\draw[->](tw) -- (tin);
\draw[->,dashed,red](tw) ..controls (-0.5,3.5).. (cross); 
\draw[->,dashed,red](tw) -- (fat); 
\draw[->](tw) ..controls (2.8,6).. (layeredtw); 
\draw[->](planar) -- (minor); 
\draw[->](planar) ..controls (8,7).. (disk);
\draw[->](layeredtw) -- (bltw);
\draw[->](layeredtin) -- (bltin);
\draw[->](tw) ..controls (4,5).. (minor);
\end{tikzpicture}}
\end{center}
\caption{Relationships between the main graph classes related to the paper, where an arrow represents class inclusion or implication between class properties. The following shorthands are adopted. $\mathsf{tw}$ and $\mathsf{tree}\textnormal{-}\alpha$ are shorthands for treewidth and tree-independence-number, respectively. \textsf{bounded number of crossings per edge} stands for the class of graphs embeddable on a surface of bounded genus with a bounded number of crossings per edge. \textsf{unit disks}, \textsf{disks}, and \textsf{fat objects} are shorthands for the class of intersection graphs of a collection of unit disks in the plane, disks in the plane, and $c$-fat objects in some $d$-dimensional space, respectively. We reference only the inclusions or implications not directly following from the definition. A dashed red arrow indicates that determining whether the corresponding inclusion holds is, to the best of our knowledge, open.}\label{reldiagram}
\end{figure}
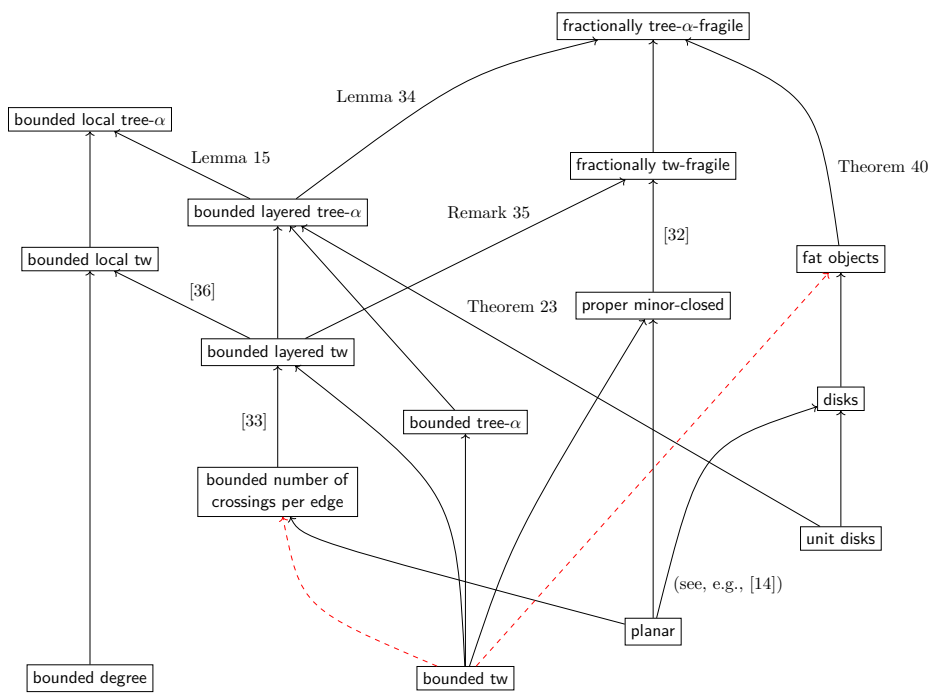


\section{Preliminaries}
\label{sec:prelimA}

We let $\mathbb{N} = \{1, 2, \ldots\}$ and $\mathbb{N}_0 = \mathbb{N} \cup \{0\}$ and, for each $n \in \mathbb{N}$, we let $[n] = \{1, \ldots, n\}$. We consider only finite simple graphs. If $G'$ is a subgraph of $G$ and $G'$ contains all the edges of $G$ with both endpoints in $V(G')$, then $G'$ is an \textit{induced subgraph} of $G$ and we write $G' = G[V(G')]$. For a vertex $v \in V(G)$ and $r \in \mathbb{N}_0$, the $r$-\textit{closed neighborhood} $N^{r}_{G}[v]$ is the set of vertices at distance at most $r$ from $v$ in $G$. The \textit{degree} $d_{G}(v)$ of a vertex $v \in V(G)$ is the number of edges incident to $v$ in $G$. The \textit{maximum degree} $\Delta(G)$ of $G$ is the quantity $\max\left\{d_{G}(v): v \in V\right\}$. Given a graph $G = (V, E)$ and $V' \subseteq V$, the operation of \textit{deleting the set of vertices $V'$} from $G$ results in the graph $G - V' = G[V\setminus V']$. Given a family $\mathcal{H}$ of graphs, a graph is \textit{$\mathcal{H}$-free} if it does not contain any induced subgraph isomorphic to a graph in $\mathcal{H}$. The complete bipartite graph (or biclique) with parts of sizes $r$ and $s$ is denoted by $K_{r,s}$. An \textit{independent set} of a graph is a set of pairwise non-adjacent vertices. The maximum size of an independent set of $G$ is denoted by $\alpha(G)$. A \textit{clique} of a graph is a set of pairwise adjacent vertices. The maximum size of a clique of $G$ is denoted by $\omega(G)$. For each $a, b \in \mathbb{N}$, we denote by $R(a,b)$ the smallest integer $n$ such that every graph on at least $n$ vertices contains either a clique of size $a$ or an independent set of size $b$ (such an $n$ exists by Ramsey's theorem).

\textbf{Intersection graphs of unit disks and unit-width rectangles.} We now explain how the geometric realizations of these intersection graphs are encoded. First, a collection of unit disks is a family of closed disks in $\mathbb{R}^2$ with common radius $c \in \mathbb{R}$, whereas a collection of unit-width rectangles is a family of axis-aligned closed rectangles in $\mathbb{R}^2$ with common width $c \in \mathbb{R}$. A collection of unit disks is encoded by a collection of points in $\mathbb{R}^2$ representing the centers of the disks. Each rectangle in a collection of unit-width rectangles is encoded by the four points in $\mathbb{R}^2$ representing its four vertices. In general, and unless otherwise stated, when we refer to a rectangle we mean an axis-aligned closed rectangle in $\mathbb{R}^2$. 
 
\textbf{VPG and EPG graphs.} Given a rectangular grid $\mathcal{G}$, its horizontal lines are referred to as \textit{rows} and its vertical lines as \textit{columns}. For a VPG (EPG) graph $G$, the pair $\mathcal{R} = (\mathcal{G},\mathcal{P})$ is a \textit{VPG representation} (\textit{EPG representation}) of $G$. More generally, a \emph{grid representation} of a graph $G$ is a triple $\mathcal{R} = (\mathcal{G},\mathcal{P},x)$ where $x \in \{e,v\}$, such that $(\mathcal{G},\mathcal{P})$ is an EPG representation of $G$ if $x = e$, and $(\mathcal{G},\mathcal{P})$ is a VPG representation of $G$ if $x = v$. Note that, irrespective of whether $x=e$ (i.e., $G$ is an EPG graph) or $x=v$ (i.e., $G$ is a VPG graph), if two vertices $u,v \in V(G)$ are adjacent in $G$ then $P_u$ and $P_v$ share at least one grid-point. A \textit{bend-point} of a path $P\in\mathcal{P}$ is a grid-point corresponding to a bend of $P$ and a \textit{segment} of $P$ is either a vertical or horizontal line segment in the polygonal curve constituting $P$. Paths in $\mathcal{P}$ are encoded as follows. For each $P \in \mathcal{P}$, we have one sequence $s(P)$ of points in $\mathbb{R}^2$: $s(P) = (x_1,y_1), (x_2,y_2), \dots, (x_{\ell_P},y_{\ell_P})$ consists of the endpoints $(x_1,y_1)$ and $(x_{\ell_P},y_{\ell_P})$ of $P$ and all the bend-points of $P$ in their order of appearance when traversing $P$ from $(x_1,y_1)$ to $(x_{\ell_P},y_{\ell_P})$. If each path in $\mathcal{P}$ has number of bends polynomial in $|V(G)|$, then the size of this data structure is polynomial in $|V(G)|$. Given $s(P)$, we can easily determine the horizontal part $h(P)$ of the path $P$ (i.e., the projection of $P$ onto the $x$-axis). Note that our results for VPG and EPG graphs (\Cref{layeredVPGA,indepPTASA} and \Cref{pathA}), although stated for constant number of bends, still hold for polynomial (in $|V(G)|$) number of bends, with a worse polynomial running time.  

\textbf{PTAS.} A PTAS for a maximization problem is an algorithm which takes an instance $I$ of the problem and a parameter $\varepsilon > 0$ and produces a solution within a factor $1 - \varepsilon$ of the optimal in time $n^{O(f(1/\varepsilon))}$. A PTAS with running time $f(1/\varepsilon)\cdot n^{O(1)}$ is called an efficient PTAS (EPTAS for short).


\section{Comparing different notions of fatness}\label{fatcomp}

In this section we introduce our notion of fatness, called $c$-fatness, and compare it with alternative notions of fatness from the literature. We then show that all these notions are equivalent when restricting to convex objects.

We first require some additional definitions. Let $d \geq 2$ be an arbitrary but fixed integer. An \textit{object} in $\mathbb{R}^d$ is a path-connected compact set $O \subset \mathbb{R}^d$. The \textit{size} of an object $O$ in $\mathbb{R}^d$, denoted $s(O)$, is the side length of its smallest enclosing axis-aligned hypercube. Unless otherwise stated, all boxes considered in the paper are axis-aligned. 

Chan~\cite{Cha03} introduced the following definition of fatness: A collection of objects in $\mathbb{R}^d$ is fat if, for any $r$ and size-$r$ box $R$, we can choose a constant number $c$ of points in $\mathbb{R}^d$ such that every object that intersects $R$ and has size at least $r$ contains at least one of the chosen points. In particular, if a collection of objects is fat according to this definition, every size-$r$ box intersects at most $c$ pairwise disjoint objects of size at least $r$ from the collection. Chan also stated that collections of balls and collections of boxes with bounded aspect ratios are fat (recall that the aspect ratio of a box is the ratio of its largest side length over its smallest side length). We slightly generalize this fatness definition as follows.

\begin{definition}\label{defc-fat} 
Let $c\in \mathbb{R}$ be a constant. A collection $\mathcal{O}$ of objects in $\mathbb{R}^d$ is $c$-$\textit{fat}$ if, for any $r\in \mathbb{R}$ and any closed box $R$ of side length $r$, there exist at most $c$ pairwise non-intersecting objects from $\mathcal{O}$ of size at least $r$ and which intersect $R$.
\end{definition}

Loosely speaking, a collection of objects is fat according to the previous definition if it satisfies a sort of ``low-density property''. In fact, our notion of fatness captures that of low density, defined by Har-Peled and Quanrud~\cite{HP17} as follows. Given a constant $\rho$, a collection $\mathcal{O}$ of objects in $\mathbb{R}^d$ has \textit{density $\rho$} if any object $R$ (not necessarily in $\mathcal{O}$) intersects at most $\rho$ objects from $\mathcal{O}$ with diameter at least that of $R$. It should be noticed that, in some of the their arguments, Har-Peled and Quanrud make implicit use of the notion of $\rho$-fatness (see \cite[Lemma~3.6]{HP17}). 
As mentioned, collections of balls and collections of boxes with bounded aspect ratios in $\mathbb{R}^d$ are $c$-fat, for some $c\geq 1$; we will formally show this below. As another interesting example, we will also show that certain collections of annuli in $\mathbb{R}^d$ are $c$-fat.  

\begin{remark} When working with a $c$-fat collection of objects, we assume that some reasonable operations can be done in constant time: determining center, size and diameter of an object, as well as its projection onto one of the axes, deciding if two objects intersect, and constructing the geometric realization of the collection. 
\end{remark}

Several definitions of fatness have been proposed in the literature, essentially all of which are equivalent for convex objects. In the rest of this section, we compare the notion of $c$-fatness defined above with three among the most general (in the sense that they apply to arbitrary objects in arbitrary dimensions) notions of fatness. We refer the reader to \cite{vdS94} for a general overview on fatness. 

Let $d \geq 2$ be an arbitrary but fixed integer. For a measurable set $A \subseteq \mathbb{R}^d$, we denote by $\vol(A)$ the Lebesgue measure of $A$. The following definition is due to van der Stappen et al.~\cite{SHO93}. Given $k > 0$, an object $O \subseteq \mathbb{R}^d$ is $k$-\textit{locally fat}\footnote{Notice that, in \cite{SHO93}, such a notion is referred to as $k$-fatness.} if, for each closed $d$-dimensional ball $B$ with center in $O$ and whose boundary intersects $O$ (or, equivalently, not properly containing $O$), $\mathrm{vol}(B) \leq k^d \cdot \mathrm{vol}(O \cap B)$. It can be seen that there are no $k$-locally fat objects for $k < 2$ and balls are exactly the $2$-locally fat objects \cite{SHO93}. The terminology is justified by the fact that a constant portion of the proximity of every point of the object must be covered by the object and so no object with infinitesimally thin protuberances is locally fat.

The following two fatness definitions have a more global nature, with the first one being arguably the ``standard'' definition. Given $k \geq 1$, an object $O \subseteq \mathbb{R}^d$ is $k$-\textit{globally fat} if there exist two $d$-dimensional balls $B_{\mbox{\tiny{in}}}$ and $B_{\mbox{\tiny{out}}}$ with radius $R_{\mbox{\tiny{in}}}$ and $R_{\mbox{\tiny{out}}}$, respectively, such that $B_{\mbox{\tiny{in}}} \subseteq O \subseteq B_{\mbox{\tiny{out}}}$ and $R_{\mbox{\tiny{out}}} \leq k \cdot R_{\mbox{\tiny{in}}}$. The second definition is again due to van der Stappen et al.~\cite{SHO93}. Given $k \geq 1$, an object $O \subseteq \mathbb{R}^d$ is $k$-\textit{thick} if, denoting by $B_{O}$ the minimal enclosing ball of $O$ (i.e., the $d$-dimensional ball $B_{O}$ of smallest volume such that $O \subseteq B_{O}$), we have that $\mathrm{vol}(B_O) \leq k^d \cdot \mathrm{vol}(O)$. Clearly, balls are exactly the $1$-globally fat objects and are $1$-thick.   

Observe that a $k$-locally fat (or $k$-globally fat, or $k$-thick) object is $k'$-locally fat (or $k'$-globally fat, or $k'$-thick) for any $k' \geq k$. Moreover, in all these fatness definitions, the value of $k$ can be seen as a qualitative measure of fatness: the smaller the value of $k$, the fatter the object. The three definitions above naturally extend to collections of objects as follows: Given $k$, we say that a collection of objects in $\mathbb{R}^d$ is \textit{$k$-locally fat} (or \textit{$k$-globally fat}, or \textit{$k$-thick}) if every object in the collection is $k$-locally fat (or $k$-globally fat, or $k$-thick). 

It is easy to see that, in general, global fatness is a stronger notion than thickness, as we show next for completeness. Note that, in this section, we make repeated use of the following well-known fact: For any $d$-dimensional ball $B$ with radius $r$, $\mathrm{vol}(B) = C_d \cdot r^d$, for a constant $C_d$ depending on $d$. 

\begin{lemma}[Folklore]\label{folk} Every $k$-globally fat object in $\mathbb{R}^d$ is $k$-thick. Moreover, there exists a $k$-thick collection of objects in $\mathbb{R}^2$ which is not $k'$-globally fat, for any $k' \geq 1$.
\end{lemma}

\begin{proof} Let $O$ be a $k$-globally fat object in $\mathbb{R}^d$. By assumption, there exist $k \geq 1$ and two $d$-dimensional balls $B_{\mbox{\tiny{in}}}$ and $B_{\mbox{\tiny{out}}}$ with radius $R_{\mbox{\tiny{in}}}$ and $R_{\mbox{\tiny{out}}}$, respectively, such that $B_{\mbox{\tiny{in}}} \subseteq O \subseteq B_{\mbox{\tiny{out}}}$ and $R_{\mbox{\tiny{out}}} \leq k \cdot R_{\mbox{\tiny{in}}}$. Therefore, $\mathrm{vol}(B_{\mbox{\tiny{out}}}) \leq k^d \cdot \mathrm{vol}(B_{\mbox{\tiny{in}}})$. Denoting by $B_{O}$ the minimal enclosing ball of $O$, we then obtain that $\mathrm{vol}(B_O) \leq \mathrm{vol}(B_{\mbox{\tiny{out}}}) \leq k^d \cdot \mathrm{vol}(B_{\mbox{\tiny{in}}}) \leq k^d \cdot \mathrm{vol}(O)$, and so $O$ is $k$-thick. As for the second statement, simply consider the collection of combs depicted in \Cref{Fig:comb}.
\end{proof}

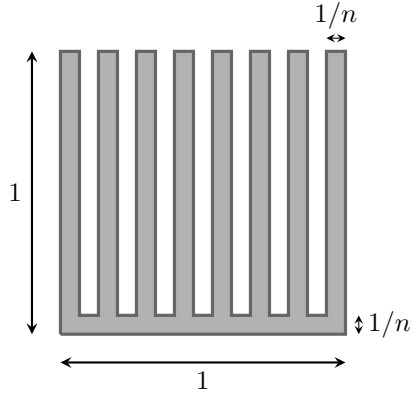
\begin{figure}
\centering
\begin{tikzpicture}[scale=2.5]
\draw[very thick,fill=gray,opacity=.6] (0,0) -- (0,1.5) -- (.1,1.5) -- (.1,.1) -- (.2,.1) -- (.2,1.5) -- (.3,1.5) -- (.3,.1) -- (.4,.1) -- (.4,1.5) -- (.5,1.5) -- (.5,.1) -- (.6,.1) -- (.6,1.5) -- (.7,1.5) -- (.7,.1)-- (.8,.1) -- (.8,1.5) -- (.9,1.5) -- (.9,.1) -- (1,.1) -- (1,1.5) -- (1.1,1.5) -- (1.1,.1) -- (1.2,.1) -- (1.2,1.5) -- (1.3,1.5) -- (1.3,.1) -- (1.4,.1) -- (1.4,1.5) -- (1.5,1.5) -- (1.5,0) -- (0,0);

\draw[thick,<->,>=stealth] (-.15,0) -- (-.15,1.5) node[midway,left] {1};
\draw[thick,<->,>=stealth] (0,-.15) -- (1.5,-.15) node[midway,below] {1};
\draw[<->,>=stealth] (1.4,1.57) -- (1.5,1.57) node[midway,above] {$1/n$};
\draw[<->,>=stealth] (1.57,0) -- (1.57,.1) node[midway,right] {$1/n$};
\end{tikzpicture}
\caption{A comb $O_n$ with $(n+1)/2$ teeth each of width $1/n$, for some odd $n \in \mathbb{N}$. It is easy to see that $O_n$ is $\sqrt{\pi}$-thick but the collection $\{O_n : n \in \mathbb{N}\}$ is not $k$-globally fat, for any $k \geq 1$.}\label{Fig:comb}
\end{figure}

We now observe how our notion of $c$-fatness relates to local fatness, global fatness and thickness. We first show that $c$-fatness does not imply thickness and hence global fatness either. An \textit{annulus} is a closed region in $\mathbb{R}^d$ bounded by two concentric $d$-dimensional balls. The bigger ball is called the \textit{outer ball}, whereas the other is called the \textit{inner ball}.

\begin{lemma}\label{annuli} Fix $R>0$. The collection of annuli in $\mathbb{R}^d$ with outer balls of radius $R$ is $c$-fat, for some $c \geq 1$, but not $k$-thick, for any $k \geq 1$.  
\end{lemma}

\begin{proof} The fact that the collection of annuli with outer balls of radius $R$ is $c$-fat, for some $c \geq 1$, follows from the fact that two such annuli are disjoint if and only if the corresponding outer balls are and from \cite[Theorem~2.9]{vdS94} (see also \cite[Theorem~2.2]{OS96}). Consider now the remaining statement. Fix $k \geq 1$ and take $n > k$. Consider an annulus $A$ with outer ball of radius $R$ and inner ball of radius $(1-1/n^d)^{1/d}R$. The minimal enclosing ball of $A$ has volume $C_d\cdot R^d$, whereas $\vol(A) = C_d \cdot \frac{R^d}{n^d}$. The choice of $n$ implies that $A$ is not $k$-thick.   
\end{proof}

We now show that global fatness (and hence thickness) does not imply $c$-fatness. The following result comes in handy. Recall that an \textit{outerstring graph} is the intersection graph of a set of curves in $\mathbb{R}^2$ that lie inside a disk such that each curve intersects the boundary of the disk in one of its endpoints.

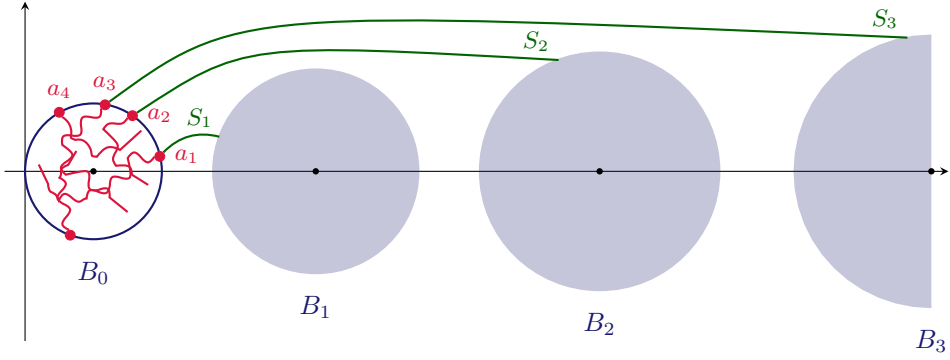
\begin{figure}
\centering
\begin{tikzpicture}[scale=0.9]
\draw[->,>=stealth] (0,-2.5) -- (0,2.5);
\draw[->,>=stealth] (-.3,0) -- (13.5,0);

\draw[thick,MidnightBlue] (1,0) circle (1cm);
\draw[thick,MidnightBlue,fill=MidnightBlue,opacity=.25] (4.25,0) circle (1.5cm);
\draw[thick,MidnightBlue,fill=MidnightBlue,opacity=.25] (8.4,0) circle (1.75cm);
\draw[thick,MidnightBlue,domain=90:270,fill=MidnightBlue,opacity=.25] plot ({13.25 + 2*cos(\x)},{2*sin(\x)});
\node[cir] at (1,0) {};
\node[draw=none] at (1,-1.5) {{\color{MidnightBlue} $B_0$}};
\node[cir] at (4.25,0) {};
\node[draw=none] at (4.25,-2) {{\color{MidnightBlue} $B_1$}};
\node[cir] at (8.4,0) {};
\node[draw=none] at (8.4,-2.25) {{\color{MidnightBlue} $B_2$}};
\node[cir] at (13.25,0) {};
\node[draw=none] at (13.25,-2.5) {{\color{MidnightBlue} $B_3$}};

\node[circr,label=right:{{\color{Crimson} \small $a_1$}}] (a1) at (1.97,.22) {};
\node[circr,label=right:{{\color{Crimson} \small $a_2$}}] (a2) at (1.57,.82) {}; 
\node[circr,label=above:{{\color{Crimson} \small $a_3$}}] (a3) at (1.17,.98) {};
\node[circr,label=above:{{\color{Crimson} \small $a_4$}}] (a4) at (.5,.87) {};
\node[circr] (an) at (0.66,-.94) {};

\draw[thick,DarkGreen] (a1) edge[bend left] (2.84,.51); 
\node[draw=none,DarkGreen,above left] at (2.88,.51) {\small $S_1$};

\draw[thick,DarkGreen] (a2) ..controls (3.2,1.9).. (7.8,1.64);
\node[draw=none,DarkGreen,above left] at (7.8,1.64) {\small $S_2$};

\draw[thick,DarkGreen] (a3) ..controls (3.1,2.4).. (12.9,1.97);
\node[draw=none,DarkGreen,above left] at (12.9,1.97) {\small $S_3$};

\draw[Crimson,thick] decorate [decoration={coil,aspect=0,amplitude=2pt}] {(a1) .. controls (1.6,.2) and (.7,-1).. (.2,.1)};
\draw[Crimson,thick] decorate [decoration={coil,aspect=0,amplitude=1pt}] {(a2) .. controls (1.1,.6) and (1.4,0).. (1.8,-.2)};
\draw[Crimson,thick] decorate [decoration={coil,aspect=0,amplitude=2pt}] {(a3) ..controls (.4,.3).. (1.5,-.6)};
\draw[Crimson,thick] decorate [decoration={coil,aspect=1,amplitude=1pt}] {(a4) .. controls (1,0) and (1.4,.4).. (1.7,.6)};
\draw[Crimson,thick] decorate [decoration={coil,aspect=0,amplitude=2pt}] {(an) .. controls (.5,-.2).. (.75,.3)};
\end{tikzpicture}
\caption{Realizing an outerstring graph as the intersection graph of a globally fat collection of objects in $\mathbb{R}^2$.}\label{Fig:outerstring}
\end{figure}

\begin{lemma}\label{outerfat} Let $k > 1$. Every outerstring graph can be realized as the intersection graph of a $k$-globally fat collection of objects in $\mathbb{R}^d$, for any $d \geq 2$.
\end{lemma}

\begin{proof} We show the result for $d=2$ and leave the easy extension to the reader. Fix an arbitrary $1 < k < \sqrt{2}$. Let $G$ be an outerstring graph on $n$ vertices. By possibly rescaling, we may assume that its geometric realization consists of the disk $B_0$ centered at $(1/2,0)$ and with radius $1/2$. Fix an arbitrary counterclockwise order $a_1,\ldots,a_n$ of the endpoints of the curves on the boundary of $B_0$. For each positive integer $i$, let $r_i = (\frac{1}{k-1})^{2i}$ and $c_i = (\frac{1}{k-1})^{2i-1} + (\frac{1}{k-1})^{2i}$, and let $B_i$ be the disk with radius $r_i$ centered at $(c_i,0)$. Hence, $B_1,\ldots,B_n$ are disks of increasing radius arranged along the positive $x$-axis. It is easy to see that $B_1$ is disjoint from $B_0$. Observe now that, for $i \geq 1$, any point in $B_i$ has $x$-coordinate at most $c_i + r_i = (\frac{1}{k-1})^{2i-1} + (\frac{1}{k-1})^{2i} + (\frac{1}{k-1})^{2i}$, whereas any point in $B_{i+1}$ has $x$-coordinate at least $c_{i+1} - r_{i+1} = (\frac{1}{k-1})^{2i+1}$. Since $1 < k < \sqrt{2}$, we have that $(\frac{1}{k-1})^{2i+1} > (\frac{1}{k-1})^{2i-1} + (\frac{1}{k-1})^{2i} + (\frac{1}{k-1})^{2i}$, and so $B_{i+1}$ lies completely to the right of $B_i$ for each $i \geq 0$, hence the disks $B_0, B_1, \ldots,B_n$ are pairwise disjoint. It is easy to see that, for every $i \in \{1,\ldots,n\}$, we can connect $a_i$ to $B_i$ using a thin string $S_i$ in such a way that $S_i$ does not intersect any disk other than $B_0$ and $B_i$, $S_i$ is contained in the disk with center $(c_i,0)$ and radius $c_i$, and the strings $S_1,\ldots,S_n$ are pairwise disjoint (see \Cref{Fig:outerstring}). 

Let $X_1,\ldots, X_n$ be the geometric objects obtained in this way. Clearly, $G$ is isomorphic to the intersection graph of $\{X_1,\ldots, X_n\}$. By construction, $X_i$ is contained in a disk of radius $c_i$ and contains a disk of radius $r_i$. Since 
\begin{equation*}
\frac{c_i}{r_i}
 = \frac{(\frac{1}{k-1})^{2i-1} + (\frac{1}{k-1})^{2i}}{(\frac{1}{k-1})^{2i}} = \frac{1 + \frac{1}{k-1}}{\frac{1}{k-1}} = k,
\end{equation*}
the collection $\{X_1,\ldots, X_n\}$ is $k$-globally fat. Since this holds for any $1 < k < \sqrt{2}$, we conclude by recalling that a $k$-globally fat collection is $k'$-globally fat for any $k' \geq k$.
\end{proof}

In order to show the following consequence of \Cref{outerfat} we make use of two results, \Cref{bicliqueA} and \Cref{treealphafatA}, which will be proved in \Cref{sec:fragilityA} and \Cref{fatA}, respectively. 

\begin{corollary}\label{cfatincomp} The class of intersection graphs of $k$-globally fat objects in $\mathbb{R}^d$ is fractionally $\tin$-fragile if and only if $k=1$. In particular, there exists a $k$-globally fat collection of objects in $\mathbb{R}^d$, for some $k > 1$, which is not $c$-fat, for any $c \geq 1$. 
\end{corollary}

\begin{proof} The class of intersection graphs of $1$-globally fat objects in $\mathbb{R}^d$, which coincides with the class of intersection graphs of balls in $\mathbb{R}^d$, is fractionally $\tin$-fragile thanks to \Cref{treealphafatA}. It is easy to see that the class of complete bipartite graphs is contained in that of outerstring graphs. By \Cref{outerfat}, the latter is contained in the class of intersection graphs of $k$-globally fat objects in $\mathbb{R}^d$, for any $k > 1$. It then suffices to observe that the class of complete bipartite graphs is not fractionally $\tin$-fragile (by \Cref{bicliqueA}) and that, for any $c \geq 1$, the class of intersection graphs of $c$-fat collections of objects in $\mathbb{R}^d$ is fractionally $\tin$-fragile (\Cref{treealphafatA}).  
\end{proof}

Finally, we observe that local fatness is a stronger notion than $c$-fatness. 

\begin{lemma}\label{cfatA} If $\mathcal{O}$ is a $k$-locally fat collection of objects in $\mathbb{R}^d$, then $\mathcal{O}$ is $(2k)^d$-fat. Moreover, for any fixed $R > 0$, the collection of annuli in $\mathbb{R}^d$ with outer balls of radius $R$ is $c$-fat, for some $c \geq 1$, but not $k$-locally fat, for any $k \geq 1$.  
\end{lemma}

\begin{proof} The first statement immediately follows from \cite[Theorem~2.9]{vdS94} (see also \cite[Theorem~2.2]{OS96}). As for the second, it follows from \Cref{annuli} and the observation that every $k$-locally fat object is $k$-thick, which is left as an easy exercise.  
\end{proof}

We note some consequences of \Cref{cfatA}. Since balls in any dimension are $2$-locally fat, a collection of balls in $\mathbb{R}^d$ is $4^d$-fat. Moreover, since the length of a main diagonal of a $d$-dimensional box of side length $l$ is $l\sqrt{d}$, a size-$r$ box with aspect ratio at most $t$ has volume at least $(\frac{r}{t\sqrt{d}})^d$ and so is $td\sqrt[d]{C_d}$-locally fat (see, e.g., \cite[Theorem~2.7]{vdS94}). This implies that a collection of boxes (not necessarily axis-aligned) in $\mathbb{R}^d$ with aspect ratios at most $t$ is $C_d(2td)^d$-fat. 

We conclude this section by considering the case of collections of convex objects, where all definitions of fatness introduced above turn out to be equivalent. It is known that, when restricting to convex objects in $\mathbb{R}^d$, local fatness implies global fatness \cite[Lemma~4.6]{OS96}, and that local fatness and thickness are equivalent \cite[Theorem~2.5]{SHO93}. In order to add $c$-fatness to the picture, it is enough to show that, for collections of convex objects, $c$-fatness implies global fatness. In the next two results, we assume that the collections of convex objects are closed under rotations and translations.

\begin{lemma}\label{fatglobalfat} If $\mathcal{O}$ is a $c$-fat collection of convex objects in $\mathbb{R}^d$, then $\mathcal{O}$ is $cd\sqrt{d}$-globally fat.
\end{lemma}

\begin{proof} Let $\{e_1,\ldots,e_{d}\}$ be the standard basis of $\mathbb{R}^d$ and let $O \in \mathcal{O}$ be of size $r$. We first claim that, for each $i \in \{1,\ldots, d\}$, there exist points $p_i,q_i \in O$ such that $\vv{p_iq_i} = \frac{2r}{c} e_i$. Suppose, to the contrary, that there exists $i\in \{1,\ldots,d\}$ such that, for every pair of points $p_i,q_i \in O$, we have that $\vv{p_iq_i} \neq \frac{2r}{c}e_i$. Since $O$ is convex, it is not difficult to see that this implies that $O$ is contained in a box of width less than $\frac{2r}{c}$ along the $i$-th axis. But then we can find a subcollection of at least $c+1$ pairwise non-intersecting objects of size $r$ which intersect a fixed box of side length $r$, as depicted in \Cref{Fig:cfat}, contradicting the assumption that $\mathcal{O}$ is $c$-fat.
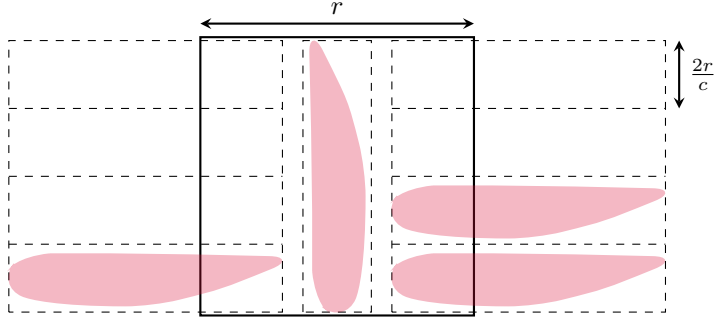
\begin{figure}
\centering
\begin{tikzpicture}[scale=0.9]
\draw[dashed] (0,0) rectangle (4,4); 
\draw[dashed] (0,1) -- (4,1);
\draw[dashed] (0,2) -- (4,2);
\draw[dashed] (0,3) -- (4,3);
\draw[thick] (2.8,-.05) rectangle (6.8,4.05);
\draw[dashed] (5.6,0) rectangle (9.6,4);
\draw[dashed] (5.6,1) -- (9.6,1);
\draw[dashed] (5.6,2) -- (9.6,2);
\draw[dashed] (5.6,3) -- (9.6,3);
\draw[thick,<->,>=stealth] (2.8,4.25) -- (6.8,4.25) node[midway,above] {$r$};
\draw[thick,<->,>=stealth] (9.8,3) -- (9.8,4) node[midway,right] {$\frac{2r}{c}$};


\draw[thick,fill=Crimson,opacity=.3,Crimson] [rounded corners=10pt] (.6,.85) .. controls (2,.85).. (4.2,.8) ..controls (3.1,.35)..  (2.1,.1) .. controls (.7,.1) and (.3,.2) .. (-.16,.4) .. controls (0,.6) and (.12,.8) .. (.6,.85);
\draw[thick,fill=Crimson,opacity=.3,Crimson] [rounded corners=10pt] (6.2,.85) .. controls (7.6,.85).. (9.8,.8) ..controls (8.7,.35)..  (7.7,.1) .. controls (6.3,.1) and (5.9,.2) .. (5.44,.4) .. controls (5.6,.6) and (5.72,.8) .. (6.2,.85);
\draw[thick,fill=Crimson,opacity=.3,Crimson] [rounded corners=10pt] (6.2,1.85) .. controls (7.6,1.85).. (9.8,1.8) ..controls (8.7,1.35)..  (7.7,1.1) .. controls (6.3,1.1) and (5.9,1.2) .. (5.44,1.4) .. controls (5.6,1.6) and (5.72,1.8) .. (6.2,1.85);
\draw[dashed] (4.3,0) rectangle (5.3,4);
\draw[thick,fill=Crimson,opacity=.3,Crimson] [rounded corners=10pt] (4.45,.6) .. controls (4.45,2).. (4.4,4.2) ..controls (4.95,3.1)..  (5.2,2.1) .. controls (5.2,.7) and (5.1,.3) .. (4.9,-.16) .. controls (4.7,0) and (4.5,.12) .. (4.45,.6);
\end{tikzpicture}
\caption{A construction used in the proof of \Cref{fatglobalfat}.}\label{Fig:cfat}
\end{figure}

We now show that there exists a point $t\in O$ such that $t+\frac{2r}{cd}e_i \in O$ for every $i$. This immediately follows from a general claim we show next.

\begin{claim}\label{orthogonalclaim} Let $O$ be a convex object in $\mathbb{R}^d$ and let $\varepsilon >0$. Suppose that there exists a set of vectors $\{e_1,\ldots,e_{k}\}$ such that, for each $i \in \{1,\ldots, k\}$, there exist points $p_i,q_i \in O$ with $\vv{p_iq_i} = \varepsilon e_i$. Then there exists a point $t\in O$ such that $t+\frac{\varepsilon}{k}e_i \in O$ for every $i\in \{1,\ldots,k\}$.
\end{claim}

\begin{claimproof}[Proof of \Cref{orthogonalclaim}] We proceed by induction on $k$. The base case $k=1$ trivially holds. Let now $k > 1$ and suppose that there exists a set of vectors $\{e_1,\ldots,e_{k}\}$ such that, for each $i \in \{1,\ldots, k\}$, there exist points $p_i,q_i \in O$ with $\vv{p_iq_i} = \varepsilon e_i$. By the induction hypothesis, we can find $s\in O$ such that $s+\frac{\varepsilon}{k-1}e_i \in O$ for every $i\in \{1,\ldots,k-1\}$. Consider now $t = \frac{k-1}{k}s + \frac{1}{k}p_k$. Since $O$ is convex and $s, p_k \in O$, we have that $t \in O$. Moreover, since $s+\frac{\varepsilon}{k-1}e_i \in O$ for each $i \in \{1,\ldots,k-1\}$, we have that $t+\frac{\varepsilon}{k}e_i = \frac{k-1}{k}(s+\frac{\varepsilon}{k-1}e_i) + \frac{1}{k}p_k \in O$ for each $i \in \{1,\ldots,k-1\}$. Similarly, since $p_k+\varepsilon e_k = q_k \in O$, it follows that $t+\frac{\varepsilon}{k}e_k = \frac{k-1}{k}s + \frac{1}{k}(p_k + \varepsilon e_k) \in O$. 
\end{claimproof}

The above implies that $O$ contains a hypercube of side length $\frac{r}{cd}$ and so a ball of radius $\frac{r}{2cd}$ as well. On the other hand, since $O$ has size $r$, the side length of its smallest enclosing hypercube is at most $r$ and so the radius of its smallest enclosing ball is at most $r\sqrt{d}/2$. This implies that $O$ is $cd\sqrt{d}$-globally fat.
\end{proof}

\begin{proposition}\label{equivconvex} Let $\mathcal{O}$ be a collection of convex objects in $\mathbb{R}^d$. The following are equivalent:
\begin{enumerate}
\item $\mathcal{O}$ is $k_0$-fat for some $k_0 \geq 1$;
\item $\mathcal{O}$ is $k_1$-globally fat for some $k_1 \geq 1$; 
\item $\mathcal{O}$ is $k_2$-thick for some $k_2 \geq 1$;
\item $\mathcal{O}$ is $k_3$-locally fat for some $k_3 \geq 2$. 
\end{enumerate}  
\end{proposition}

\begin{proof} The implication (1) $\Rightarrow$ (2) follows from \Cref{fatglobalfat}. The implication (2) $\Rightarrow$ (3) follows from \Cref{folk}. The implication (3) $\Rightarrow$ (4) follows from \cite[Theorem~2.5]{SHO93}. Finally, the implication (4) $\Rightarrow$ (1) follows from \Cref{cfatA}.
\end{proof}


\section{Layered and local tree-independence number}\label{sec:layeredA}

In this section we initiate a study of the notion of layered tree-independence number. In order to do so, we first have to recall the key definitions of tree-independence number and layering. A \textit{tree decomposition} of a graph $G$ is a pair $\mathcal{T} = (T, \{X_t\}_{t\in V(T)})$, where $T$ is a tree whose every node $t$ is assigned a vertex subset $X_t \subseteq V(G)$, called \textit{bag}, such that the following conditions are satisfied: 
\begin{description}
\item[(T1)] Every vertex of $G$ belongs to at least one bag; 
\item[(T2)] For every $uv \in E(G)$, there exists a bag containing both $u$ and $v$; 
\item[(T3)] For every $u \in V(G)$, the subgraph $T_u$ of $T$ induced by $\{t \in V(T) : u \in X_t\}$ is connected. 
\end{description}
The \textit{width} of $\mathcal{T} = (T, \{X_t\}_{t\in V(T)})$ is the maximum value of $|X_t| - 1$ over all $t \in V(T)$. The \textit{treewidth} of a graph $G$, denoted $\tw(G)$, is the minimum width of a tree decomposition of $G$. The \textit{independence number} of $\mathcal{T}$, denoted $\alpha(\mathcal{T}$), is the quantity $\max_{t\in V(T)} \alpha(G[X_t])$. The \textit{tree-independence number} of a graph $G$, denoted $\tin(G)$, is the minimum independence number of a tree decomposition of $G$. Clearly, $\tin(G) \leq \tw(G)+1$, for any graph $G$. On the other hand, tree-independence number is a width parameter more powerful than treewidth, as there exist classes with bounded tree-independence number and unbounded treewidth (for example, chordal graphs \cite{DaMS22}).

A \textit{layering} of a graph $G$ is a partition $(V_0, V_1,\ldots)$ of $V(G)$ such that, for every edge $vw \in E(G)$, if $v \in V_i$ and $w \in V_j$, then $|i - j| \leq 1$. Each set $V_i$ is a \textit{layer}. The \textit{layered width} of a tree decomposition $\mathcal{T} = (T, \{X_t\}_{t\in V(T)})$ of a graph $G$ is the minimum integer $\ell$ 
for which there exists a layering $(V_0,V_1,\ldots)$ of $G$ such that, for each bag $X_t$ and layer $V_i$, we have $|X_t \cap V_i| \leq \ell$. The \textit{layered treewidth} of a graph $G$ is the minimum layered width of a tree decomposition of $G$. Layerings with one layer show that the layered treewidth of $G$ is at most $\tw(G)+1$. We now introduce the analogue of layered treewidth for the width parameter tree-independence number.

\begin{definition} The \textit{layered independence number} of a tree decomposition $\mathcal{T} = (T, \{X_t\}_{t\in V(T)})$ of a graph $G$ is the minimum integer $\ell$ for which there exists a layering $(V_0,V_1,\ldots)$ of $G$ such that, for each bag $X_t$ and layer $V_i$, we have $\alpha(G[X_t \cap V_i]) \leq \ell$. The \textit{layered tree-independence number} of a graph $G$ is the minimum layered independence number of a tree decomposition of $G$. 
\end{definition}

Layerings with one layer show that the layered tree-independence number of $G$ is at most $\tin(G)$. Moreover, the layered tree-independence number of a graph is clearly at most its layered treewidth. The proof of \cite[Lemma~10]{DMW17} shows, mutatis mutandis, that graphs of bounded layered tree-independence number have $O(\sqrt{n})$ tree-independence number, as we observe next.

\begin{lemma}\label{sqrttreealphaA} Let $k \in \mathbb{N}$ and let $G$ be a $n$-vertex graph. Given a tree decomposition $\mathcal{T} = (T,\{X_t\}_{t\in V(T)})$ of $G$ and a layering $(V_0,V_1,\ldots)$ of $G$ such that, for each bag $X_t$ and layer $V_i$, $\alpha(G[X_t \cap V_i]) \leq k$, it is possible to compute, in time polynomial in $n$ and $|V(T)|$, a tree decomposition of $G$ with independence number at most $2\sqrt{kn}$. In particular, every $n$-vertex graph with layered tree-independence number $k$ has tree-independence number at most $2\sqrt{kn}$.
\end{lemma}

\begin{proof} Let $p = \lceil \sqrt{n/k} \rceil$. For each $j \in \{0, \ldots, p-1\}$, let $W_j = V_j \cup V_{p+j} \cup V_{2p+j} \cup \cdots$. Observe that $(W_0, W_1, \ldots, W_{p-1})$ is a partition of $V(G)$. We then find $j \in \{0, \ldots, p-1\}$ such that $|W_j| \leq \frac{n}{p} \leq \sqrt{kn}$. Now, each component $K$ of $G - W_j$ is contained within $p - 1$ consecutive layers and so, since $\alpha(G[X_t \cap V_i]) \leq k$ for each bag $X_t$ and layer $V_i$, restricting the bags in $\mathcal{T}$ to $V(K)$ gives a tree decomposition of $K$ with independence number at most $k(p - 1) \leq \sqrt{kn}$. We then merge the tree decompositions of the components of $G - W_j$ into a tree decomposition of $G - W_j$ with independence number at most $\sqrt{kn}$. Finally, adding $W_j$ to every bag of this tree decomposition gives a tree decomposition of $G$ with independence number at most $\sqrt{kn} + |W_j| \leq 2\sqrt{kn}$.
\end{proof}

Given a width parameter $p$, a graph class $\mathcal{G}$ has \textit{bounded local $p$} if there is a function $f\colon \mathbb{N} \rightarrow \mathbb{N}$ such that for every integer $r \in \mathbb{N}$, graph $G \in \mathcal{G}$, and vertex $v \in V(G)$, the subgraph $G[N^{r}[v]]$ has $p$-width at most $f(r)$. In \cite{DMW17}, it is shown that if every graph in a class $\mathcal{G}$ has layered treewidth at most $\ell$, then $\mathcal{G}$ has bounded local treewidth with $f(r) = \ell(2r + 1) - 1$. Similarly, bounded layered tree-independence number implies bounded local tree-independence number, as we show next. 

\begin{lemma}\label{blayeredblocalA} If every graph in a class $\mathcal{G}$ has layered tree-independence number at most $\ell$, then $\mathcal{G}$ has bounded local tree-independence number with $f(r) = \ell(2r + 1)$.
\end{lemma}

\begin{proof} Let $r \in \mathbb{N}$, $G \in \mathcal{G}$ and $v \in V(G)$. Let $G' = G[N^{r}[v]]$. By assumption, $G$ has a tree decomposition $\mathcal{T} = (T, \{X_t\}_{t\in V(T)})$ of layered independence number $\ell$ with respect to some layering $(V_0, V_1,\ldots)$. Suppose that $v \in V_i$. Then $V(G') \subseteq V_{i-r} \cup \cdots \cup V_{i+r}$ and so, for each bag $X_t$, we have that $\alpha(G'[X_t]) \leq \sum_{j=-r}^{r}\alpha(G[X_t \cap V_{i-j}]) \leq \ell(2r + 1)$. This implies that $\tin(G') \leq \ell(2r + 1)$.  
\end{proof}

\begin{corollary}\label{layeredKnnA} The layered tree-independence number of $K_{n,n}$ is at least $n/5$.
\end{corollary}

\begin{proof} Suppose, to the contrary, that the layered tree-independence number of $K_{n,n}$ is less than $n/5$. Since the diameter of $K_{n,n}$ is $2$, \Cref{blayeredblocalA} implies that $\tin(K_{n,n}) < n$, contradicting the fact that $\tin(K_{n,n}) = n$ \cite{DaMS22}.  
\end{proof}

In general, bounded local tree-independence number does not imply bounded layered tree-independence number. This will be a consequence of \Cref{layeredtofragileA} and \Cref{degreefragility}. However, it is known that a proper minor-closed class has bounded layered treewidth if and only if it has bounded local treewidth if and only if it excludes some apex graph\footnote{An apex graph is a graph that can be made planar by deleting a single vertex.} as a minor (see, e.g., \cite{DMW17}). The next result extends this equivalence to layered tree-independence number and local tree-independence number.

\begin{theorem}\label{equivlayeredA} The following are equivalent for a minor-closed class $\mathcal{G}$:
\begin{enumerate}
\item Some apex graph is not in $\mathcal{G}$;
\item $\mathcal{G}$ has bounded local tree-independence number;
\item $\mathcal{G}$ has linear local tree-independence number (i.e., $f(r)$ is linear in $r$);
\item $\mathcal{G}$ has bounded layered tree-independence number.
\end{enumerate}
\end{theorem}

\begin{proof} Consider first (1) $\Rightarrow$ (4). By \cite{DMW17}, if $\mathcal{G}$ excludes some apex graph as a minor, then $\mathcal{G}$ has bounded layered treewidth, hence bounded layered tree-independence number as well. The implication (4) $\Rightarrow$ (3) follows from \Cref{blayeredblocalA}. The implication (3) $\Rightarrow$ (2) follows by definition. Finally, consider (2) $\Rightarrow$ (1). Let $G_n$ be the graph obtained from the $n\times n$-grid graph (the Cartesian product of two $n$-vertex paths) by adding a dominating vertex $v_n$. Observe that the class $\{G_n : n\in \mathbb{N}\}$ has unbounded local tree-independence number, as $v_n$ is dominating and the class of grids has unbounded tree-independence number (since it is not $(\tw,\omega)$-bounded, see \cite[Lemma~3.2]{DaMS22}). Hence, if $\mathcal{G}$ contains all apex graphs, then in particular it contains $\{G_n : n\in \mathbb{N}\}$ and so has unbounded local tree-independence number.
\end{proof}

\Cref{equivlayeredA} implies the following result from \cite{DMiS22}.

\begin{corollary} A minor-closed class has bounded tree-independence number if and only if some planar graph is not in the class.
\end{corollary}

\begin{proof} Let $\mathcal{G}$ be a minor-closed class. If $\mathcal{G}$ has bounded tree-independence number then, since the class of walls has unbounded tree-independence number \cite{DMS21a,DaMS22}, $\mathcal{G}$ does not contain a planar graph as a minor. 

Conversely, we claim that, for every planar graph $H$ there is an integer $c$ such that every $H$-minor-free graph $G$ has tree-independence number at most $c$. Let $H^{+}$ be the apex
graph obtained from $H$ by adding a dominating vertex $v$ and let $G^{+}$ be the graph obtained from $G$ by adding a dominating vertex $x$. It is shown in \cite{DMW17} that $G^{+}$ is $H^{+}$-minor-free. By \Cref{equivlayeredA}, $G^{+}$ has layered tree-independence number at most $\ell$, for some fixed integer $\ell$. Since $G^{+}$ has radius $1$, at most three layers are used. Thus $G^{+}$, and hence $G$, have tree-independence number at most $3\ell$.
\end{proof}

We now consider the behavior of layered tree-independence number with respect to graph powers. For $p \in \mathbb{N}$, the \textit{$p$-th power} of a graph $G$ is the graph $G^p$ with vertex set $V(G^p) = V(G)$ where, for distinct $u,v \in V(G^p)$, $uv \in E(G^p)$ if and only if $u$ and $v$ are at distance at most $p$ in $G$. Bonomo-Braberman and Gonzalez \cite{BG22} showed that fixed powers of bounded treewidth and bounded degree graphs are of bounded treewidth. More specifically, for any graph $G$ and $p \geq 2$, $\tw(G^p) \leq (\tw(G) + 1)(\Delta(G) + 1)^{\lceil\frac{p}{2}\rceil}-1$. It follows from the work of Dujmovi\'c et al.~\cite{DMW23} that powers of graphs of bounded layered treewidth and bounded maximum degree have bounded layered treewidth. The upper bound was improved by Dujmovi\'c et al.~\cite{DEMWW22}, who showed that if $G$ has layered treewidth $k$, then $G^p$ has layered treewidth less than $2pk\Delta(G)^{\lfloor\frac{p}{2}\rfloor}$. Lima et al.~\cite{LMM24} showed that, for any graph $G$ and odd $p \in \mathbb{N}$, $\tin(G^p) \leq \tin(G)$ and that, for every fixed even $p \in \mathbb{N}$, there is no function $f$ such that $\tin(G^p) \leq f(\tin(G))$ for all graphs $G$. We show that odd powers of bounded layered tree-independence number graphs have bounded layered tree-independence number and that this result does not extend to even powers. Before doing so, we need a definition and a result from \cite{LMM24}. Given a graph $G$ and a family $\mathcal{H} = \{H_j\}_{j\in J}$ of subgraphs of $G$, we denote by $G(\mathcal{H})$ the graph with vertex set $J$, in which two distinct elements $i, j \in J$ are adjacent if and only if $H_i$ and $H_j$ either have a vertex in common or there is an edge in $G$ connecting them.

\begin{lemma}[Lima et al.~\cite{LMM24}]\label{sametreeA} Let $G$ be a graph and let $k$ and $d$ be positive integers. For $v \in V(G)$, let $H_v$ be the subgraph of $G$ induced by the vertices at distance at most $d$ from $v$, and let $\mathcal{H} = \{H_v\}_{v\in V(G)}$. Then $G^{k+2d}$ is isomorphic to $G^k(\mathcal{H})$. 
\end{lemma}

\begin{theorem}\label{layeredtreealgoA} Let $G$ be a graph and let $d$ be a positive integer. Given a tree decomposition $\mathcal{T} = (T,\{X_t\}_{t\in V(T)})$ of $G$ and a layering $(V_1,\ldots,V_{m})$ of $G$ such that, for each bag $X_t$ and layer $V_i$, $\alpha(G[X_t \cap V_i]) \leq k$, it is possible to compute in $O(|V(T)| \cdot (|V(G)| + |E(G)|))$ time a tree decomposition $\mathcal{T'} = (T,\{X'_t\}_{t\in V(T)})$ of $G^{1+2d}$ and a layering $(V'_1,\ldots,V'_{\lceil \frac{m}{1+2d} \rceil})$ of $G^{1+2d}$ such that, for each bag $X'_t$ and layer $V'_i$, $\alpha(G^{1+2d}[X'_t \cap V'_i]) \leq (1+4d)k$. In particular, if $G$ has layered tree-independence number $k$, then $G^{1+2d}$ has layered tree-independence number at most $(1+4d)k$.
\end{theorem}

\begin{proof} Let $\mathcal{T} = (T,\{X_t\}_{t\in V(T)})$ and $(V_1,\ldots,V_{m})$ be the given tree decomposition and layering of $G$, respectively. For each vertex $u \in V(G)$, let $l(u)$ be the unique index $i$ such that $u \in V_i$. For each $v \in V(G)$, let $H_v$ be the subgraph of $G$ induced by the vertices at distance at most $d$ from $v$, and let $\mathcal{H} = \{H_v\}_{v\in V(G)}$. Let $\mathcal{T'} = (T,\{X'_t\}_{t\in V(T)})$, with $X'_t = \{v \in V(G) : V(H_v) \cap X_t \neq \varnothing\}$ for each $t \in V(T)$. By \cite[Lemma~6.1]{DaMS22}, $\mathcal{T}'$ is a tree decomposition of $G(\mathcal{H})$ and hence, by \Cref{sametreeA}, of $G^{1+2d}$ as well. Moreover, for each $v \in V(G)$, $V(H_v) \cap X_t \neq \varnothing$ if and only if $v$ is at distance at most $d$ from $X_t$ in $G$ and the set of all such vertices $v$ can be computed using BFS in $O(|V(G)| + |E(G)|)$ time. Therefore, $\mathcal{T}'$ can be computed in $O(|V(T)| \cdot (|V(G)| + |E(G)|)$ time. For each $1 \leq i \leq \lceil \frac{m}{1+2d} \rceil$, let now $V'_i = \bigcup_{(1+2d)(i-1) < j \leq (1+2d)i}V_j$. We claim that $(V'_1,\ldots,V'_{\lceil \frac{m}{1+2d} \rceil})$ is a layering of $G^{1+2d}$. Clearly, these sets partition $V(G^{1+2d})$. Moreover, for each edge $uv \in E(G^{1+2d})$, we have that $d_{G}(u,v) \leq 1+2d$ and so $|l(i)-l(j)| \leq 1+2d$. Consequently, $u$ and $v$ belong to either the same $V'_{i}$ or to consecutive $V'_{i}$'s. 

We now show that, for each $1 \leq i \leq \lceil \frac{m}{1+2d} \rceil$ and $t\in V(T)$, $\alpha(G^{1+2d}[V'_i \cap   X'_t]) \leq (1+4d)k$. Suppose, to the contrary, that $\alpha(G^{1+2d}[V'_i \cap X'_t]) > (1+4d)k$ for some $i$ and $t$ as above. Then there exists an independent set $U = \{{u_1},\ldots,{u_{(1+4d)k+1}}\}$ of $G^{1+2d}$ contained in $V'_i \cap X'_t$. By construction, $V(H_{u_p}) \cap X_t \neq \varnothing$ for each $1 \leq p \leq (1+4d)k+1$ and, for each such $p$, we pick an arbitrary vertex in $V (H_{u_p}) \cap X_t$ and denote it by $r(u_p,t)$. Note that, for $p \neq q$, $r(u_p,t)$ is distinct from and non-adjacent to $r(u_q,t)$ in $G$, for otherwise either $H_{u_p}$ and $H_{u_q}$ share a vertex or there is an edge connecting them in $G$, from which $u_pu_q \in G(\mathcal{H}) = G^{1+2d}$, contradicting the fact that $U$ is an independent set of $G^{1+2d}$. Now, by construction, each $r(u_p,t)$ belongs to $V(H_{u_p})$ and so is at distance at most $d$ in $G$ from $u_p$. Moreover, $u_p$ belongs to the layer $V'_i$ of $G^{1+2d}$, for each $1\leq p \leq (1+4d)k + 1$. Therefore, $(1+2d)(i-1) < l(u_p) \leq (1+2d)i$ and $(1+2d)(i-1) - d < l(r(u_p,t)) \leq (1+2d)i + d$, for each $1\leq p \leq (1+4d)k + 1$. That is, each $r(u_p,t)$ belongs to one of the $1+4d$ consecutive layers $V_{(1+2d)(i-1) - d + 1},\ldots, V_{(1+2d)i + d}$ of $G$. Since $|U| > (1+4d)k$, at least one such layer $V_r$ contains $k+1$ vertices of the form $r(u_p,t)$. Therefore, $\alpha(G[V_r \cap X_t]) > k$, a contradiction.
\end{proof}

\begin{lemma}\label{countereven} Fix an even $k \in \mathbb{N}$. There exist graphs $G$ with tree-independence number $1$ and such that the layered tree-independence number of $G^k$ is arbitrarily large. 
\end{lemma}

\begin{proof} By \cite[Lemma~5.16]{LMM24}, for every graph $H$, there exists a chordal graph $G$ such that $G^k$ contains an induced subgraph isomorphic to $H$. Take $H = K_{5n,5n}$ and one such $G$. By \Cref{layeredKnnA}, the layered tree-independence number of $G^k$ is at least $n$, whereas $\tin(G) = 1$.
\end{proof}


\subsection{Intersection graphs with bounded layered tree-independence number}\label{sec:layeredlemmas}

In this section we show that the following graph classes have bounded layered tree-independence number: intersection graphs of similarly-sized $c$-fat families of objects in $\mathbb{R}^2$ (in particular, unit disk graphs), unit-width rectangle graphs, and VPG/EPG graphs where the paths have bounded horizontal part and number of bends. As it will appear from the proofs, our tree decompositions witnessing this are in fact path decompositions. \Cref{sqrttreealphaA} then implies that graphs from these classes have $O(\sqrt{n})$ tree-independence number and we argue that this is tight up to constant factors. 

Note that, in general, intersection graphs of disks or rectangles in the plane and of paths on a grid (VPG/EPG graphs) all have unbounded layered tree-independence number (see \Cref{Fig:unboundlayered}). This follows from the fact that large bicliques and large grids with a dominating vertex have large layered tree-independence number, thanks to \Cref{layeredKnnA} and the proof of \Cref{equivlayeredA}, respectively.

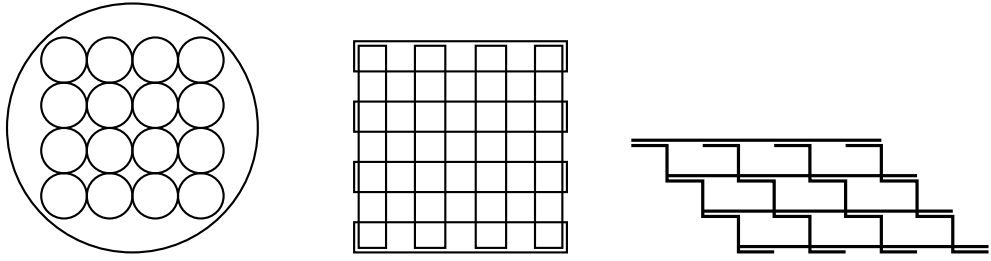
\begin{figure}[h!]
\centering
\begin{subfigure}{.3\linewidth}
\centering
\begin{tikzpicture}[scale=.6]
\draw[thick] (0,0) circle (.5cm);
\draw[thick] (1,0) circle (.5cm);
\draw[thick] (2,0) circle (.5cm);
\draw[thick] (3,0) circle (.5cm);

\draw[thick] (0,1) circle (.5cm);
\draw[thick] (1,1) circle (.5cm);
\draw[thick] (2,1) circle (.5cm);
\draw[thick] (3,1) circle (.5cm);

\draw[thick] (0,2) circle (.5cm);
\draw[thick] (1,2) circle (.5cm);
\draw[thick] (2,2) circle (.5cm);
\draw[thick] (3,2) circle (.5cm);

\draw[thick] (0,3) circle (.5cm);
\draw[thick] (1,3) circle (.5cm);
\draw[thick] (2,3) circle (.5cm);
\draw[thick] (3,3) circle (.5cm);

\draw[thick] (1.5,1.5) circle (2.75cm);
\end{tikzpicture}
\end{subfigure}
\begin{subfigure}{.3\linewidth}
\centering
\begin{tikzpicture}[scale=.4]

\draw[thick] (0,0) rectangle (7,1);
\draw[thick] (0,2) rectangle (7,3);
\draw[thick] (0,4) rectangle (7,5);
\draw[thick] (0,6) rectangle (7,7);

\draw[thick] (.15,.15) rectangle (1.05,6.85);
\draw[thick] (2,.15) rectangle (3,6.85);
\draw[thick] (4,.15) rectangle (5,6.85);
\draw[thick] (5.95,.15) rectangle (6.85,6.85);

\end{tikzpicture}
\end{subfigure}
\begin{subfigure}{.34\linewidth}
\centering
\begin{tikzpicture}[scale=.47,rotate=90]

\draw[very thick] (0,0) -- (0,1) -- (1,1) -- (1,2) -- (2,2) -- (2,3) -- (3,3) -- (3,4);
\draw[very thick] (0,2) -- (0,3) -- (1,3) -- (1,4) -- (2,4) -- (2,5) -- (3,5) -- (3,6);
\draw[very thick] (0,4) -- (0,5) -- (1,5) -- (1,6) -- (2,6) -- (2,7) -- (3,7) -- (3,8);
\draw[very thick] (0,6) -- (0,7) -- (1,7) -- (1,8) -- (2,8) -- (2,9) -- (3,9) -- (3,10);

\draw[very thick] (.15,0) -- (.15,7);
\draw[very thick] (1.15,1) -- (1.15,8);
\draw[very thick] (2.15,2) -- (2.15,9);
\draw[very thick] (3.15,3) -- (3.15,10);

\end{tikzpicture}
\end{subfigure}
\caption{Examples showing that intersection graphs of disks and rectangles in $\mathbb{R}^2$ and VPG/EPG graphs have unbounded layered tree-independence number: Realization of the $4\times 4$-grid graph with a dominating vertex as a disk graph (left), and of $K_{4,4}$ as an intersection graph of rectangles (middle) and as a VPG/EPG graph (right).}\label{Fig:unboundlayered}
\end{figure}

A collection of objects is $k$-\textit{similarly-sized} if the ratio of the largest and smallest object diameter is at most some absolute constant $k \geq 1$. 

\begin{theorem}\label{layeredfat}
Let $G$ be the intersection graph of a $k$-similarly-sized $c$-fat family $\mathcal{O}$ of $n$ objects in $\mathbb{R}^2$, for some constants $c$ and $k$. It is possible to compute, in $O(n\log n)$ time, a tree decomposition $\mathcal{T} = (T,\{X_t\}_{t\in V(T)}\})$ and a layering $(V_1, V_2, \ldots)$ of $G$ such that $|V(T)| = O(n)$ and, for each bag $X_t$ and layer $V_i$, $\alpha(G[X_t \cap V_i]) \leq \lceil 2\sqrt{2}k \rceil c$. In particular, $G$ has layered tree-independence number at most $\lceil 2\sqrt{2}k \rceil c$.
\end{theorem}

\begin{proof} Let $d_{\mathrm{min}}$ and $d_{\mathrm{max}}$ be the minimum and maximum diameter of the objects in $\mathcal{O}$, respectively. Since $\mathcal{O}$ is $k$-similarly-sized, $d_{\mathrm{max}} \leq k\cdot d_{\mathrm{min}}$. Without loss of generality, the family $\mathcal{O}$ is contained in the positive quadrant. For each vertex $v \in V(G)$, let $O_{v}$ be the corresponding object in $\mathcal{O}$. For each $O_{v}$, let $I_v$ be its projection onto the $y$-axis. Since $O_{v}$ is path-connected, $I_v$ is an interval, and let $p_v \geq q_v$ be its endpoints. Sort $\{p_v,q_v\}_{v \in V(G)}$ in $O(n\log n)$ time to obtain an ordering $z_1\leq \cdots \leq z_{2n}$. 

For each $i \in \mathbb{N}$, let $C_i = \{(x,y) \in \mathbb{R}^2 : (i-1)kd_{\mathrm{min}} \leq x < ikd_{\mathrm{min}}\}$ be the \textit{$i$-th vertical strip}. For each $i \in \mathbb{N}$, the \textit{$i$-th horizontal strip} $R_i$ is defined as follows. If $i$ is odd, then $R_i = \{(x,y) \in \mathbb{R}^2 : z_{(i+1)/2} \leq y \leq \frac{z_{(i+1)/2}+z_{(i+3)/2}}{2}\}$, whereas if $i$ is even, then $R_i = \{(x,y) \in \mathbb{R}^2 : \frac{z_{i/2}+z_{(i+2)/2}}{2} \leq y \leq z_{(i+2)/2}\}$. Note that some horizontal strips might be degenerate, i.e., they are a horizontal line. See \Cref{Fig:layeredunit} for an example.

We first construct a tree decomposition of $G$. Consider a path $T$ with $4n-2$ vertices $\{t_1,\ldots,t_{4n-2}\}$ and, for each $1\leq i \leq 4n-2$, let $X_{t_i} = \{v \in V(G) : O_v \cap R_i \neq \varnothing\}$. Clearly, for each $v\in V(G)$, there exists $i$ with $v \in X_{t_i}$. Consider now an edge $uv \in E(G)$. There exists a point $(x, y) \in \mathbb{R}^2$ contained in both $O_{u}$ and $O_{v}$, and this point belongs to some $R_k$. Therefore, $\{u, v\} \subseteq X_{t_{k}}$. Finally, for each $v \in V(G)$, the horizontal strips intersecting $O_{v}$ are consecutive, and so $\{t_i \in V(T) : v \in X_{t_i}\}$ induces a subpath of $T$. This shows that $\mathcal{T} = (T,\{X_{t_i}\}_{t_i\in V(T)})$ is a tree decomposition of $G$ and it is easy to see that such tree decomposition can be computed in $O(n)$ time. 

We now construct a layering of $G$ as follows. For each $j \in \mathbb{N}$, let $V_j$ be the set of vertices whose corresponding objects have leftmost points inside $C_j$. It is easy to see that $(V_1, V_2, \ldots)$ is a partition of $V(G)$. Consider now $i$ and $j$ with $i-j \geq 2$. If $u\in V_i$ and $v\in V_j$, then $O_v$ does not intersect $C_i$, as each object has diameter at most $k\cdot d_{\mathrm{min}}$, and $O_u$ does not intersect $C_{i-1}$. Hence, $O_u$ does not intersect $O_v$ and so $uv \not \in E(G)$. Therefore, $(V_1, V_2, \ldots)$ is a layering of $G$, which can clearly be computed in $O(n)$ time. 

Consider an arbitrary bag $X_{t_i}$ and layer $V_j$ as defined above. Suppose that $i$ is odd (the case $i$ even is similar and thus left to the reader). Let $v \in X_{t_i} \cap V_j$. Then $O_v$ intersects the line $y = z_{\frac{i+1}{2}}$, or else the lowermost points of $O_v$ would have $y$-coordinate strictly between consecutive $z_j$'s. Moreover, the leftmost points of $O_v$ belong to $C_j$. But then, since the diameter of $O_v$ is at most $k\cdot d_{\mathrm{min}}$, each object $O_v$ with $v \in X_{t_i} \cap V_j$, intersects the line $y = z_{\frac{i+1}{2}}$ in a point with $x$-coordinate in the interval $[(j-1)kd_{\mathrm{min}}, (j+1)kd_{\mathrm{min}}]$. Consider now a family $\mathcal{F}$ of axis-aligned closed boxes with pairwise disjoint interiors and satisfying the following properties: Each box has size $d_{\mathrm{min}}/\sqrt{2}$, lower side on the line $y = z_{\frac{i+1}{2}}$, and the union of the lower sides of these boxes covers the horizontal line segment $\{(x, y) \in \mathbb{R}^2 : (j-1)kd_{\mathrm{min}} \leq x \leq (j+1)kd_{\mathrm{min}} \ \mbox{and} \ y = z_{\frac{i+1}{2}}\}$. Clearly, there exists such a family $\mathcal{F}$ of size $\lceil 2\sqrt{2}k \rceil$. Since $\mathcal{O}$ is $c$-fat and each object in $\mathcal{O}$ has size at least $d_{\mathrm{min}}/\sqrt{2}$, we have that each box in $\mathcal{F}$ intersects at most $c$ pairwise non-intersecting objects from $\mathcal{O}$. This implies that there are at most $\lceil 2\sqrt{2}k \rceil c$ pairwise non-intersecting objects from $\mathcal{O}$ whose corresponding vertices belong to $X_{t_i} \cap V_j$, thus concluding the proof.
\end{proof}

As mentioned above, the similarly-sized constraint cannot be dropped, for example because of the class of disk graphs. However, we will see in \Cref{fatA} that, for any fixed $d \geq 2$, the class of intersection graphs of $c$-fat families of objects in $\mathbb{R}^d$ is efficiently fractionally $\tin$-fragile, a property weaker than boundedness of layered tree-independence number.  

\begin{figure}
\begin{center}
\begin{tikzpicture}[scale=1.2]
\draw[->,>=stealth] (-1.5,-2.7) -- (-1.5,3.9);
\draw[->,>=stealth] (-1.5,-2.7) -- (5.5,-2.7);

\node[draw=none,below] at (-1.5,-2.7) {{\small 0}};

\draw[dashed] (.5,3.9) -- (.5,-2.7) node[below] {{\small $2r$}};
\draw[dashed] (2.5,3.9) -- (2.5,-2.7) node[below] {{\small $4r$}};
\draw[dashed] (4.5,3.9) -- (4.5,-2.7) node[below] {{\small $6r$}};

\node[circle,draw,fill=PowderBlue,inner sep=2,opacity=.3] at (-4,3.41) {};
\draw (-4,3.41) ellipse (.8cm and .15cm); 

\draw[thick] (-4,3.26) -- (-4,3.16);

\node[circle,draw,fill=Pink,inner sep=2,opacity=.3] at (-4.2,3.01) {};
\node[circle,draw,fill=PowderBlue,inner sep=2,opacity=.3] at (-3.8,3.01) {};
\draw (-4,3.01) ellipse (.8cm and .15cm); 

\draw[thick] (-4,2.86) -- (-4,2.75);

\node[circle,draw,fill=Pink,inner sep=2,opacity=.3] at (-4.2,2.6) {};
\node[circle,draw,fill=PowderBlue,inner sep=2,opacity=.3] at (-3.8,2.6) {};
\draw (-4,2.6) ellipse (.8cm and .15cm); 

\draw[thick] (-4,2.45) -- (-4,2.36);

\node[circle,draw,fill=Pink,inner sep=2,opacity=.3] at (-4.4,2.21) {};
\node[circle,draw,fill=Olive,inner sep=2,opacity=.3] at (-4,2.21) {};
\node[circle,draw,fill=PowderBlue,inner sep=2,opacity=.3] at (-3.6,2.21) {};
\draw (-4,2.21) ellipse (.8cm and .15cm); 

\draw[thick] (-4,2.06) -- (-4,1.97);
\draw[thick,dotted] (-4,1.97) -- (-4,1.8);

\node[draw=none] at (-4,1.1) {\vdots};

\draw[thick] (-4,.14) -- (-4,.23);
\draw[thick,dotted] (-4,.23) -- (-4,.4);

\node[circle,draw,fill=Olive,inner sep=2,opacity=.3] at (-4.4,-.01) {};
\node[circle,draw,fill=Gold,inner sep=2,opacity=.3] at (-4,-.01) {};
\node[circle,draw,fill=Coral,inner sep=2,opacity=.3] at (-3.6,-.01) {};
\draw (-4,-.01) ellipse (.8cm and .15cm); 

\draw[thick] (-4,-.16) -- (-4,-.25);
\draw[thick,dotted] (-4,-.25) -- (-4,-.42);

\node[draw=none] at (-4,-1.16) {\vdots};

\draw[thick] (-4,-2.16) -- (-4,-2.07);
\draw[thick,dotted] (-4,-2.07) -- (-4,-1.9);

\node[circle,draw,fill=RoyalBlue,inner sep=2,opacity=.3] at (-4,-2.31) {};
\draw (-4,-2.31) ellipse (.8cm and .15cm); 

\node[draw=none,below] at (-4,-2.7) {$\mathcal{T}$};

\node[circle,draw,fill=Pink,inner sep=3.2,opacity=.3] at (-.7,-3.2) {};
\node[circle,draw,fill=PowderBlue,inner sep=3.2,opacity=.3] at (-.3,-3.2) {};
\node[draw=none] at (-.5,-3.7) {{\small $V_1$}};

\node[circle,draw,fill=Olive,inner sep=3.2,opacity=.3] at (1.1,-3.2) {};
\node[circle,draw,fill=Gold,inner sep=3.2,opacity=.3] at (1.5,-3.2) {};
\node[circle,draw,fill=Coral,inner sep=3.2,opacity=.3] at (1.9,-3.2) {};
\node[draw=none] at (1.5,-3.7) {{\small $V_2$}};

\node[circle,draw,fill=RoyalBlue,inner sep=3.2,opacity=.3] at (3.5,-3.2) {};
\node[draw=none] at (3.5,-3.7) {{\small $V_3$}};

\draw[dashed] (5.5,.79) -- (-1.5,.79) node[left] {{\small $z_7$}};
\draw[dashed] (5.5,2.81) -- (-1.5,2.81) node[left] {{\small $z_{11}$}};

\draw[dashed] (5.5,1.59) -- (-1.5,1.59) node[left] {{\small $z_9$}};
\draw[dashed] (5.5,3.61) -- (-1.5,3.61) node[left] {{\small $z_{12}$}};

\draw[dashed] (5.5,-.01) -- (-1.5,-.01) node[left] {{\small $z_5 = z_6$}};
\draw[dashed] (5.5,2.01) -- (-1.5,2.01) node[left] {{\small $z_{10}$}};

\draw[dashed] (5.5,-.71) -- (-1.5,-.71) node[left] {{\small $z_3$}};
\draw[dashed] (5.5,1.31) -- (-1.5,1.31) node[left] {{\small $z_8$}}; 

\draw[dashed] (5.5,-2.01) -- (-1.5,-2.01) node[left] {{\small $z_2$}};

\draw[dashed] (5.5,-.39) -- (-1.5,-.39) node[left] {{\small $z_4$}};
\draw[dashed] (5.5,-2.41) -- (-1.5,-2.41) node[left] {{\small $z_1$}};

\draw[dashed] (5.5,3.21) -- (-1.5,3.21);
\draw[dashed] (5.5,2.41) -- (-1.5,2.41);
\draw[dashed] (5.5,1.8) -- (-1.5,1.8);
\draw[dashed] (5.5,1.45) -- (-1.5,1.45);
\draw[dashed] (5.5,1.045) -- (-1.5,1.045);
\draw[dashed] (5.5,.39) -- (-1.5,.39);
\draw[dashed] (5.5,-.2) -- (-1.5,-.2);
\draw[dashed] (5.5,-.55) -- (-1.5,-.55);
\draw[dashed] (5.5,-1.36) -- (-1.5,-1.36);
\draw[dashed] (5.5,-2.21) -- (-1.5,-2.21);

\node[draw=none] at (5.8,-2.31) {{\small $R_1$}};
\node[draw=none] at (5.8,-2.11) {{\small $R_2$}};
\node[draw=none] at (5.8,-1.685) {{\small $R_3$}};
\node[draw=none] at (5.8,-1.035) {{\small $R_4$}};
\node[draw=none] at (5.8,-.42) {\vdots};
\node[draw=none] at (6.2,-.01) {{\small $R_9=R_{10}$}};
\node[draw=none] at (5.8,1.13) {\vdots};
\node[draw=none] at (5.8,2.21) {{\small $R_{19}$}};
\node[draw=none] at (5.8,2.6) {{\small $R_{20}$}};
\node[draw=none] at (5.8,3.01) {{\small $R_{21}$}};
\node[draw=none] at (5.8,3.41) {{\small $R_{22}$}};

\draw[thin] (0,1.8) circle (1cm);
\draw[fill=Pink,opacity=.3] (0,1.8) circle (1cm);

\draw[thin] (1.3,2.6) circle (1cm);
\draw[fill=PowderBlue,opacity=.3] (1.3,2.6) circle (1cm);

\draw[thin] (2.3,1) circle (1cm);
\draw[fill=Olive,opacity=.3] (2.3,1) circle (1cm);

\draw[thin] (2.9,.3) circle (1cm);
\draw[fill=Gold,opacity=.3] (2.9,.3) circle (1cm);

\draw[thin] (1.7,-1) circle (1cm);
\draw[fill=Coral,opacity=.3] (1.7,-1) circle (1cm);

\draw[thin] (3.6,-1.4) circle (1cm);
\draw[fill=RoyalBlue,opacity=.3] (3.6,-1.4) circle (1cm);
\end{tikzpicture}
\caption{A tree decomposition and a layering of a unit disk graph witnessing layered tree-independence number at most $2$.}
\label{Fig:layeredunit}
\end{center}
\end{figure}
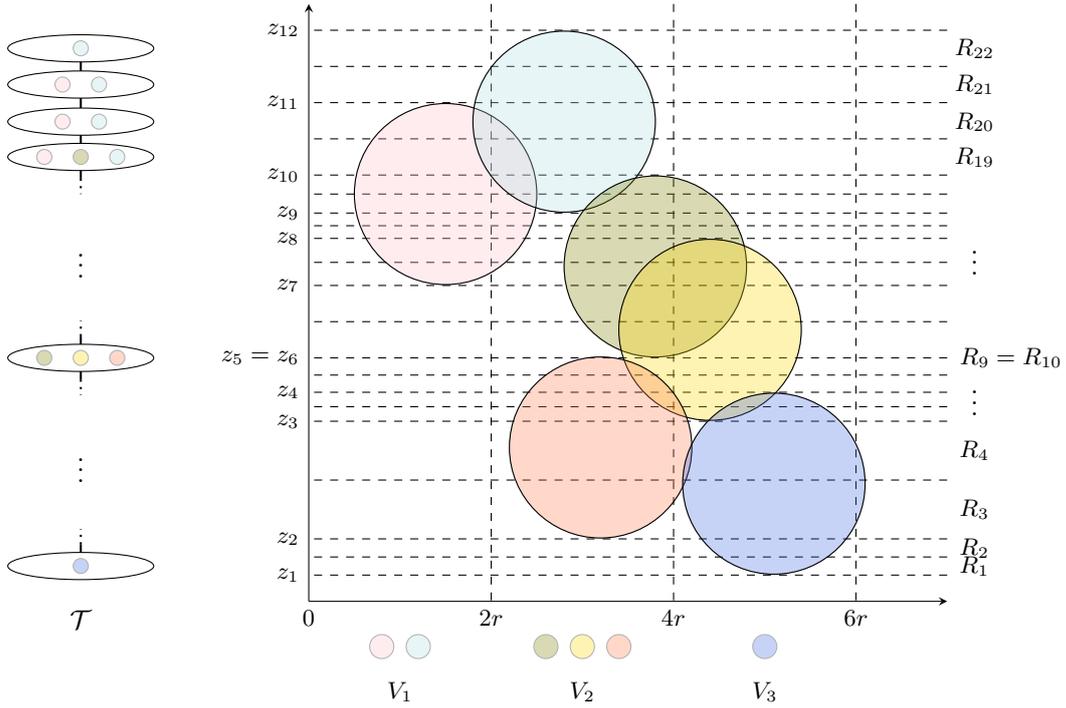

An important special case of \Cref{layeredfat} is that of unit disk graphs, for which we can obtain the following improved bound.

\begin{theorem}\label{unitdisklayeredA} Let $G$ be the intersection graph of a family $\mathcal{D}$ of $n$ unit disks. It is possible to compute, in $O(n\log n)$ time, a tree decomposition $\mathcal{T} = (T,\{X_t\}_{t\in V(T)}\})$ and a layering $(V_1, V_2, \ldots)$ of $G$ such that $|V(T)| = O(n)$ and, for each bag $X_t$ and layer $V_i$, $\alpha(G[X_t \cap V_i]) \leq 3$. In particular, $G$ has layered tree-independence number at most $3$. 
\end{theorem}

\begin{proof} Let $r$ be the common radius of the disks. Without loss of generality, the family $\mathcal{D}$ is contained in the positive quadrant. For each vertex $v \in V(G)$, let $D_{v}$ be the corresponding disk in $\mathcal{D}$. Let $z_1\leq \cdots \leq z_{2n}$ be the ordering obtained as in the proof of \Cref{layeredfat} by projecting the disks onto the $y$-axis. The construction of a tree decomposition and layering of $G$ is similar to that of \Cref{layeredfat}. For each $i \in \mathbb{N}$, let $C_i = \{(x,y) \in \mathbb{R}^2 : 2(i-1)r \leq x < 2ir\}$ be the $i$-th vertical strip. For each $i \in \mathbb{N}$, the $i$-th horizontal strip $R_i$ is defined as in \Cref{layeredfat}. We now construct a tree decomposition $\mathcal{T} = (T,\{X_{t_i}\}_{t_i\in V(T)})$ and a layering $(V_1, V_2, \ldots)$ of $G$ precisely as in \Cref{layeredfat} (see \Cref{Fig:layeredunit}). 

Consider now an arbitrary bag $X_{t_i}$ and layer $V_j$. Suppose that $i$ is odd (the case $i$ even is similar). Let $v \in X_{t_i} \cap V_j$. Then it is easy to see that the center of $D_v$ must belong to the region $S = \{(x, y) \in \mathbb{R}^2 : 2ir - r \leq x < 2ir + r \ \mbox{and} \ z_{\frac{i+1}{2}} - r \leq y \leq z_{\frac{i+1}{2}} + r\}$. But there are at most three non-intersecting unit disks with centers in $S$. Indeed, suppose four such centers lie inside $S$. By assumption, they are at pairwise distance greater than $2r$. If they are in convex position, then two of them must be at distance greater than $2\sqrt{2}r$. If they are not in convex position, then two of them must be at distance greater than $2\sqrt{3}r$. In either case we obtain a contradiction to the fact that the diameter of $S$ is at most $2\sqrt{2}r$. 
\end{proof}

Note that \Cref{unitdisklayeredA} (and hence \Cref{layeredfat} as well) cannot be extended to $\mathbb{R}^d$ with $d \geq 3$. Indeed, every $3$-dimensional grid graph is the intersection graph of a family of unit balls in $\mathbb{R}^3$. Moreover, since the class of $3$-dimensional grids has unbounded layered treewidth \cite{DEW17}, it has unbounded layered tree-independence number as well. This follows from the fact that, by Ramsey's theorem, every graph with tree-independence number at most $k$ and clique number at most $p$ has treewidth at most $R(k+1,p+1)-2$ \cite{DaMS22}.  

It is easy to see that the collection of unit-width rectangles is not $c$-fat. However, we now argue that intersection graphs of unit-width rectangles form another class of bounded layered tree-independence number.

\begin{theorem}\label{layeredrectA}
Let $G$ be the intersection graph of a family $\mathcal{R}$ of $n$ unit-width rectangles. It is possible to compute, in $O(n\log n)$ time, a tree decomposition $\mathcal{T} = (T,\{X_t\}_{t\in V(T)}\})$ and a layering $(V_1, V_2, \ldots)$ of $G$ such that $|V(T)| = O(n)$ and, for each bag $X_t$ and layer $V_i$, $\alpha(G[X_t \cap V_i]) \leq 1$. In particular, $G$ has layered tree-independence number at most $1$.
\end{theorem}

\begin{proof} Let $c$ be the common width of the rectangles. Without loss of generality, the family $\mathcal{R}$ is contained in the positive quadrant. For each vertex $v \in V(G)$, let $R_{v}$ be the corresponding rectangle in $\mathcal{R}$. Let $z_1\leq \cdots \leq z_{2n}$ be the ordering obtained as in the proof of \Cref{layeredfat} by projecting the rectangles onto the $y$-axis. The construction of a tree decomposition and layering of $G$ is again similar to that of \Cref{layeredfat}, the only difference being in the definition of the vertical strips: For each $i \in \mathbb{N}$, let $C_i = \{(x,y) \in \mathbb{R}^2 : (i-1)c \leq x < ic\}$ be the $i$-th vertical strip. We then construct a tree decomposition $\mathcal{T} = (T,\{X_{t_i}\}_{t_i\in V(T)})$ and a layering $(V_1, V_2, \ldots)$ of $G$ precisely as in \Cref{layeredfat}. 

Consider an arbitrary bag $X_{t_i}$ and layer $V_j$. Suppose that $i$ is odd (the case $i$ even is similar). Let $v \in X_{t_i} \cap V_j$. Then the rectangle $R_v$ must intersect the line $y = z_{\frac{i+1}{2}}$ and the leftmost points of $R_v$ belong to $C_j$. Clearly, there are no two disjoint such rectangles, thus concluding the proof.  
\end{proof}

We now consider VPG and EPG graphs.

\begin{theorem}\label{layeredVPGA}
Let $G$ be a graph on $n$ vertices together with a grid representation $\mathcal{R} = (\mathcal{G}, \mathcal{P}, x)$ such that each path in $\mathcal{P}$ has horizontal part of length at most $\ell$, for some fixed $\ell \geq 1$, and number of bends constant. It is possible to compute, in $O(n\log n)$ time, a tree decomposition $\mathcal{T} = (T, \{X_t\}_{t\in V(T)})$ and a layering $(V_1,V_2\ldots)$ of $G$ such that $|V(T)| = O(n)$ and, for each bag $X_t$ and layer $V_i$, 
\begin{itemize}
\item $\alpha(G[X_t \cap V_i]) \leq 2\ell$, if $x = v$; 
\item $\alpha(G[X_t \cap V_i]) \leq 6\ell-1$, if $x = e$. 
\end{itemize}
In particular, $G$ has layered tree-independence number at most $2\ell$, if $x = v$, and at most $6\ell-1$, if $x = e$. 
\end{theorem}

\begin{proof} For each vertex $v \in V(G)$, let $P_{v}$ be the corresponding path in $\mathcal{P}$. Let $z_1\leq \cdots \leq z_{2n}$ be the ordering obtained as in the proof of \Cref{layeredfat} by projecting the paths onto the $y$-axis. The construction of a tree decomposition and layering of $G$ is again similar to that of \Cref{layeredfat}, the only difference being in the definition of the vertical strips: For each $i \in \mathbb{N}$, let $C_i = \{(x,y) \in \mathbb{R}^2 : (i-1)\ell \leq x < i\ell\}$ be the $i$-th vertical strip. We then construct a tree decomposition $\mathcal{T} = (T,\{X_{t_i}\}_{t_i\in V(T)})$ and a layering $(V_1, V_2, \ldots)$ of $G$ precisely as in \Cref{layeredfat}. 

Consider an arbitrary bag $X_{t_i}$ and layer $V_j$. Suppose that $i$ is odd (the case $i$ even is similar). Let $I$ be an independent set of $G[X_{t_i} \cap V_j]$. Observe that, for each $v \in I \subseteq X_{t_i} \cap V_j$, the path $P_v$ contains a grid-point at the intersection of the row of $\mathcal{G}$ indexed by $z_{\frac{i+1}{2}}$ and a column of $\mathcal{G}$ indexed by $p$, for some $ (j-1)\ell \leq p < (j+1)\ell$. If $x = v$, then each such grid-point belongs to at most one $P_v$ with $v \in I$, from which $|I| \leq 2\ell$. If $x = e$, each of the at most $6\ell - 1$ grid-edges containing such grid-points belongs to at most one $P_v$ with $v \in I$ (or else two distinct paths $P_u$ and $P_v$ with $u, v \in I$ share a grid-edge), from which $|I| \leq 6\ell - 1$.                    
\end{proof}

\begin{remark} Let $G$ be the intersection graph of a family $\mathcal{O}$ of objects in $\mathbb{R}^2$. It is interesting to observe what are the key properties that guarantee boundedness of layered tree-independence number for $G$ in the previous four results. What we need in our arguments are the following two properties:
\begin{itemize}
\item The objects in $\mathcal{O}$ have similar widths: the ratio of the largest ($w_{\mathrm{max}}$) and smallest ($w_{\mathrm{min}}$) object width (i.e., length of the horizontal part) is bounded; 
\item Any horizontal segment of length at most $w_{\mathrm{max}}$ stabs a bounded number of pairwise non-intersecting objects from $\mathcal{O}$.    
\end{itemize}
\end{remark}

We have observed earlier (in the proof of \Cref{sqrttreealphaA}) that if a graph $G$ has layered tree-independence number at most $\ell$ and the witnessing layering consists of $c$ layers, then $G$ has tree-independence number at most $\ell c$. We can apply this easy observation to the tree decomposition and layering built in the proof of \Cref{unitdisklayeredA,layeredrectA,layeredVPGA}, respectively, in order to obtain constant tree-independence number in case the corresponding geometric realization is contained in an axis-aligned rectangle with bounded width. This is the content of the next three results, whose proofs are easy and thus only sketched. These results will then be used in \Cref{sec:improvedA} to obtain simple Baker-style PTASes for \textsc{Max Weight Independent Set} with running times matching or improving the state of the art.   

\begin{corollary}\label{diskA}
Let $G$ be the intersection graph of a family $\mathcal{D}$ of $n$ unit disks of common radius $r$ such that its geometric realization is contained in an axis-aligned rectangle of width at most $\ell$, for some integer $\ell > 0$. It is possible to compute, in $O(n\log n)$ time, a tree decomposition $\mathcal{T} = (T, \{X_t\}_{t\in V(T)})$ of $G$ such that $|V(T)| = O(n)$ and $\alpha(\mathcal{T}) \leq 3\lceil\frac{\ell}{2r}\rceil$. 
\end{corollary}

\begin{proof} Build a tree decomposition $\mathcal{T} = (T, \{X_t\}_{t\in V(T)})$ and a layering $(V_1, V_2, \ldots)$ of $G$ as in the proof of \Cref{unitdisklayeredA}. Recall that, for each $j \in \mathbb{N}$, $V_j$ is the set of vertices whose corresponding disks have leftmost points inside the $j$-th vertical strip $C_j = \{(x,y) \in \mathbb{R}^2 : 2(j-1)r \leq x < 2jr\}$. Since the geometric realization of $G$ is contained in an axis-aligned rectangle of width at most $\ell$, the disks in $\mathcal{D}$ intersect at most $\lceil\frac{\ell}{2r}\rceil$ vertical strips and so there are at most $\lceil\frac{\ell}{2r}\rceil$ layers. By \Cref{unitdisklayeredA}, for each bag $X_t$ and layer $V_i$, $\alpha(G[X_t \cap V_i]) \leq 3$. Therefore, $\alpha(\mathcal{T}) \leq \max_{t \in V(T)}\sum_{i\in\mathbb{N}}\alpha(G[X_t \cap V_i]) \leq 3\lceil\frac{\ell}{2r}\rceil$.    
\end{proof}

Mutatis mutandis we obtain the following.

\begin{corollary}\label{rectangleA}
Let $G$ be the intersection graph of a family $\mathcal{R}$ of $n$ unit-width rectangles of common width $c$ such that its geometric realization is contained in an axis-aligned rectangle of width at most $\ell$, for some integer $\ell > 0$. It is possible to compute, in $O(n\log n)$ time, a tree decomposition $\mathcal{T} = (T, \{X_t\}_{t\in V(T)})$ of $G$ such that $|V(T)| = O(n)$ and $\alpha(\mathcal{T}) \leq \lceil\frac{\ell}{c}\rceil$. 
\end{corollary}

\begin{corollary}\label{pathA}
Let $G$ be a $n$-vertex graph together with a grid representation $\mathcal{R} = (\mathcal{G}, \mathcal{P},x)$ such that $\mathcal{G}$ contains at most $\ell$ columns, for some integer $\ell \geq 1$, and each path in $\mathcal{P}$ has number of bends constant. It is possible to compute, in $O(n\log n)$ time, a tree decomposition $\mathcal{T} = (T, \{X_t\}_{t\in V(T)})$ of $G$ such that $|V(T)| = O(n)$ and:    
\begin{itemize}
\item $\alpha(\mathcal{T}) \leq \ell$, if $x = v$;
\item $\alpha(\mathcal{T}) \leq 3\ell - 1$, if $x = e$.  
\end{itemize}
\end{corollary}

\begin{proof} Note that a direct application of \Cref{layeredVPGA} gives the upper bounds $2\ell$ (if $x = v$) and $6\ell-1$ (if $x = e$). These can be easily improved as follows. Simply build a tree decomposition $\mathcal{T} = (T, \{X_t\}_{t\in V(T)})$ as in the proof of \Cref{layeredVPGA}. Now, let $t \in V(T)$ and let $I$ be an independent set of $G[X_t]$. For each $v \in I$, the path $P_v$ contains a grid-point on a fixed row of $\mathcal{G}$. But these grid-points are within $\ell$ columns, and so $|I| \leq \ell$, if $x = v$, and $|I| \leq 3\ell - 1$, if $x = e$.                    
\end{proof}

In general, bounded layered tree-independence number implies sublinear (in the number of vertices) tree-independence number. More specifically, \Cref{layeredfat,layeredrectA,layeredVPGA}, paired with \Cref{sqrttreealphaA}, have the following immediate consequence.

\begin{corollary}\label{tighttreealpha} If $G$ is a graph on $n$ vertices belonging to one of the following classes, then $G$ has $O(\sqrt{n})$ tree-independence number:
\begin{itemize}
\item Intersection graphs of $c$-fat $k$-similarly-sized families of objects in $\mathbb{R}^2$, for some constants $c$ and $k$ (hence, in particular, unit disk graphs);
\item Intersection graphs of unit-width rectangles in $\mathbb{R}^2$;
\item VPG or EPG graphs where each path has bounded horizontal part and number of bends.
\end{itemize}
\end{corollary}

In particular, unit disk graphs have $O(\sqrt{n})$ tree-independence number, from which it is easy to see that unit disk graphs of bounded degree have $O(\sqrt{n})$ treewidth. The latter was first observed by Fomin et al.~\cite{FLS12}. 

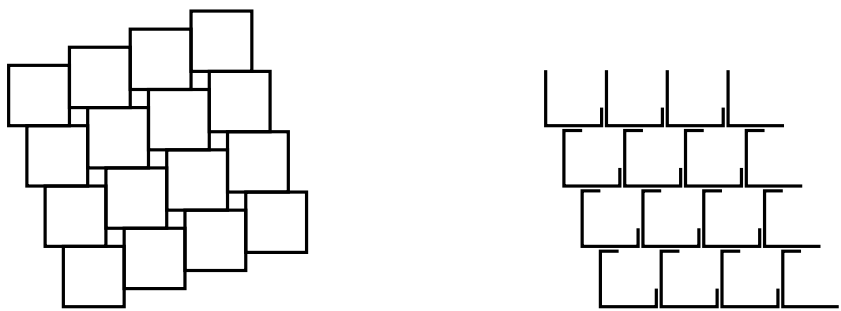
\begin{figure}
\begin{subfigure}{.45\linewidth}
\begin{center}
\begin{tikzpicture}[scale=0.8]
\draw[very thick] (0,0) rectangle (1,1);
\draw[very thick] (-.3,1) rectangle (.7,2); 
\draw[very thick] (-.6,2) rectangle (.4,3);
\draw[very thick] (-.9,3) rectangle (.1,4);

\draw[very thick] (1,.3) rectangle (2,1.3);
\draw[very thick] (.7,1.3) rectangle (1.7,2.3); 
\draw[very thick] (.4,2.3) rectangle (1.4,3.3);
\draw[very thick] (.1,3.3) rectangle (1.1,4.3);

\draw[very thick] (2,.6) rectangle (3,1.6);
\draw[very thick] (1.7,1.6) rectangle (2.7,2.6); 
\draw[very thick] (1.4,2.6) rectangle (2.4,3.6);
\draw[very thick] (1.1,3.6) rectangle (2.1,4.6);

\draw[very thick] (3,.9) rectangle (4,1.9);
\draw[very thick] (2.7,1.9) rectangle (3.7,2.9); 
\draw[very thick] (2.4,2.9) rectangle (3.4,3.9);
\draw[very thick] (2.1,3.9) rectangle (3.1,4.9);
\end{tikzpicture}
\end{center}
\end{subfigure}
\hspace*{.5cm}
\begin{subfigure}{.45\linewidth}
\begin{center}
\begin{tikzpicture}[scale=0.8]
\draw[very thick] (.3,.92) -- (0,.92) -- (0,0) -- (.92,0) -- (.92,.3);
\draw[very thick] (1.3,.92) -- (1,.92) -- (1,0) -- (1.92,0) -- (1.92,.3);
\draw[very thick] (2.3,.92) -- (2,.92) -- (2,0) -- (2.92,0) -- (2.92,.3);
\draw[very thick] (3.3,.92) -- (3,.92) -- (3,0) -- (3.92,0);

\draw[very thick] (0,1.92) -- (-.3,1.92) -- (-.3,1) -- (.62,1) -- (.62,1.3);
\draw[very thick] (1,1.92) -- (.7,1.92) -- (.7,1) -- (1.62,1) -- (1.62,1.3);
\draw[very thick] (2,1.92) -- (1.7,1.92) -- (1.7,1) -- (2.62,1) -- (2.62,1.3);
\draw[very thick] (3,1.92) -- (2.7,1.92) -- (2.7,1) -- (3.62,1);

\draw[very thick] (-.3,2.92) -- (-.6,2.92) -- (-.6,2) -- (.32,2) -- (.32,2.3);
\draw[very thick] (.7,2.92) -- (.4,2.92) -- (.4,2) -- (1.32,2) -- (1.32,2.3);
\draw[very thick] (1.7,2.92) -- (1.4,2.92) -- (1.4,2) -- (2.32,2) -- (2.32,2.3);
\draw[very thick] (2.7,2.92) -- (2.4,2.92) -- (2.4,2) -- (3.32,2);

\draw[very thick] (-.9,3.92) -- (-.9,3) -- (.02,3) -- (.02,3.3);
\draw[very thick] (.1,3.92) -- (.1,3) -- (1.02,3) -- (1.02,3.3);
\draw[very thick] (1.1,3.92) -- (1.1,3) -- (2.02,3) -- (2.02,3.3);
\draw[very thick] (2.1,3.92) -- (2.1,3) -- (3.02,3);
\end{tikzpicture}
\end{center}
\end{subfigure}
\caption{The $4\times 4$-grid graph as the intersection graph of a family of unit squares, and as a VPG/EPG graph where each path has bounded horizontal part and at most $3$ bends.}
\label{Fig:gridgraph}
\end{figure}

We now argue that the bounds obtained in \Cref{tighttreealpha} are tight up to constant factors. Since every grid graph can be realized as the intersection graph of a family of unit disks or unit-width rectangles, and is a VPG/EPG graph where each path has bounded horizontal part and number of bends (see \Cref{Fig:unboundlayered,Fig:gridgraph}), it is enough to show that the $n \times n$-grid graph has tree-independence number $\Omega(n)$. In order to do so, we first recall some definitions. A pair of vertex subsets $(A,B)$ is a \textit{separation} in a graph $G$ if $A \cup B = V(G)$ and there is no edge between $A \setminus B$ and $B \setminus A$. The \textit{size} of a separation $(A, B)$ is the quantity $|A \cap B|$. A separation $(A,B)$ is \textit{balanced} if $|A \setminus B| \leq 2|V(G)|/3$ and $|B \setminus A| \leq 2|V(G)|/3$. We also need the following known result (see \cite[Lemma~7.20]{CFK} for a proof).

\begin{lemma}[Folklore]\label{separator} Let $\mathcal{T}=(T,\{X_t\}_{t\in V(T)})$ be a tree decomposition of a graph $G$. There exists a balanced separation $(A, B)$ in $G$ such that $A \cap B = X_t$, for some bag $X_t$ of $\mathcal{T}$.
\end{lemma}

\begin{lemma} The $n \times n$-grid graph has tree-independence number $\Omega(n)$.
\end{lemma}

\begin{proof} Let $G$ be the $n \times n$-grid graph and let $\mathcal{T}$ be an arbitrary tree decomposition of $G$. It is well known (and an easy exercise) that, for $n \geq 4$, every balanced separation of $G$ has size at least $n/4$. \Cref{separator} then implies that there exists a bag of $\mathcal{T}$ of size at least $n/4$. But since $G$ is bipartite, we obtain that $\alpha(\mathcal{T}) \geq n/8$.   
\end{proof}


\section{Fractional $\boldsymbol{\tin}$-fragility}\label{sec:fragilityA}

Let $p$ be a width parameter in $\{\tw, \tin\}$. Fractional $\tw$-fragility was first defined in \cite{Dvo16}. We provide here an equivalent definition from \cite{Dvo22}, which was explicitly extended to the case $p = \tin$ in \cite{GWP22}, and which constitutes the key notion of the paper. We then investigate necessary and sufficient conditions guaranteeing fractional $\tin$-fragility.

\begin{definition}
For $\beta \leq 1$, a $\beta$-general cover of a graph $G$ is a multiset $\mathcal{C}$ of subsets of $V(G)$ such that each vertex belongs to at least $\beta|\mathcal{C}|$ elements of the cover. The $p$-width of the cover is $\max_{C \in \mathcal{C}}p(G[C])$.

For a parameter $p$, a graph class $\mathcal{G}$ is fractionally $p$-fragile if there exists a function $f\colon\mathbb{N}\rightarrow\mathbb{N}$ such that, for every $r\in\mathbb{N}$, every $G \in \mathcal{G}$ has a $(1 - 1/r)$-general cover with $p$-width at most $f(r)$. 

A fractionally $p$-fragile class $\mathcal{G}$ is efficiently fractionally $p$-fragile if there exists an algorithm that, for every $r\in\mathbb{N}$ and $G \in \mathcal{G}$, returns in $\mathsf{poly}(|V(G)|)$ time a $(1 - 1/r)$-general cover $\mathcal{C}$ of $G$ and, for each $C \in \mathcal{C}$, a tree decomposition of $G[C]$ of width (if $p = \tw$) or independence number (if $p = \tin$) at most $f(r)$, for some function $f\colon\mathbb{N}\rightarrow\mathbb{N}$.
\end{definition}

Note that classes of bounded tree-independence number are efficiently fractionally $\tin$-fragile thanks to \cite{DFGK22}. Hence, the family of efficiently fractionally $\tin$-fragile classes contains the two incomparable families of bounded tree-independence number classes and efficiently fractionally $\tw$-fragile classes (to see that they are incomparable, consider chordal graphs and planar graphs). We now identify one more subfamily, consisting of classes of bounded layered tree-independence number. 

\begin{lemma}\label{layeredtofragileA} Let $\ell \in \mathbb{N}$ and let $G$ be a graph. For each $r \in \mathbb{N}$, given a tree decomposition $\mathcal{T} = (T,\{X_t\}_{t\in V(T)})$ and a layering $(V_0,V_1,\ldots)$ of $G$ such that, for each bag $X_t$ and layer $V_i$, $\alpha(G[X_t \cap V_i]) \leq \ell$, it is possible to compute in $O(|V(G)| + |E(G)| + |V(G)|\cdot|V(T)|)$ time a $(1 - 1/r)$-general cover $\mathcal{C}$ of $G$ and, for each $C \in \mathcal{C}$, a tree decomposition of $G[C]$ with independence number at most $\ell (r-1)$. In particular, if every graph in a class $\mathcal{G}$ has layered tree-independence number at most $\ell$, then $\mathcal{G}$ is fractionally $\tin$-fragile with $f(r) = \ell (r-1)$.
\end{lemma}

\begin{proof} Fix $r \in \mathbb{N}$. Let $\mathcal{T} = (T, \{X_t\}_{t\in V(T)})$ and $(V_0,V_1,\ldots)$ be the given tree decomposition and layering of $G$, respectively. For each $m \in \{0,\ldots,r-1\}$, let $C_m = \bigcup_{i \not\equiv m \pmod r}V_i$. We claim that $\mathcal{C} = \{C_m : 0 \leq m \leq r-1\}$ is a $(1-1/r)$-general cover of $G$ with tree-independence number at most $\ell (r-1)$. Observe first that each $v \in V(G)$ is not covered by exactly one element of $\mathcal{C}$ and so it belongs to $r-1 = (1-1/r)|\mathcal{C}|$ elements of $\mathcal{C}$. Let now $C \in \mathcal{C}$. Each component $K$ of $G[C]$ is contained within $r-1$ consecutive layers and so, since $\alpha(G[X_t \cap V_i]) \leq \ell$ for each bag $X_t$ and layer $V_i$, restricting the bags in $\mathcal{T}$ to $V(K)$ gives a tree decomposition of $K$ with independence number at most $\ell (r-1)$. We then merge the tree decompositions of the components of $G[C]$ into a tree decomposition of $G[C]$ with independence number at most $\ell (r-1)$.    
\end{proof}

\begin{remark}\label{layeredtwfragile} The argument used in \Cref{layeredtofragileA} also implies that, if every graph in a class $\mathcal{G}$ has bounded layered treewidth, then $\mathcal{G}$ is fractionally $\tw$-fragile. 
\end{remark}

But what are necessary conditions for fractional $\tin$-fragility? Contrary to fractional $\tw$-fragility, where sublinear-size separators are needed \cite{Dvo16}, one should expect that in the case of fractional $\tin$-fragility it is not the size of a separator that has to be small but rather its independence number. In order to formally prove this, we introduce some notation. 

Recall that a pair of vertex subsets $(A,B)$ is a \textit{separation} in a graph $G$ if $A \cup B = V(G)$ and there is no edge between $A \setminus B$ and $B \setminus A$. The \textit{separator} of a separation $(A, B)$ is the set $A \cap B$. The \textit{size} of a separation $(A, B)$ is the quantity $|A \cap B|$, whereas the \textit{independence number} of $(A, B)$ is the quantity $\alpha(G[A \cap B])$. A separation $(A,B)$ is \textit{balanced} if $|A \setminus B| \leq 2|V(G)|/3$ and $|B \setminus A| \leq 2|V(G)|/3$. The minimum independence number of a balanced separation in $G$ is denoted by $s\mhyphen\alpha(G)$. For a graph class $\mathcal{G}$, let $s\mhyphen\alpha_{\mathcal{G}}(n)$ denote the smallest non-negative integer such that every graph in $\mathcal{G}$ with at most $n$ vertices has a balanced separation of independence number at most $s\mhyphen\alpha_{\mathcal{G}}(n)$. In other words, $s\mhyphen\alpha_{\mathcal{G}}(n) = \max\{s\mhyphen\alpha(G): G \in \mathcal{G}, |V(G)| \leq n\}$. We say that $\mathcal{G}$ has \textit{separators of sublinear independence number} if $\lim_{n\rightarrow\infty}\frac{s\mhyphen\alpha_{\mathcal{G}}(n)}{n} = 0$. 

\begin{lemma}\label{separators} Let $\mathcal{G}$ be a fractionally $\tin$-fragile class and let $c \in \mathbb{N}$ be arbitrary. Then there exists $k \in \mathbb{N}$ such that, for each $G \in \mathcal{G}$ with $|V(G)| \geq k$, $s\mhyphen\alpha(G) < |V(G)|/c$. In particular, every fractionally $\tin$-fragile class has separators of sublinear independence number. 
\end{lemma}

\begin{proof}
Since $\mathcal{G}$ is fractionally $\tin$-fragile, there exists a function $f\colon\mathbb{N}\rightarrow\mathbb{N}$ such that, for every $r\in\mathbb{N}$, every $G \in \mathcal{G}$ has a $(1 - 1/r)$-general cover $\mathcal{C}$ such that, for each $C \in \mathcal{C}$, $\tin(G[C]) \leq f(r)$. We show that the statement holds by taking $k = 3cf(2c)$. To this end, let $G\in \mathcal{G}$ be an arbitrary graph with $|V(G)| \geq k$, and let $\mathcal{C}$ be a $(1-\frac{1}{2c})$-general cover of $G$ as above. Then there exists $C \in \mathcal{C}$ such that $|C| \geq (1-\frac{1}{2c})|V(G)|$. Moreover, there exists a tree decomposition $\mathcal{T}=(T,\{X_t\}_{t\in V(T)})$ of $G[C]$ with independence number at most $f(2c)$. 
By \Cref{separator}, there exists a separation $(A, B)$ in $G[C]$ such that $A \cap B = X_t$, for some bag $X_t$ of $\mathcal{T}$, and $|A \setminus B| \leq 2|C|/3$ and $|B \setminus A| \leq 2|C|/3$. Observe now that $(A \cup (V(G) \setminus C), B \cup (V(G) \setminus C))$ is a balanced separation in $G$ with separator $X = (A\cap B) \cup (V(G) \setminus C) = X_t \cup (V(G) \setminus C)$. Moreover, \[\alpha(G[X]) \leq \alpha(G[X_t]) + |V(G)\setminus C| \leq f(2c) + \frac{1}{2c}|V(G)| \leq \frac{1}{3c}|V(G)| + \frac{1}{2c}|V(G)| < \frac{1}{c}|V(G)|.\] This implies that $s\mhyphen\alpha(G) < |V(G)|/c$, as claimed.
\end{proof}

Recall from \Cref{layeredKnnA} that large induced bicliques are an obstruction to small layered tree-independence number. \Cref{separators} immediately implies that they remain an obstruction to fractional $\tin$-fragility in hereditary graph classes. In fact, this is true even if the graph class is not hereditary, as we show next. Here the \textit{induced biclique number} of a graph $G$ is the maximum $n \in \mathbb{N}$ such that the complete bipartite graph $K_{n,n}$ is an induced subgraph of $G$. 

\begin{lemma}\label{bicliqueA}
Every fractionally $\tin$-fragile graph class has bounded induced biclique number. 
\end{lemma}

\begin{proof}
Observe first that, if $\mathcal{C}$ is a $\beta$-general cover of a graph $G$ and $H$ is an induced subgraph of $G$, then $\mathcal{C} \cap H = \{C \cap V(H) : C \in \mathcal{C}\}$ is a $\beta$-general cover of $H$. Recalling that tree-independence number does not increase when taking induced subgraphs, it is therefore enough to show the following. For any function $f\colon \mathbb{N} \rightarrow \mathbb{N}$ and integer $r >2$, there exists $n$ such that no $(1-1/r)$-general cover of $K_{n,n}$ has tree-independence number less than $f(r)$. To this end, fix arbitrary $f\colon \mathbb{N} \rightarrow \mathbb{N}$ and $r>2$. Consider a copy $H$ of $K_{n,n}$, with $n > f(r)/(1 - 2/r)$. Let $\mathcal{C}$ be a $(1-1/r)$-general cover of $H$. Then every vertex of $H$ belongs to at least $(1-1/r)|\mathcal{C}|$ elements of $\mathcal{C}$ and so there exists $C \in \mathcal{C}$ of size at least $2n(1-1/r)$. Let $A$ and $B$ be the two bipartition classes of $H$. Then $|A\cap C| \geq |C| - |B| \geq 2n(1-1/r) - n = n(1-2/r) > f(r)$ and, similarly, $|B\cap C| > f(r)$. Therefore, $H[C]$ contains $K_{f(r),f(r)}$ as an induced subgraph and since $\tin(K_{f(r),f(r)}) = f(r)$ \cite{DaMS22}, $\tin(H[C]) \geq f(r)$. 
\end{proof}

\Cref{separators} has another consequence. Recall that, for fixed $\alpha > 0$, a graph $G$ is an \textit{$\alpha$-expander} if, for every $S \subseteq V(G)$ of size at most $|V(G)|/2$, there are at least $\alpha|S|$ edges of $G$ with exactly one endpoint in $S$. $3$-regular $\alpha$-expanders on $n$ vertices are known to exist for each sufficiently large even $n$ and each of their separations has size $\Omega(n)$ (see, e.g., \cite{DSW16,Dvo16}). Therefore, \Cref{separators} and Brooks' theorem imply that $3$-regular $\alpha$-expanders are not fractionally $\tin$-fragile. The same holds for $1$-subdivisions of $3$-regular $\alpha$-expanders, where we use \cite[Lemma~28]{Dvo16} to bound the size of a separation. We thus obtain the following result.

\begin{corollary}\label{degreefragility} There exist classes of subcubic $K_{2,2}$-free graphs which are not fractionally $\tin$-fragile.    
\end{corollary}

We conclude this section with another explicit construction of a graph class which is not fractionally $\tin$-fragile. The \textit{$d$-dimensional grid of side length $n$}, denoted $G_{d,n}$, is the graph with vertex set $[n]^d = \{(x_1,\ldots,x_d) : x_i \in \{1, 2, \ldots, n\} \ \mbox{for each} \ i\}$, where two vertices $(x_1,\ldots,x_d)$ and $(y_1,\ldots,y_d)$ are adjacent if and only if $\sum_{1 \leq i \leq d} |x_i-y_i|=1$.

\begin{lemma}\label{gridsunbounded} Let $I, J \subseteq \mathbb{N}$ with $I$ infinite and $J \neq \{1\}$. The class $\{G_{d,n} : d\in I, n \in J\}$ is not fractionally $\tin$-fragile. 
\end{lemma}

\begin{proof} Fix arbitrary $f\colon \mathbb{N}\rightarrow\mathbb{N}$ and $r>2$. For such a choice, fix $d \in \mathbb{N}$ such that $\frac{r-4}{2r}d + 1 \geq R(3, f(r))$, where $R(3, f(r))$ denotes a Ramsey number. We now show that every $(1-1/r)$-general cover of $G_{d,n}$ has tree-independence at least $f(r)$. Let $\mathcal{C}$ be a $(1-1/r)$-general cover of $G_{d,n}$. Then every vertex of $G_{d,n}$ belongs to at least $(1-1/r)|\mathcal{C}|$ elements of $\mathcal{C}$ and so there exists $C \in \mathcal{C}$ containing at least $(1-1/r)|V(G_{d,n})| = (1-1/r)n^d$ vertices of $G_{d,n}$. Fix such a $C$ and let $G$ be the subgraph of $G_{d,n}$ induced by $C$. We claim that $\tin(G) \geq f(r)$.

Observe first that, for each $v \in V(G_{d,n})$, $d \leq d_{G_{d,n}}(v) \leq 2d$. Hence, $2|E(G_{d,n})| = \sum_{v \in V(G_{d,n})}d_{G_{d,n}}(v) \geq d\cdot n^d$. Consider now the graph $G'$ obtained from $G_{d,n}$ by deleting the vertex set $C$. Clearly, $G'$ has at most $n^d/r$ vertices. Since deleting a vertex from $G_{d,n}$ decreases the number of edges of the resulting graph by at most $2d$, we have that $|E(G)|\geq |E(G_{d,n})| - 2d|V(G')|$, from which $\sum_{v \in V(G)}d_{G}(v) \geq d\cdot n^d - 2\cdot 2d\cdot n^d/r = d\cdot n^d(1 - 4/r)$. Therefore, the average degree of $G$ is at least $d(1-4/r)$ and so $\tw(G) \geq \frac{r-4}{2r}d$, for example by \cite[Corollary~1]{CS05}. This implies that every tree decomposition of $G$ has a bag of size at least $\frac{r-4}{2r}d + 1 \geq R(3, f(r))$ and, since $G$ is triangle-free, it follows that $\tin(G) \geq f(r)$.
\end{proof}

However, as we shall see in the next section, any family of bounded-dimensional grids is fractionally $\tin$-fragile. 


\section{Intersection graphs of fat objects}\label{fatA}

Let $d \geq 2$ be an arbitrary but fixed integer. The main result of this section is the following: The class of intersection graphs of $c$-fat collections of objects in $\mathbb{R}^d$, defined in \Cref{fatcomp}, is efficiently fractionally $\tin$-fragile (Result \ref{item:third}). More precisely, we show the following.

\begin{theorem}\label{treealphafatA} Let $\mathcal{O}$ be a $c$-fat collection of objects in $\mathbb{R}^d$ and let $G$ be its intersection graph. For each $r_0 > 1$, let $f(r_0) = 2\Big\lceil \frac{1}{1-\big(1-\frac{1}{r_0}\big)^{\frac{1}{d}}}\Big\rceil$. Then we can compute in linear time a $(1-1/r_0)$-general cover $\mathcal{C}$ of $G$ of size at most $(f(r_0)/2-1)^d$. Moreover, for each $C \in \mathcal{C}$, we can compute in linear time a tree decomposition $\mathcal{T} = (T,\{X_t\}_{t\in V(T)})$ of $G[C]$, with $|V(T)| \leq |V(G)|+1$, such that $\alpha(\mathcal{T}) \leq cf(r_0)^{2d}$. 
\end{theorem}

Before proving \Cref{treealphafatA}, we outline the idea in the case of disk graphs in $\mathbb{R}^2$. Suppose first that we are trying to find a $(1-1/r_0)$-general cover $\mathcal{C}$ of bounded tree-independence number of a unit disk graph. We can build each element of the cover starting from an appropriate grid in the plane as follows. Suppose that $\mathcal{H}(y)$ is a grid in $\mathbb{R}^2$, indexed by some $y \in \mathbb{R}^2$, splitting the plane into a collection $\mathcal{B}$ of squares of side length $2r_0$. We first discard all disks intersecting $\mathcal{H}(y)$. The vertices corresponding to the remaining disks will form the element $C(y)$ of the cover. We can obtain a tree decomposition for the subgraph induced by $C(y)$ as follows. We add a node for each square $B \in \mathcal{B}$ and associate to this node a bag containing precisely the vertices whose corresponding disks lie in $B$. We then connect the nodes appropriately to obtain a tree. The resulting tree decomposition will have small independence number since, inside each square $B$, there are at most $4r_{0}^2$ pairwise non-intersecting disks. Shifting the grid $\mathcal{H}(y)$ around the plane via the vector $y$ and proceeding as above will ensure that every vertex of the unit disk graph is covered by sufficiently many elements. 

The situation is more challenging if disks have different radiuses. When both large and small disks occur, if the grids are too dense (i.e., they divide the plane into very small squares), then large disks will not belong to most elements of the cover, whereas if the grids are too sparse (i.e., they divide the plane into very large squares), then there might be too many pairwise non-intersecting small disks inside each square. To resolve this, we use an idea from \cite[Theorem~4]{DGLTT22}. Specifically, we sort disks into different ranks according to their radius, so that the larger the radius the smaller the rank. Large disks will be ``covered'' by sparse grids, whereas small disks will be ``covered'' by dense grids. For each possible value $i$ of the rank, we will consider grids of rank $i$ arising in a quadtree-like manner from a fixed rank-$0$ grid (a sparesest grid), and we will discard rank-$i$ disks intersecting rank-$i$ grids. The vertices corresponding to the remaining disks will form an element $C(y)$ of the cover. We will then add a node for each square $B_i$ induced by the rank-$i$ grid, and associate to this node a bag containing precisely the vertices whose corresponding disks intersect $B_i$ and have rank at most $i$. Finally, for each node $t_i$ corresponding to a rank-$i$ square $B_i$, we will add the edge $t_it_j$ if $t_j$ correspond to the rank-$j$ square $B_j$, with $j > i$, such that $B_j$ is contained in $B_i$. As before, we will shift the grids around the plane via $y$ to ensure that we obtain indeed a general cover. 

\begin{proof}[Proof of \Cref{treealphafatA}] In this proof, $[n]$ denotes the set $\{0, 1, \ldots, n\}$. Fix an arbitrary $r_0 > 1$. In the following, for ease of notation, we simply let $r:= f(r_0) = 2\Big\lceil \frac{1}{1-\big(1-\frac{1}{r_0}\big)^{\frac{1}{d}}}\Big\rceil$, and denote by $O_v$ the object corresponding to the vertex $v \in V(G)$. By possibly rescaling, we may assume that each object in the collection $\mathcal{O}$ has size at most $1$. For each $v \in V(G)$, define the \textit{rank} of the object $O_v$ as the quantity $\rk(O_v) = \lfloor \log_{\frac{1}{r}}s(O_v)\rfloor$. Let $k_0 = \max_{v \in V(G)}\rk(O_v)$. For each $0 \leq i \leq k_0$, $1 \leq j \leq d$ and $y = (y_1,\ldots,y_d) \in [\frac{r}{2}-1]^d$, let $\mathcal{H}^i_j(y)$ be the set of points in $\mathbb{R}^d$ whose $j$-th coordinate is equal to $m_j(\frac{1}{r})^{i-1}+y_j\sum_{k={i}}^{k_0+1}(\frac{1}{r})^{k}$, for some $m_j \in \mathbb{Z}$, and let $\mathcal{H}^i(y) = \bigcup_{1\leq j \leq d}\mathcal{H}^i_j(y)$. Moreover, let $V^i = \{v \in V(G) : \rk(O_v)=i\}$, $C^i(y) = \{v \in V^i: O_v \cap \mathcal{H}^i(y) = \varnothing\}$, and $C(y) = \bigcup_{0 \leq i \leq k_0}C^i(y)$. See \Cref{Fig:refinedgrid}.

\begin{figure}
\centering
\begin{tikzpicture}
\draw (0,0) -- (0,-.1);
\draw (2,0) -- (2,-.1);
\draw (4,0) -- (4,-.1);
\draw (6,0) -- (6,-.1);
\draw (8,0) -- (8,-.1);

\draw (0,0) -- (-.1,0);
\draw (0,2) -- (-.1,2);
\draw (0,4) -- (-.1,4);
\draw (0,6) -- (-.1,6);
\draw (0,8) -- (-.1,8);

\draw[very thick] (0,0) -- (0,8.5);
\draw[very thick] (2,0) -- (2,8.5);
\draw[very thick] (4,0) -- (4,8.5);
\draw[very thick] (6,0) -- (6,8.5);
\draw[very thick] (8,0) -- (8,8.5);

\draw[very thick] (0,0) -- (8.5,0);
\draw[very thick] (0,2) -- (8.5,2);
\draw[very thick] (0,4) -- (8.5,4);
\draw[very thick] (0,6) -- (8.5,6);
\draw[very thick] (0,8) -- (8.5,8);

\node[draw= none] at (0,-.5) {\small $y'_1$};
\node[draw= none] at (2,-.5) {\small $y'_1 + (\frac{1}{r})^{i-1}$};
\node[draw= none] at (4,-.5) {\small $y'_1 + 2(\frac{1}{r})^{i-1}$};
\node[draw= none] at (6,-.5) {\small $y'_1 + 3(\frac{1}{r})^{i-1}$};
\node[draw= none] at (8,-.5) {\small $y'_1 + 4(\frac{1}{r})^{i-1}$};

\node[draw= none] at (-.5,0) {\small $y'_2$};
\node[draw= none] at (-1,2) {\small $y'_2 + (\frac{1}{r})^{i-1}$};
\node[draw= none] at (-1,4) {\small $y'_2 + 2(\frac{1}{r})^{i-1}$};
\node[draw= none] at (-1,6) {\small $y'_2 + 3(\frac{1}{r})^{i-1}$};
\node[draw= none] at (-1,8) {\small $y'_2 + 4(\frac{1}{r})^{i-1}$};

\foreach \i in {0,...,34}{
\pgfmathsetmacro{\x}{0.25*\i}
\draw[red,opacity=.2] (\x,0) -- (\x,8.5);
\draw[red,opacity=.2] (0,\x) -- (8.5,\x);
}

\draw[thick] (4.1,2.1) rectangle (5.9,3.9);
\node[draw=none] at (5,3) {\small $B^i(y,(2,1))$};
\end{tikzpicture}
\caption{Grid used in the proof of \Cref{treealphafatA}: $\mathcal{H}^i(y)$ (black) and $\mathcal{H}^{i+1}(y)$ (red) in dimension $2$, where $y'_j = y_j \sum_{k=i}^{k_0+1} (\frac{1}{r})^k$.}\label{Fig:refinedgrid}
\end{figure}
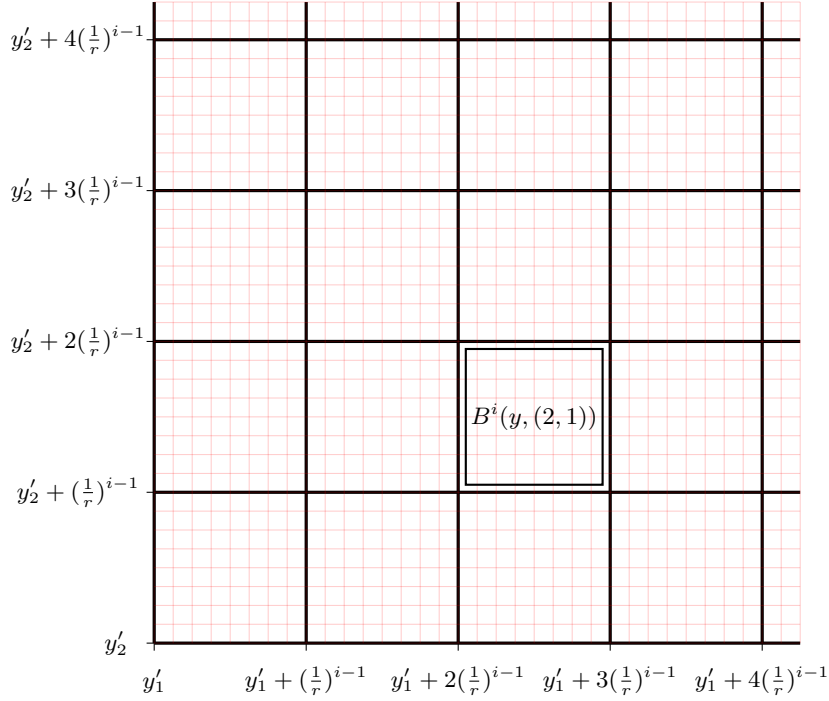

\begin{claim}
$\mathcal{C} = \{C(y) : y \in [\frac{r}{2}-1]^d\}$ is a $(1-1/r_0)$-general cover of $G$ of size $(f(r_0)/2-1)^d$.
\end{claim}

\begin{claimproof} For each $1\leq j \leq d$, let $e_j$ be the unit vector in $\mathbb{R}^d$ whose $j$-th coordinate is $1$. Let $v\in V(G)$ and let $i = \rk(O_v)$. Then $v \in C(y)$ for some $y \in [\frac{r}{2}-1]^d$ if and only if $O_v$ does not intersect $\mathcal{H}^i(y)$, and the latter happens if and only if $O_v$ does not intersect $\mathcal{H}^i_j(y)$ for any $1 \leq j \leq d$. Note that $\mathcal{H}^i_j(y)$ is a collection of hyperplanes in $\mathbb{R}^d$, which are orthogonal to the $j$-th axis and at pairwise distance at least $(\frac{1}{r})^{i-1} = r(\frac{1}{r})^{i}$. Moreover, for any $1 \leq j \leq d$, 
$\mathcal{H}^i_{j}(y + e_j)$ can be obtained by shifting $\mathcal{H}^i_j(y)$ along the $j$-th axis of a quantity $\sum_{k={i}}^{k_0+1}(\frac{1}{r})^{k}$, and it is easy to see that, since $r \geq 2$, we have $(\frac{1}{r})^{i} < \sum_{k={i}}^{k_0+1}(\frac{1}{r})^{k} < 2(\frac{1}{r})^{i}$. On the other hand, since $\rk(O_v)=i$, it follows from the definition of rank that $O_v$ can be enclosed in a box of size $(\frac{1}{r})^{i}$. We now count the number of points $y \in [\frac{r}{2}-1]^d$ such that $O_v$ does not intersect $\mathcal{H}^i(y)$. For fixed $1 \leq j \leq d$, the previous observations imply that there is at most one value of $y_j$ for which a point $y \in [\frac{r}{2}-1]^d$ is such that $\mathcal{H}^i_j(y) \cap O_v \neq \varnothing$, so at least $\frac{r}{2} - 1$ values of $y_j$ for which $y$ is such that $\mathcal{H}^i_j(y) \cap O_v = \varnothing$. Therefore, there are at least $(\frac{r}{2}-1)^d$ points $y \in [\frac{r}{2}-1]^d$ such that $\mathcal{H}^i(y)$ does not intersect $O_v$. Since the set $[\frac{r}{2}-1]^d$ has size $(\frac{r}{2})^d$, the proportion of elements of $\mathcal{C}$ containing $v$ is at least $\frac{(\frac{r}{2}-1)^d}{(\frac{r}{2})^d} = (1-\frac{2}{r})^d \geq ((1-\frac{1}{r_0})^{\frac{1}{d}})^d = 1-\frac{1}{r_0}$, as claimed.
\end{claimproof}

Note that, for fixed $i$ and $y$ as above, the collection of hyperplanes $\mathcal{H}^i(y)$ splits the space into boxes of size $(\frac{1}{r})^{i-1}$. We now consider these boxes. For a vector $m = (m_1, \ldots, m_d) \in \mathbb{Z}^d$, consider the box $B^i(y,m) = \{(x_1,\ldots,x_d) \in \mathbb{R}^d : m_j(\frac{1}{r})^{i-1}+y_j\sum_{k={i}}^{k_0+1}(\frac{1}{r})^{k} < x_j < (m_j+1)(\frac{1}{r})^{i-1}+y_j\sum_{k={i}}^{k_0+1} (\frac{1}{r})^{k}, \ \mbox{for every} \ 1 \leq j \leq d\}$. 

\begin{claim}\label{subboxA} For fixed $y \in [\frac{r}{2}-1]^d$ and any $0 \leq i' < i \leq k_0$, $\mathcal{H}^{i'}(y) \subseteq \mathcal{H}^{i}(y)$. Moreover, for any $m \in \mathbb{Z}^d$, the box $B^{i}(y,m)$ is completely contained in a box of the form $B^{i'}(y,m')$ for exactly one vector $m' \in \mathbb{Z}^{d}$. 
\end{claim}

\begin{claimproof} Let $y = (y_1, \ldots, y_d)$. Let $x=(x_1,\ldots,x_d) \in \mathcal{H}^{i'}(y)$. There exists $1 \leq j \leq d$ such that $x \in \mathcal{H}^{i'}_j(y)$. Then there exists $m_j \in \mathbb{Z}$ such that 
\begin{align*}
x_j &=  m_j\Big(\frac{1}{r}\Big)^{i'-1}+y_j\sum_{k={i'}}^{k_0+1}\Big(\frac{1}{r}\Big)^{k} = m_j\Big(\frac{1}{r}\Big)^{i'-i}\Big(\frac{1}{r}\Big)^{i-1} +y_j\sum_{k={i'}}^{i-1}\Big(\frac{1}{r}\Big)^{k} +y_j\sum_{k={i}}^{k_0+1}\Big(\frac{1}{r}\Big)^{k} \\
&= \Big(m_j\Big(\frac{1}{r}\Big)^{i'-i}+y_j\sum_{k={i'}}^{i-1}\Big(\frac{1}{r}\Big)^{k-(i-1)}\Big)\Big(\frac{1}{r}\Big)^{i-1} + y_j\sum_{k={i}}^{k_0+1}\Big(\frac{1}{r}\Big)^{k}.
\end{align*}
Since the coefficient of $(\frac{1}{r})^{i-1}$ is an integer, we conclude that $x \in \mathcal{H}^{i}_j(y)$ and so $\mathcal{H}^{i'}(y) \subseteq \mathcal{H}^{i}(y)$.

To prove the remaining statement simply recall that, for $m, m' \in \mathbb{Z}^d$, $B^{i}(y,m)$ is one of the boxes induced by $\mathbb{R}^d \setminus \mathcal{H}^{i}(y)$ and $B^{i'}(y,m')$ is one of the boxes induced by $\mathbb{R}^d \setminus \mathcal{H}^{i'}(y)$. Since $\mathcal{H}^{i}(y)$ is a refinement of $\mathcal{H}^{i'}(y)$, for any box of the form $B^{i}(y,m)$ there must be exactly one box of the form $B^{i'}(y,m)$ containing it.
\end{claimproof}

We now construct a tree decomposition of $G[C(y)]$, for each element $C(y)$ of the $(1-1/r_0)$-general cover $\mathcal{C}$ defined above. Therefore, fix $y \in [\frac{r}{2}-1]^d$. For each $0 \leq i \leq k_0$ and $m \in \mathbb{Z}^d$, let $A^i(y,m) = \{v \in C^i(y): O_v \cap B^i(y,m) \neq \varnothing\}$ and let $X_{t^i(y,m)} = \bigcup_{0\leq k \leq i}\{v \in C^k(y) : O_v \cap B^i(y,m) \neq \varnothing\}$. In words, $X_{t^i(y,m)}$ is the set of vertices corresponding to objects of rank at most $i$ in $C(y)$ and intersecting the box $B^i(y,m)$. For each pair $(i,m)$, build a node $t^i(y,m)$ if $A^i(y,m) \neq \varnothing$ and associate to it the set $ X_{t^i(y,m)}$, which will be the corresponding bag in the tree decomposition we are building. We say that $t^{i_1}(y,{m_1})$ is a \textit{parent} of $t^{i_2}(y,m_2)$ if the following conditions are satisfied: $i_1 < i_2$, $B^{i_1}(y,m_1) \supseteq B^{i_2}(y,m_2)$ and, among all pairs satisfying these two conditions, $(i_1,m_1)$ has largest value of the first entry. Observe that, by \Cref{subboxA}, each node $t^{i_2}(y,{m_2})$ has at most one parent. For each pair of nodes $t^{i_1}(y,{m_1})$, $t^{i_2}(y,m_2)$ such that $t^{i_1}(y,{m_1})$ is a parent of $t^{i_2}(y,m_2)$, we then add the edge $t^{i_1}(y,{m_1})t^{i_2}(y,m_2)$. We claim that the resulting graph $F(y)$ is acyclic. Suppose, to the contrary, that it contains a cycle with vertices $t^{i_0}(y,{m_0}), t^{i_1}(y,m_1), \ldots, t^{i_{\ell-1}}(y,m_{\ell-1})$ in cyclic order. Without loss of generality, $t^{i_0}(y,{m_0})$ is a parent of $t^{i_1}(y,m_1)$. This implies that, for each $k$, $t^{i_k}(y,{m_k})$ is a parent of $t^{i_{k+1}}(y,m_{k+1})$ (indices modulo $\ell$). Therefore, by definition of parent, $i_0 < i_1 < \cdots < i_{\ell-1}$ and $i_{\ell-1} < i_0$, a contradiction.     

We then glue the components of $F(y)$ into a tree by adding a node $t^{-1}$ and making $t^{-1}$ adjacent to an arbitrary node of each component of $F(y)$. Let the resulting tree be $T(y)$. Observe that $|V(T(y))| \leq |V(G)|+1$. Indeed, $t^i(y,m)$ is a node of $T(y)$ only if $A^i(y,m) \neq \varnothing$ and, for fixed $y$, $A^{i_1}(y,m_1) \subseteq V(G)$ is disjoint from $A^{i_2}(y,m_2) \subseteq V(G)$. Let $\mathcal{T} = (T(y),\{X_{t^i(y,m)}\}_{t^i(y,m)\in V(T(y))})$, where we assign the empty bag to the node $t^{-1}$. Clearly, $\mathcal{T}$ can be computed in linear time.

\begin{claim} $\mathcal{T} = (T(y),\{X_{t^i(y,m)}\}_{t^i(y,m)\in V(T(y))})$ is a tree decomposition of $G[C(y)]$. 
\end{claim}

\begin{claimproof} We first check that (T1) holds. Let $v \in C(y)$ and suppose that $\rk(O_v) = i$. Then $O_v$ intersects one of the boxes $B^{i}(y,m)$, for some $m \in \mathbb{Z}^d$, and so $v \in A^i(y,m)$. Therefore, $t^i(y,m) \in V(T(y))$ and $v \in X_{t^i(y,m)}$. 

We now check that (T2) holds. Let $u,v \in C(y)$ such that $uv \in E(G)$. Then $O_u \cap O_v \neq \varnothing$ and let $x=(x_1,\ldots,x_d)$ be a point in this intersection. Without loss of generality, $\rk(O_u) \leq \rk(O_v)=i$. Since $v \in C^i(y)$, $x \not \in \mathcal{H}^i(y)$. This implies that $x$ is contained in a box $B^i(y,m)$, for some $m\in\mathbb{Z}^d$, and so $v \in A^i(y,m)$, $t^i(y,m) \in V(T(y))$ and $\{u,v\} \subseteq X_{t^i(y,m)}$.  

We finally check that (T3) holds. For $v \in C(y)$, let $T(y)_v$ be the subgraph of $T(y)$ induced by the set of nodes of $T(y)$ whose bag contains $v$. Let $v \in C(y)$ and suppose that $\rk(O_v) = i$. Observe first that there is a unique $m \in \mathbb{Z}^d$ such that $v \in X_{t^{i}(y,m)}$, or else $O_v \cap \mathcal{H}^i(y) \neq \varnothing$ and $v \not \in C(y)$. Observe now that, by definition, $v \not \in X_{{t^{i_1}}(y,m_1)}$ for any $i_1 < i$ and $m_1 \in \mathbb{Z}^d$. Suppose finally that $v \in X_{t^{i_1}(y,m_1)}$, for some $i_1 > i$ and $m_1\in \mathbb{Z}^d$. Then $O_v$ intersects $B^{i_1}(y,m_1)$ and, by \Cref{subboxA}, there is a unique $m' \in \mathbb{Z}^d$ such that $B^{i_1}(y,m_1)$ is completely contained in $B^{i}(y,m')$ (it is easy to see that $m' = m$). Hence, ${t^{i_1}(y,m_1)}$ must have a parent, say ${t^{i_2}(y,m_2)}$ for some $i_2$ such that $i_1 > i_2 \geq i$ and $m_2 \in \mathbb{Z}^d$. This means that $B^{i_2}(y,m_2) \supseteq B^{i_1}(y,m_1)$, and so $v \in X_{t^{i_2}(y,m_2)}$. We then deduce inductively that there must be a path from $t^{i_1}(y,m_1)$ to $t^{i}(y,m)$ in $T(y)_v$. Therefore, $T(y)_v$ is connected.
\end{claimproof} 

\begin{claim} $\alpha(\mathcal{T}) \leq cr^{2d}$.
\end{claim}

\begin{claimproof} Fix an arbitrary node $t^i(y,m)$ of $T(y)$. We bound the independence number of the subgraph of $G[C(y)]$ induced by the bag $X_{t^i(y,m)}$. Observe first that, for any $v \in X_{t^i(y,m)}$, $O_v$ intersects $B^i(y,m)$, which is a box of side length $(\frac{1}{r})^{i-1}$. Consider a collection $\mathcal{B}$ of $r^{2d}$ generic closed boxes in $\mathbb{R}^d$ of side length $(\frac{1}{r})^{i+1}$ and such that their union is exactly $B^i(y,m)$. Let $P \subseteq X_{t^i(y,m)}$ be an independent set of $G[C(y)]$ and let $\mathcal{P} = \{O_v : v \in P\}$ be the corresponding subcollection of $\mathcal{O}$ of pairwise non-intersecting objects. For each $v \in P$, $\rk(O_v) \leq i$ and so $s(O_v) \geq (\frac{1}{r})^{i+1}$. Moreover, $O_v$ intersects at least one box from $\mathcal{B}$. Therefore, since the collection $\mathcal{O}$ is $c$-fat, there are at most $c$ objects in $\mathcal{P}$ intersecting any fixed box in $\mathcal{B}$, and so $|P| \leq cr^{2d}$, as claimed.
\end{claimproof}

This concludes the proof of \Cref{treealphafatA}.
\end{proof}

We note some consequences of \Cref{treealphafatA}. First, there exist fractionally $\tin$-fragile classes of unbounded local tree-independence number.

\begin{corollary}\label{fragileunboundedlocal} The class of disk graphs is fractionally $\tin$-fragile but has unbounded local tree-independence number. 
\end{corollary}

\begin{proof} Let $\mathcal{G}$ be the class of intersection graphs of disks in $\mathbb{R}^2$. By \Cref{treealphafatA}, $\mathcal{G}$ is fractionally $\tin$-fragile. Let now $G_n$ be the graph obtained from the $n\times n$-grid graph by adding a dominating vertex. By the proof of \Cref{equivlayeredA}, the class $\mathcal{G}' = \{G_n : n \in \mathbb{N}\}$ has unbounded local tree-independence number. We conclude by observing that $\mathcal{G'} \subseteq \mathcal{G}$.
\end{proof}

It is easy to see that every $d$-dimensional grid of side length $n$ can be realized as the intersection graph of a family of unit balls in $\mathbb{R}^d$. Therefore, \Cref{treealphafatA} and \Cref{gridsunbounded} immediately imply the following dichotomy.

\begin{corollary}\label{gridsdich} Let $I, J \subseteq \mathbb{N}$, with $J \neq \{1\}$, and let $\mathcal{G} = \{G_{d,n} : d \in I, n \in J\}$ be a family of grids. Then $\mathcal{G}$ is fractionally $\tin$-fragile if and only if $I$ is finite. 
\end{corollary}

Note that \Cref{gridsdich} cannot be strengthened to bounded layered tree-independence number. Indeed, even though $2$-dimensional grids have bounded layered tree-independence number (since, for example, they are planar), the family $\{G_{3,n} : n \in \mathbb{N}\}$ of $3$-dimensional grids has unbounded layered tree-independence number, as observed in \Cref{sec:layeredlemmas}. 

In \Cref{layeredtreealgoA} we showed that the class of graphs of bounded layered tree-independence number is closed under taking odd powers. It is not clear whether odd powers of graphs from \textit{any} fractionally $\tin$-fragile class form a fractionally $\tin$-fragile class. Nevertheless, we now show that another fractionally $\tin$-fragile class, namely that of intersection graphs of $c$-fat collections of objects in $\mathbb{R}^d$, is closed under taking odd powers.   

\begin{theorem}\label{c-fatoddpower} Let $G$ be the intersection graph of a $c$-fat collection of objects $\mathcal{O}$ in $\mathbb{R}^d$. Let $k$ be a positive integer. Then $G^{2k+1}$ is the intersection graph of a $(3^d(2k+1)^dc)$-fat collection of objects in $\mathbb{R}^d$. Moreover, given $G$ and $\mathcal{O}$, such a collection can be computed in time polynomial in $|V(G)|$.
\end{theorem}

\begin{proof} For each $v \in V(G)$, denote by $O_v$ the object in $\mathcal{O}$ corresponding to the vertex $v \in V(G)$. For each non-negative integer $j$, let $N^j[O_v]$ be the subset of objects from $\mathcal{O}$ corresponding to vertices in $N^{j}_{G}[v]$, and let $O^j_v = \bigcup_{O \in N^j[O_v]}O$. Note that $O^j_v$ is path-connected and compact in $\mathbb{R}^d$.  

Consider the collection of objects $\mathcal{O}^k = \{O^k_v: v \in V(G)\}$. We claim that its intersection graph $G'$ is isomorphic to $G^{2k+1}$. The map sending each $v \in V(G)$ to $O^k_v$ gives a natural bijection between the vertex sets of $G^{2k+1}$ and $G'$. Consider now two distinct vertices $u,v \in G'$, corresponding to $O^k_u$ and $O^k_v$, respectively. We have that $u$ and $v$ are adjacent in $G'$ if and only if there exist $O_x \subseteq O^k_u$ and $O_y \subseteq O^k_v$ with $O_x \cap O_y \neq \varnothing$, for some $x,y\in V(G)$ (note that it might be $x = y$). But $d_G(u,x)$ and $d_G(v,y)$ are both at most $k$, from which $d_G(u,v) \leq 2k+1$, and so $u$ and $v$ are adjacent in $G^{2k+1}$. Conversely, if $u$ and $v$ are distinct vertices adjacent in $G^{2k+1}$, then $d_G(u,v) \leq 2k+1$. Take a shortest $u, v$-path $P$ in $G$. Let $P_u$ be the subpath of $P$ with endpoint $u$ and containing vertices at distance at most $k$ from $u$. Similarly, let $P_v$ be the subpath of $P$ with endpoint $v$ and containing vertices at distance at most $k$ from $v$. Then $\bigcup_{x \in V(P_v)}O_x \subseteq O^k_v$, $\bigcup_{x \in V(P_u)}O_x \subseteq O^k_u$ and $\bigcup_{x \in V(P_v)}O_x \cap \bigcup_{x \in V(P_u)}O_x \neq \varnothing$, from which $O^k_u \cap O^k_v \neq \varnothing$.         

We now show that the collection $\mathcal{O}^k$ is $(3^d(2k+1)^dc)$-fat. Suppose, to the contrary, that for some $r$ there exists a closed box $R$ of side length $r$ intersecting a subcollection $\mathcal{P}\subseteq \mathcal{O}^k$ of more than $3^d(2k+1)^dc$ pairwise non-intersecting objects of size at least $r$. Let $\mathcal{P} = \{O^k_{v_1},\ldots, O^k_{v_m}\}$, for some $m > 3^d(2k+1)^dc$. Let $R'$ be the box of side length $3r$ with the same center as $R$. 

\begin{claim}\label{powerfat} For each $i \in \{1, \ldots, m\}$, there exists $A_i \in N^k[O_{v_i}]$ of size at least $\frac{r}{2k+1}$ and such that $A_i \cap R' \neq \varnothing$.
\end{claim}

\begin{claimproof}[Proof of \Cref{powerfat}] Suppose, to the contrary, that there exists $i \in \{1, \ldots, m\}$ such that every object in $N^k[O_{v_i}]$ either does not intersect $R'$ or has size less than $\frac{r}{2k+1}$. We distinguish two cases. 

\textbf{Case I:} Every object in $N^k[O_{v_i}]$ intersects $R'$. By the assumption from the previous paragraph, every object in $N^k[O_{v_i}]$ has size less than $\frac{r}{2k+1}$. In particular, $O_{v_i}$ is contained in a box of size less than $\frac{r}{2k+1}$. We show inductively that, for each $j \in \{0, \ldots, k\}$, $O^j_{v_i}$ is contained in a box of size less than $(2j+1)\frac{r}{2k+1}$. The base case $j=0$ follows from the previous observation. Take now $j \in \{0, \ldots, k-1\}$ and suppose that $O^j_{v_i}$ is contained in a box $X$ of size less than $(2j+1)\frac{r}{2k+1}$. Since every object in $N^k[O_{v_i}]$ has size less than $\frac{r}{2k+1}$, every object in $N^{j+1}[O_{v_i}] \setminus N^{j}[O_{v_i}]$ is contained in a box of size less than $\frac{r}{2k+1}$, and moreover every such object must intersect the box $X$. Therefore, $O^{j+1}_{v_i}$ is contained in a box of size less than $(2(j+1)+1)\frac{r}{2k+1}$. Taking $j = k$, we deduce that $O^k_{v_i}$ is contained in a box of size less than $(2k+1)\frac{r}{2k+1} = r$, contradicting the fact that $O^k_{v_i}$ has size at least $r$.

\textbf{Case II:} There exists $B \in N^k[O_{v_i}]$ not intersecting $R'$. Let $u \in V(G)$ be the vertex corresponding to $B$. Then, as $N^{k}[O_{v_i}]$ intersects $R$, there exists a path $u_1\cdots u_{p}$ in $G$ from $u_1 = u$ to some $u_p$ of length at most $2k$ satisfying the following property: For each $j \in \{1, \ldots, p\}$, $O_{u_j}$ belongs to $N^{k}[O_{v_i}]$ and $O_{u_p}$ intersects $R$. Since $B$ does not intersect $R'$ and $O_{u_p}$ intersects $R\subseteq R'$, there exists an index $\ell \in \{2, \ldots, p\}$ such that $O_{u_{\ell-1}}$ does not intersect $R'$ whereas $O_{u_j}$ intersects $R'$ for all $j \geq \ell$. By assumption, each $O_{u_j}$ with $j \geq \ell$ has size less than $\frac{r}{2k+1}$. Then, using a similar argument as in Case I, it is easy to see that $O = \bigcup_{j=\ell}^{p}O_{u_j}$ is contained in a box of size less than $r$. But $O$ intersects $R$ (as $O_{u_p}$ does) and so $O$ is entirely contained in $R'$. Since $O_{u_{\ell-1}}$ intersects $O_{u_\ell}$ and the latter is contained in $R'$, we conclude that $O_{u_{\ell-1}}$ intersects $R'$, a contradiction. 
\end{claimproof}

By the previous claim, for each $i \in \{1, \ldots, m\}$, $N^k[O_{v_i}]$ contains an object $A_i$ of size at least $\frac{r}{2k+1}$ which intersects the box $R'$ of side length $3r$. Since $\mathcal{P}$ consists of pairwise non-intersecting objects, the same holds for the family $\mathcal{A} = \{A_i: i \in \{1, \ldots, m\}\}$. Moreover, again by assumption, $|\mathcal{A}| > 3^d(2k+1)^dc$. Observe now that the box $R'$ can be decomposed into a union of $3^d(2k+1)^d$ sub-boxes, each of side length $\frac{r}{2k+1}$. Then, by the pigeonhole principle, there must be at least $c+1$ distinct objects in $\mathcal{A}$ all intersecting the same sub-box of side length $\frac{r}{2k+1}$, contradicting the fact that $\mathcal{O}$ is $c$-fat. 
\end{proof}

We conclude this section by observing that even powers of fractionally $\tin$-fragile classes need not be fractionally $\tin$-fragile. Indeed, similarly to the proof of \Cref{countereven}, and using \Cref{bicliqueA}, we obtain the following. 

\begin{lemma} Fix an even $k \in \mathbb{N}$. The class $\mathcal{G}$ of chordal graphs is fractionally $\tin$-fragile but the class $\{G^k : G \in \mathcal{G}\}$ is not.  
\end{lemma}


\section{PTASes}\label{sec:ptasesA}
In this section we prove Results \ref{item:first}, \ref{item:second}, \ref{item:fourth}. Before providing the corresponding PTASes, we highlight some examples of problems captured by the framework of $(c, h, \psi)$-\textsc{Max Weight Induced Subgraph}, which is addressed by Result \ref{item:first} and which was defined in \Cref{sec:mainres}. To this purpose, it is useful to recall the following well-known observations (see, e.g., \cite{Bel13}): the subgraph and minor containment relations, as well as the property of being $q$-colorable, for fixed $q$, are all expressible in $\mathsf{MSO}_2$. This immediately implies that problems such as \textsc{Max Weight Independent Set}, \textsc{Max Weight Induced Forest} and \textsc{Max Weight Induced Planar Subgraph} fall in the framework. The same holds for \textsc{Max Weight Induced $q$-Colorable Subgraph}, which is the problem that, for a fixed positive integer $q$ and given a vertex-weighted graph $G$, asks to find a maximum-weight subset $F \subseteq V(G)$ such that $G[F]$ is $q$-colorable. The unweighted case $q = 2$ is known as \textsc{Max Bipartite Subgraph} \cite{JMMR20}. 

Another example is the following. For a positive integer $k$, let $\mathcal{H}_k$ be a set of connected graphs which are contained in $K_k$. For fixed $\mathcal{H}_k$, \textsc{Max $\mathcal{H}_k$-Free Node Set} is the problem that, given a graph $G$, asks to find a maximum-size subset $F \subseteq V(G)$ such that $G[F]$ is $\mathcal{H}_k$-subgraph-free \cite{LSH18}. A notable special case is \textsc{Max} $k$-\textsc{Dependent Set} \cite{DKL93}, the problem of finding a maximum-size induced subgraph of maximum degree at most $k$ (the case $k=1$ is also known as \textsc{Dissociation Set} \cite{ODF11,Yan81}). \textsc{Max $\mathcal{H}_k$-Free Node Set} corresponds to $(k-1, 1, \psi)$-\textsc{Max Weight Induced Subgraph}, where $\psi$ is the $\mathsf{MSO}_2$ formula expressing the property that none of the finitely many graphs in $\mathcal{H}_k$ is a subgraph of $G[F]$. We refer the reader to \cite{FTV15} for several other examples of problems which are special cases of $(c, h, \psi)$-\textsc{Max Weight Induced Subgraph}.

We now define \textsc{Max Weight Independent Packing}, which is addressed by Result \ref{item:second}. Given a graph $G$ and a finite family $\mathcal{H} = \{H_j\}_{j\in J}$ of connected non-null subgraphs of $G$, an \textit{independent $\mathcal{H}$-packing} in $G$ is a subfamily $\mathcal{H}' = \{H_i\}_{i\in I}$ of subgraphs from $\mathcal{H}$ (i.e., $I \subseteq J$) that are at pairwise distance at least $2$, that is, they are vertex-disjoint and there is no edge between any two of them. If the subgraphs in $\mathcal{H}$ are equipped with a weight function $w\colon J \rightarrow \mathbb{Q}_{+}$ assigning weight $w_j$ to each subgraph $H_j$, the \textit{weight} of an independent $\mathcal{H}$-packing $\mathcal{H}' = \{H_i\}_{i\in I}$ in $G$ is $\sum_{i\in I}w_i$. Fix now $h \in \mathbb{N}$. Given a graph $G$, a finite family $\mathcal{H} = \{H_j\}_{j\in J}$ of connected non-null subgraphs of $G$ such that $|V(H_j)| \leq h$ for each $j \in J$, and a weight function $w\colon J \rightarrow \mathbb{Q}_{+}$ on the subgraphs in $\mathcal{H}$, the problem \textsc{Max Weight Independent Packing} asks to find an independent $\mathcal{H}$-packing in $G$ of maximum weight. In the special case when $\mathcal{F}$ is a \textit{fixed} finite family of connected non-null graphs and $\mathcal{H}$ is the set of all subgraphs of $G$ isomorphic to a member of $\mathcal{F}$, the problem is called \textsc{Max Weight Independent $\mathcal{F}$-Packing} and is a common generalization of several problems, among which: \textsc{Independent $\mathcal{F}$-Packing} \cite{CH06}, \textsc{Max Weight Independent Set} ($\mathcal{F} = \{K_1\}$), \textsc{Max Weight Induced Matching} ($\mathcal{F} = \{K_2\}$), \textsc{Dissociation Set} ($\mathcal{F} = \{K_1, K_2\}$ and the weight function assigns to each subgraph $H_j$ the weight $|V (H_j)|$). Observe that \textsc{Max Weight Independent $\mathcal{F}$-Packing} is in fact a special case of $(c, h, \psi)$-\textsc{Max Weight Induced Subgraph}, where we take $c = h = \max_{H \in \mathcal{F}}|V(H)|$ and $\psi$ is the $\mathsf{MSO}_2$ formula expressing the property that every connected component of $G[F]$ belongs to $\mathcal{F}$.

\textsc{Max Weight Independent Packing} has a natural generalization, addressed by Result \ref{item:fourth}, and which we now define. For fixed positive integers $d$ and $h$, given a graph $G$ and a finite family $\mathcal{H} = \{H_j\}_{j\in J}$ of connected non-null subgraphs of $G$ such that $|V(H_j)| \leq h$ for every $j \in J$, a \textit{distance-$d$ $\mathcal{H}$-packing} in $G$ is a subfamily $\mathcal{H}' = \{H_i\}_{i\in I}$ of subgraphs from $\mathcal{H}$ that are at pairwise distance at least $d$. If we are also given a weight function $w\colon J \rightarrow \mathbb{Q}_{+}$, \textsc{Max Weight Distance-$d$ Packing} is the problem of finding a distance-$d$ $\mathcal{H}$-packing in $G$ of maximum weight. The case $d = 2$ coincides with \textsc{Max Weight Independent Packing}.

\subsection{Finding large induced sparse subgraphs satisfying a $\boldsymbol{\mathsf{CMSO}_2}$-definable near-monotone property in efficiently fractionally $\boldsymbol{\tin}$-fragile classes}

In this section we show that $(c, h, \psi)$-\textsc{Max Weight Induced Subgraph} admits a $\mathsf{PTAS}$ on every efficiently fractionally $\tin$-fragile class (Result \ref{item:first}). The following result will be crucial for our proof.  

\begin{theorem}[Lima et al.~\cite{LMM24}]\label{meta} For every $k$, $c$ and $\mathsf{CMSO}_2$ formula $\psi$, there exists a positive integer $g(k, c, \psi)$ such that the following holds. Let $G$ be a graph given along with a tree decomposition $\mathcal{T} =
(T, \{X_t\}_{t\in V(T)})$ of $G$ such that $\alpha(\mathcal{T}) \leq k$, and let $w\colon V(G) \rightarrow \mathbb{Q}_{+}$ be a weight function. Then, in time $g(k, c, \psi)\cdot |V (G)|^{O(R(k+1,c+1))} \cdot |V(T)|$, we can find a set $F \subseteq V(G)$ such that
\begin{itemize}
\item $G[F] \models \psi$,
\item $\omega(G[F]) \leq c$,
\item $F$ is of maximum weight subject to the conditions above,
\end{itemize}
or conclude that no such set exists.
\end{theorem}

\begin{theorem}\label{metaPTAS} Let $h \in \mathbb{N}$ and let $\psi$ be a $\mathsf{CMSO}_2$ formula expressing an $h$-near-monotone property, let $f \colon \mathbb{N} \rightarrow \mathbb{N}$ be a function and let $c \in \mathbb{N}$. There exists an algorithm that, given 
\begin{itemize}
\item $r  \in \mathbb{N}$ with $r > h$, 
\item a $n$-vertex graph $G$ equipped with a $(1-1/r)$-general cover $\mathcal{C} = \{C_1,C_2,\ldots\}$ and, for each $i$, a tree decomposition $\mathcal{T}_i=(T_i,\{X_t\}_{t\in V(T_i)})$ of $G[C_i]$ with $\alpha(\mathcal{T}_i) \leq f(r)$, 
\item and a weight function $w\colon V(G) \rightarrow \mathbb{Q}_{+}$, 
\end{itemize}
in time $|\mathcal{C}|\cdot g(f(r),c,\psi) \cdot t \cdot n^{O(R(f(r)+1,c+1))}$, where $t = \max_i |V(T_i)|$ and $g$ is the function from \Cref{meta}, either returns a subset $F \subseteq V(G)$ such that $G[F] \models \psi$, $\omega(G[F]) \leq c$, and $w(F)$ is at least a factor $(1-h/r)$ of the optimal, or concludes that no such set $F$ exists.
\end{theorem}

\begin{proof} Observe first that, if an admissible solution $F$ in $G$ exists (i.e., $F \subseteq V(G)$ is such that $G[F] \models \psi$ and $\omega(G[F]) \leq c$), then there exists a system $\{R_v \subseteq F : v \in F\}$ of subsets of $F$ such that $F\setminus \bigcup_{v\in F\setminus C_i} R_v$ is an admissible solution in $G[C_i]$, for each $i \geq 1$. Indeed, since $\psi$ expresses an $h$-near-monotone property, $G[F\setminus \bigcup_{v\in F\setminus C_i} R_v] \models \psi$. Moreover, $\omega(G[F\setminus \bigcup_{v\in F\setminus C_i} R_v]) \leq \omega(G[F]) \leq c$.

For each $i \geq 1$, we proceed as follows. Using the algorithm from \Cref{meta}, we simply look for optimal solutions $F_i \subseteq C_i$ in $G[C_i]$. For each $i$, finding $F_i$ or concluding that no such set exists can be done in time $g(f(r),c,\psi) \cdot n^{O(R(f(r)+1,c+1))} \cdot t$. The total running time is then $|\mathcal{C}|\cdot g(f(r),c,\psi) \cdot t \cdot n^{O(R(f(r)+1,c+1))}$.

If, for some $i \geq 1$, there is no admissible solution in $G[C_i]$, then we return that no admissible solution exists in $G$. Correctness follows from the first paragraph. Otherwise, if optimal solutions exist in $G[C_i]$ for every $i$, then an optimal solution $Y\subseteq V(G)$ in $G$ exists. Let now $\{R_v \subseteq Y : v \in Y\}$ be a system of subsets of $Y$ as in the definition of $h$-near-monotonicity. We pick a $C_i$ from $\mathcal{C}$ uniformly at random and observe that 
\begin{align*}
\mathbb{E}\bigg[w\bigg(Y\setminus \bigcup_{v\in Y\setminus C_i} R_v\bigg)\bigg]
& = \mathbb{E}\bigg[\sum_{y \in Y}  w(y) \mathbbm{1}_{\big\{y \notin \bigcup_{v \in Y\setminus C_i}R_v\big\}}\bigg] \\
& =  \sum_{y \in Y}  w(y) \mathbb{P}\bigg(y \notin \bigcup_{v \in Y\setminus C_i}R_v\bigg) \\
& =  \sum_{y \in Y}  w(y) \bigg(1-\mathbb{P}\bigg(y \in \bigcup_{v \in Y\setminus C_i}R_v\bigg)\bigg) \\
& =  \sum_{y \in Y}  w(y) \bigg(1-\mathbb{P}\bigg(\bigcup_{v : y \in R_v}{\{v \notin C_i\}}\bigg)\bigg) \\
& \geq  \sum_{y \in Y}  w(y)\bigg(1- \sum_{v : y \in R_v}\mathbb{P}(v \notin C_i)\bigg) \\
& \geq  \sum_{y \in Y}  w(y)\bigg(1- \sum_{v : y \in R_v}\frac{1}{r}\bigg) \\
& \geq (1-h/r)w(Y).
\end{align*}
Hence, there exists an index $j$ such that $w(Y\setminus \bigcup_{v\in Y\setminus C_j} R_v) \geq (1-h/r) w(Y)$. We then return the subset $F_j$ computed above. Since $Y\setminus \bigcup_{v\in Y\setminus C_j} R_v$ is an admissible solution in $G[C_j]$, the optimality of $F_j$ in $G[C_j]$ implies that $w(F_j) \geq w(Y\setminus \bigcup_{v\in Y\setminus C_j} R_v) \geq (1-h/r)w(Y)$. 
\end{proof}

We remark that the near-monotonicity requirement in \Cref{metaPTAS} is necessary. Indeed, for fixed $r$, it is not difficult to see that inducing an $r$-regular subgraph is a property expressible in $\mathsf{CMSO}_2$ but which is not near-monotone for $r\geq 2$. On the other hand, for each fixed $r \in \{3,4,5\}$, given a planar graph $G$, the problem of finding a maximum-size subset of $V(G)$ inducing an $r$-regular subgraph does not admit a PTAS, unless $\mathsf{P} = \mathsf{NP}$ \cite{AEIM14}. Moreover, boundedness of clique number of the solution is necessary as well. Indeed, being a clique is a monotone property that can be easily expressed by a $\mathsf{CMSO}_2$ formula. On the other hand, the problem of finding a maximum-size clique on unit ball graphs in $\mathbb{R}^4$ does not admit a PTAS, unless the Exponential Time Hypothesis (ETH) fails \cite{BBB21}. 

\subsection{Packing subgraphs at distance at least $\boldsymbol{2}$ in efficiently fractionally $\boldsymbol{\tin}$-fragile classes}

In this section we show that \textsc{Max Weight Independent Packing} admits a $\mathsf{PTAS}$ on every efficiently fractionally $\tin$-fragile class (Result \ref{item:second}). Such a PTAS relies on the following result.

\begin{theorem}[Dallard et al.~\cite{DaMS22}]\label{MWIPA} Let $k$ and $h$ be two positive integers. Given a graph $G$ and a finite family $\mathcal{H} = \{H_j\}_{j\in J}$ of connected non-null subgraphs of $G$ such that $|V(H_j)| \leq h$ for every $j \in J$, \textsc{Max Weight Independent Packing} can be solved in time $O(|V(G)|^{h(k+1)} \cdot |V(T)|)$ if $G$ is given together with a tree decomposition $\mathcal{T} = (T, \{X_t\}_{t\in V(T)})$ with $\alpha(\mathcal{T}) \leq k$.
\end{theorem}

\begin{theorem}\label{PTASffA} Let $h \in \mathbb{N}$ and let $f \colon \mathbb{N} \rightarrow \mathbb{N}$ be a function. There exists an algorithm that, given 
\begin{itemize}
\item $r \in \mathbb{N}$ with $r > h$, 
\item a $n$-vertex graph $G$ equipped with a $(1-1/r)$-general cover $\mathcal{C} = \{C_1,C_2,\ldots\}$ and, for each $i$, a tree decomposition $\mathcal{T}_i=(T_i,\{X_t\}_{t\in V(T_i)})$ of $G[C_i]$ with $\alpha(\mathcal{T}_i) \leq f(r)$, 
\item a finite family $\mathcal{H}=\{H_j\}_{j\in J}$ of connected non-null subgraphs of $G$ such that $|V(H_j)| \leq h$ for every $j \in J$, 
\item and a weight function $w\colon J \rightarrow \mathbb{Q}_{+}$ on the subgraphs in $\mathcal{H}$, 
\end{itemize}
returns in time $|\mathcal{C}|\cdot O(n^{h(f(r)+1)} \cdot t)$, where $t = \max_i |V(T_i)|$, an independent $\mathcal{H}$-packing in $G$ of weight at least a factor $(1-h/r)$ of the optimal. 
\end{theorem}

The proof of \Cref{PTASffA} is similar to that of \Cref{metaPTAS}.
 
\begin{proof} For each $i \geq 1$, we proceed as follows. Using the algorithm from \Cref{MWIPA}, we simply compute a maximum-weight independent $\mathcal{H}$-packing $\mathcal{P}_i$ in $G[C_i]$ in time $O(n^{h(f(r)+1)} \cdot t)$. The total running time is then $|\mathcal{C}|\cdot O(n^{h(f(r)+1)} \cdot t)$. For a collection $\mathcal{A}$ of subgraphs of $G$, each isomorphic to a member of $\mathcal{H}$, and a subset $C \subseteq V(G)$, let $w(\mathcal{A}) = \sum_{A\in \mathcal{A}}w(A)$ and let $\mathcal{A} \cap C = \{A \in \mathcal{A} : A \subseteq C\}$. Observe that, given a subgraph $H$ of $G$, each vertex $v \in V(H)$ is not contained in at most $|\mathcal{C}|/r$ elements of the $(1-1/r)$-general cover $\mathcal{C}$. Hence, $V(H)$ is contained in at least $(1-|V(H)|/r)|\mathcal{C}|$ elements of $\mathcal{C}$. Let $\mathcal{P}=\{P_1,P_2,\ldots\}$ be an independent $\mathcal{H}$-packing in $G$ of maximum weight. Then 

\begin{align*}
\sum_{C_i \in \mathcal{C}}w(\mathcal{P}\cap C_i)
& = \sum_{C_i \in \mathcal{C}} \sum_{P_j \in \mathcal{P}}  w(P_j) \mathbbm{1}_{\{P_j \subseteq C_i\}} \\
& =  \sum_{P_j \in \mathcal{P}}  w(P_j) \sum_{C_i \in \mathcal{C}}\mathbbm{1}_{\{P_j \subseteq C_i\}} \\
& \geq  \sum_{P_j \in \mathcal{P}}  w(P_j) (1-|V(P_j)|/r)|\mathcal{C}| \\
& \geq  \sum_{P_j \in \mathcal{P}}  w(P_j) (1-h/r)|\mathcal{C}| \\
& = |\mathcal{C}|(1-h/r)w(\mathcal{P}).
\end{align*}

By the pigeonhole principle, there exists $C_i \in \mathcal{C}$ such that $w(\mathcal{P}\cap C_i) \geq (1-h/r)w(\mathcal{P})$. We then return the maximum-weight independent $\mathcal{H}$-packing $\mathcal{P}_i$ in $G[C_i]$ computed above. Since $\mathcal{P} \cap C_i$ is an independent $\mathcal{H}$-packing in $G[C_i]$, we have that $w(\mathcal{P}_i) \geq w(\mathcal{P}\cap C_i) \geq (1-h/r)w(\mathcal{P})$.
\end{proof}

For the special case of \textsc{Max Weight Independent Set} on intersection graphs of $c$-fat collections of objects in a fixed $d$-dimensional space, we obtain the following.

\begin{corollary} There exists an algorithm that, given $r \in \mathbb{N}$, a $c$-fat collection $\mathcal{O}$ of $n$ objects in $\mathbb{R}^d$ and its intersection graph $G$, and a weight function $w\colon V(G) \rightarrow \mathbb{Q}_{+}$, returns in time $(f(r)/2-1)^d\cdot O(n^{cf(r)^{2d}+2})$, where $f(r) = 2\Big\lceil \frac{1}{1-\big(1-\frac{1}{r}\big)^{\frac{1}{d}}}\Big\rceil$, an independent set in $G$ of weight at least a factor $(1-1/r)$ of the optimal. 
\end{corollary}

\begin{proof} Given $r \in \mathbb{N}$, we use \Cref{treealphafatA} to compute in $O(n)$ time a $(1-1/r)$-general cover $\mathcal{C}$ of $G$ of size at most $(f(r)/2-1)^d$. Moreover, for each $C \in \mathcal{C}$, we compute in $O(n)$ time a tree decomposition $\mathcal{T} = (T,\{X_t\}_{t\in V(T)})$ of $G[C]$, with $|V(T)| \leq n+1$, such that $\alpha(\mathcal{T}) \leq cf(r)^{2d}$. We finally apply the algorithm from \Cref{PTASffA} (with $h=1$). The total running time is $(f(r)/2-1)^d\cdot O(n^{cf(r)^{2d}+2})$. 
\end{proof}


\subsection{Packing subgraphs at distance at least $\boldsymbol{d}$ in graphs with bounded layered tree-independence number or in intersection graphs of $\boldsymbol{c}$-fat collections}\label{sec:distanceA}

In this section we prove Result \ref{item:fourth}. The following result shows that, for each even $d\in \mathbb{N}$, \textsc{Max Weight Distance-$d$ Packing} admits a PTAS on every class of bounded layered tree-independence number.

\begin{theorem}\label{PTASltA} Let $h, \ell \in \mathbb{N}$. Let $d$ be an even positive integer. There exists an algorithm that, given 
\begin{itemize}
\item $r \in \mathbb{N}$ with $r > h$, 
\item a $n$-vertex graph $G$ equipped with a tree decomposition $\mathcal{T} = (T,\{X_t\}_{t\in V(T)}\})$ and a layering $(V_1, V_2, \ldots)$ of $G$ such that, for each bag $X_t$ and layer $V_i$, $\alpha(G[X_t \cap V_i]) \leq \ell$, 
\item a finite family $\mathcal{H}=\{H_j\}_{j\in J}$ of connected non-null subgraphs of $G$ such that $|V(H_j)| \leq h$ for every $j \in J$, 
\item and a weight function $w\colon J \rightarrow \mathbb{Q}_{+}$, 
\end{itemize}
returns, in time $r \cdot |V(T)| \cdot n^{O(r)}$, a distance-$d$ $\mathcal{H}$-packing in $G$ of weight at least a factor $(1-h/r)$ of the optimal. 
\end{theorem}

\begin{proof} Let $d = 2k$. As observed in \cite[Observation~5.14]{LMM24}, for $I \subseteq J$, the subfamily $\mathcal{H}' = \{H_i\}_{i\in I}$ is a distance-$d$ $\mathcal{H}$-packing in $G$ if and only if $\mathcal{H}'$ is an independent $\mathcal{H}$-packing in the graph $G^{d-1}$. Therefore, using BFS, we first compute in $O(n^3)$ time the graph $G^{2k-1}$. Using the algorithm from \Cref{layeredtreealgoA}, we compute in $O(|V(T)| \cdot n^2)$ time a tree decomposition $\mathcal{T}' = (T,\{X'_t\}_{t\in V(T)})$ of $G^{2k-1}$ and a layering $(V'_1,\ldots,V'_{\lceil \frac{m}{2k-1} \rceil})$ of $G^{2k-1}$ such that, for each bag $X'_t$ and layer $V'_i$, $\alpha(G^{2k-1}[X'_t \cap V'_i]) \leq (4k-3)\ell$. Using the algorithm from \Cref{layeredtofragileA}, we compute in $O(n^2 + n\cdot|V(T)|)$ time a $(1 - 1/r)$-general cover $\mathcal{C}$ of $G^{2k-1}$ of size $r$ and, for each $C \in \mathcal{C}$, a tree decomposition of $G^{2k-1}[C]$ with independence number at most $\ell (4k-3)(r-1)$. Finally, we apply the approximation algorithm from \Cref{PTASffA} to obtain, in time $r \cdot O(n^{h(\ell (4k-3)(r-1)+1)} \cdot |V(T)|)$, an independent $\mathcal{H}$-packing in $G^{2k-1}$ of weight at least a factor $(1-h/r)$ of the optimal.
\end{proof}

In a similar way, combining \Cref{PTASffA} with \Cref{treealphafatA,c-fatoddpower}, we immediately obtain the following:

\begin{theorem}\label{PTASintpower} Let $d \in \mathbb{N}$ be even. \textsc{Max Weight Distance-$d$ Packing} admits a $\mathsf{PTAS}$ for intersection graphs of $c$-fat collections of objects in $\mathbb{R}^k$, for any fixed $c$ and $k$. 
\end{theorem}

We remark that \Cref{PTASltA} cannot be extended to odd values of $d$, unless $\mathsf{P} = \mathsf{NP}$. Indeed, Eto et al.~\cite{EGM14} showed that, for each $\varepsilon > 0$ and fixed odd $d \geq 3$, it is $\mathsf{NP}$-hard to approximate \textsc{Distance-$d$ Independent Set} to within a factor of $n^{1/2 - \varepsilon}$ for chordal graphs. However, it is not clear whether \Cref{PTASintpower} can be extended to odd values of $d$.


\subsection{Packing independent unit disks, unit-width rectangles and paths with bounded horizontal part on a grid}\label{sec:improvedA}

We have seen in \Cref{sec:layeredlemmas} that classes of intersection graphs of unit disks, unit-width rectangles and paths with bounded horizontal part on a grid have bounded layered tree-independence number and hence \textsc{Max Weight Independent Set} admits a PTAS when restricted to any of them thanks to either \Cref{metaPTAS} or \Cref{PTASffA}. In this section we provide PTASes with improved running time for \textsc{Max Weight Independent Set} on these three classes.

We begin with VPG/EPG graphs. The following result generalizes the PTAS for \textsc{Max Weight Independent Set} on $B_1$-EPG graphs with bounded horizontal part given in \cite[Theorem~6]{BBCGP20}. 

\begin{theorem}\label{indepPTASA} Let $c \in \mathbb{N}$. \textsc{Max Weight Independent Set} admits a PTAS when restricted to $n$-vertex graphs with a grid representation $\mathcal{R} = (\mathcal{G}, \mathcal{P},x)$ such that: 
\begin{enumerate}
\item each path in $\mathcal{P}$ has number of bends constant;
\item\label{3rdA} the horizontal part of each path in $\mathcal{P}$ has length at most $c$.
\end{enumerate}
If $x = v$, the running time is $O(c\lceil\frac{1}{\varepsilon}\rceil \cdot n^{\lceil\frac{1}{\varepsilon}\rceil c + 3})$. If $x = e$, the running time is $O(c\lceil\frac{1}{\varepsilon}\rceil \cdot n^{3\lceil\frac{1}{\varepsilon}\rceil c + 2})$. 
\end{theorem}

\begin{proof}
Let $G$ be a $n$-vertex graph with a grid representation $\mathcal{R} = (\mathcal{G}, \mathcal{P},x)$ satisfying the conditions above. Without loss of generality, we may assume that all the paths in $\mathcal{P}$ contain only grid-points with non-negative coordinates. Moreover, we may assume that $G$ is connected. Therefore, no column in $\mathcal{G}$ is unused and so $\mathcal{G}$ has at most $(c+1)n$ columns. Further note that since any path $P \in \mathcal{P}$ has number of bends constant, we can compute the horizontal part $h(P)$ of $P$ in $O(1)$ time. Given $0 < \varepsilon < 1$, we fix $k = \lceil 1/\varepsilon \rceil$. 

For any $i \geq 0$, we denote by $X_i$ the set of vertices whose corresponding path contains a grid-edge $[(i,j),(i+1,j)]$ for some $j \geq 0$ (here and in the following $[(i,j),(i+1,j)]$ denotes the grid-edge with endpoints $(i,j)$ and $(i+1,j)$). Note that we can compute the at most $(c+1)n -1$ non-empty sets $X_{i}'s$ in $O(n)$ time. In view of applying a shifting argument, we now partition $G$ into slices via the following. For any $d \in \{0, \ldots, kc-1\}$, let $V_d = \bigcup_{\ell \in \mathbb{N}_0} X_{d + \ell kc}$ be the set of vertices whose corresponding path contains a grid-edge $[(d +\ell kc,j),(d+\ell kc +1,j)]$ for some $\ell,j \in \mathbb{N}_0$. We claim that, for any $d \in \{0,\ldots,kc-1\}$, $G - V_d$ is disconnected. Indeed, after deleting $V_d$, no vertex whose horizontal part is contained in the interval $[0, d + \ell kc]$ can be adjacent to a vertex whose horizontal part is contained in the interval $[d + \ell kc+1, (c+1)n]$. Similarly, every component of $G - V_d$ admits a grid representation in which the number of columns is bounded by $kc$. By \Cref{pathA} and \Cref{MWIPA}, for each component of $G - V_d$, we compute a maximum-weight independent set in $O(n^{kc+2})$ time, if $x = v$, or in $O(n^{3kc+1})$ time, if $x = e$. The union $U_d$ of these independent sets over the components of $G - V_d$ is then an independent set of $G$ and, after repeating the procedure above for each $d \in \{0,\ldots, kc-1\}$, we return the maximum-weight set $U$ among the $U_d$'s. The total running time is then $O(kc \cdot n^{kc+3})$, if $x = v$, or $O(kc \cdot n^{3kc+2})$, if $x = e$.  
 
It remains to show that $w(U) \geq (1 - \varepsilon)w(\mathsf{OPT})$, where $\mathsf{OPT}$ denotes an optimal solution of \textsc{Max Weight Independent Set} with instance $G$. Note that, for any $d \in \{0,\ldots , kc-1\}$, $\mathsf{OPT} \cap V_d$ is the set of vertices in $\mathsf{OPT}$ whose corresponding path contains a grid-edge $[(d +\ell kc,j),(d+\ell kc +1,j)]$ for some $\ell,j \in \mathbb{N}_0$. Since the horizontal part of each path has length at most $c$, we have that every vertex in $\mathsf{OPT}$ belongs to at most $c$ distinct $V_d$'s. Therefore, denoting by  $d_0$ an index attaining $\min_{d \in \{0, \ldots, kc-1\}} w(\mathsf{OPT} \cap V_d)$, we have that 
\[
kc\cdot w(\mathsf{OPT} \cap V_{d_0}) \leq \sum_{d=0}^{kc-1} w(\mathsf{OPT} \cap V_d) \leq c\cdot w(\mathsf{OPT})
\] and so 
\[
w(\mathsf{OPT}) = w(\mathsf{OPT}\setminus V_{d_0}) + w(\mathsf{OPT} \cap V_{d_0}) \leq w(U) + \varepsilon \cdot w(\mathsf{OPT}),
\]
thus concluding the proof.
\end{proof}

We note two consequences of \Cref{indepPTASA}. It provides a PTAS for \textsc{Max Weight Independent Set} on equilateral $B_1$-VPG graphs (i.e., $B_1$-VPG graphs where, for each path, its horizontal and vertical segment have the same length) where paths have bounded horizontal part and on unit $B_k$-VPG graphs (i.e., $B_k$-VPG graphs where each segment has unit length). Lahiri et al.~\cite{LMS15} provided a $O(\log d)$-approximation algorithm for the unweighted version on equilateral $B_1$-VPG graphs, where $d$ denotes the ratio between the maximum and minimum length of segments of paths, which for constant $d$ gives a constant-factor approximation algorithm. \Cref{indepPTASA} improves this to a PTAS. \Cref{indepPTASA} also complements the $O(k^4)$-approximation algorithm for \textsc{Dominating Set} on unit $B_k$-VPG graphs provided by Chakraborty et al.~\cite{CDM22}.

We now pass to the PTASes for intersection graphs of unit disks and unit-width rectangles. The running time of our PTAS for unit disk graphs improves on the $(1-1/k)$-approximation algorithm with running time $O(k n^{4\lceil \frac{2(k-1)}{\sqrt{3}}\rceil})$ by Matsui~\cite{Mat98}.

\begin{theorem}\label{rectanglesPTASA} \textsc{Max Weight Independent Set} admits a PTAS when restricted to:
\begin{itemize}
\item Intersection graphs of a family $\mathcal{D}$ of $n$ unit disks of common radius $1$. The running time is $O(\lceil \frac{3}{\varepsilon} \rceil \cdot n^{3\lceil\frac{3-\varepsilon}{2\varepsilon}\rceil + 3})$.
\item Intersection graphs of a family $\mathcal{R}$ of $n$ unit-width rectangles of common width $1$. The running time is $O(\lceil \frac{2}{\varepsilon} \rceil \cdot n^{\lceil \frac{2}{\varepsilon} \rceil + 2})$.
\end{itemize}
\end{theorem}

\begin{proof} Since the PTASes are similar, we introduce some common notation. Let $G$ be the intersection graph of the family $\mathcal{O}$, where $\mathcal{O}$ is either $\mathcal{D}$ or $\mathcal{R}$. Without loss of generality, we may assume that all objects in $\mathcal{O}$ are contained in the positive quadrant. Moreover, we may assume that $G$ is connected. Therefore, $\mathcal{O}$ is contained in a grid $\mathcal{G}$ with $O(n)$ columns. For each $O \in \mathcal{O}$, we compute the horizontal part $h(O)$ of $O$ (i.e., the projection of $O$ onto the $x$-axis) in $O(1)$ time. Given $0 < \varepsilon < 1$, we fix $k = \lceil 2/\varepsilon \rceil$, in the case of rectangles, and $k = \lceil 3/\varepsilon \rceil$, in the case of disks. Let $X_i$ be the set of vertices whose corresponding objects have horizontal part intersecting the half-open segment $[(i,0),(i+1,0))$. We can compute the $X_i$'s in $O(n)$ time. We now partition $G$ into slices as follows. For any $d \in \{0, \ldots, k-1\}$, let $V_d = \bigcup_{\ell \in \mathbb{N}_0} X_{d + \ell k}$. As in the proof of \Cref{indepPTASA}, it is easy to see that, for any $d \in \{0,\ldots,k-1\}$, each component of $G - V_d$ admits a geometric realization which is contained in an axis-aligned rectangle with width at most $k - 1$. If the family consists of disks, for each component of $G - V_d$, we compute a maximum-weight independent set in $O(n^{3\lceil\frac{k-1}{2}\rceil + 2})$ time (thanks to \Cref{diskA} and \Cref{MWIPA}). If the family consists of rectangles, for each component of $G - V_d$, we compute a maximum-weight independent set in $O(n^{k + 1})$ time (thanks to \Cref{rectangleA} and \Cref{MWIPA}). In either case, the union $U_d$ of these independent sets over the components of $G - V_d$ is an independent set of $G$ and, after repeating the procedure above for each $d \in \{0,\ldots, k-1\}$, we return the maximum-weight set $U$ among the $U_d$'s. The total running time is then $O(k \cdot n^{3\lceil\frac{k-1}{2}\rceil + 3})$, in the case of disks, and $O(k \cdot n^{k + 2})$, in the case of rectangles. Similarly to \Cref{indepPTASA}, it is easy to see that $w(U) \geq (1 - \varepsilon)w(\mathsf{OPT})$, where $\mathsf{OPT}$ denotes an optimal solution of \textsc{Max Weight Independent Set} with instance $G$. 
\end{proof}

\section{Subexponential-time algorithms}\label{sec:subexp}

Although the focus of this paper is on approximation schemes, in this short section we note some consequences of our work in relation to subexponential-time algorithms. In particular, the observations from \Cref{sec:layeredA} immediately lead to subexponential-time algorithms for \textsc{Max Weight Distance-$d$ Packing}, for each fixed even $d\in\mathbb{N}$, on classes of bounded layered tree-independence number, as we argue below. We have listed in \Cref{sec:ptasesA} some examples of problems falling in this framework. Further well-known examples of dual problems (minimization problems) are \textsc{$k$-Separator} (also known as \textsc{$k$-Component Order Connectivity} in its unweighted version) \cite{BMN15, Lee19} and its special case \textsc{Min Weight $3$-Path Vertex Cover} \cite{BKSS14,Lee19}.

\begin{corollary}\label{subexp} Let $\ell,d\in\mathbb{N}$ be fixed constants, with $d$ even. Let $G$ be a $n$-vertex graph for which we can compute, in time $\mathsf{poly}(n)$, a tree decomposition and a layering witnessing layered tree-independence number at most $\ell$. Then \textsc{Max Weight Distance-$d$ Packing} can be solved in $2^{O(\sqrt{n}\log{n})}$ time\footnote{An alternative way to obtain subexponential-time algorithms for some of the \textit{unweighted} problems captured by \textsc{Max Weight Distance-$d$ Packing} follows from a general result of Korhonen and Lokshtanov~\cite{KL23}, who provided $2^{O_{H}(n^{2/3}\log n)}$-algorithms on $H$-induced-minor-free graphs for unweighted problems such as \textsc{Independent Set} and \textsc{Induced Matching}. This result applies to every class of bounded layered tree-independence number, as this must be $K_{n,n}$-induced-minor-free, for some $n$.}.
\end{corollary}

\begin{proof} We first compute, in time $\mathsf{poly}(n)$, a tree decomposition and a layering witnessing layered tree-independence number at most $\ell$. We then apply \Cref{sqrttreealphaA} and compute a tree decomposition of $G$ with independence number $O(\sqrt{n})$. It is then enough to recall that \textsc{Max Weight Distance-$d$ Packing} is solvable in time $n^{O(k)}$, where $k$ is the tree-independence number of the input graph, provided a tree decomposition with independence number at most $k$ is given in input \cite{LMM24}.
\end{proof}

\begin{remark} Note that, in \Cref{subexp}, we require a tree decomposition and a layering giving constant layered tree-independence number. This is because we cannot directly use $O(\sqrt{n})$ tree-independence number in conjunction with the $\mathsf{XP}$-approximation algorithm for tree-independence number of Dallard et al.~\cite{DFGK22} (for fixed $k$, there exists an algorithm that, given a $n$-vertex graph $G$, in time $2^{O(k^2)}n^{O(k)}$, either decides that $\tin(G) > k$, or outputs a tree decomposition of $G$ with independence number at most $8k$), as this would not grant subexponential time. 
\end{remark}

\Cref{subexp} has interesting consequences. In particular, paired with \Cref{layeredfat}, it immediately gives a $2^{O(\sqrt{n}\log{n})}$-time algorithm for \textsc{Max Weight Distance-$d$ Packing} on intersection graphs of similarly-sized $c$-fat families of objects in $\mathbb{R}^2$. A $2^{O(\sqrt{n}\log{n})}$-time algorithm for the special case $d=2$ (i.e., for \textsc{Max Weight Independent Set}) on unit disk graphs was first given in \cite{TP02}. More recently, de Berg et al.~\cite{BBK20} provided $2^{O(\sqrt{n})}$-time algorithms for the \textit{unweighted} version of many problems on intersection graphs of similarly-sized globally fat objects in $\mathbb{R}^d$ and this is tight under the ETH. Examples of such problems are \textsc{Independent Set}, \textsc{Induced Matching}, \textsc{Distance-$d$ Dominating Set} and, more generally, problems whose solutions (or the complements thereof) can contain at most a constant number of vertices from any clique. 

Although the $2^{O(\sqrt{n})}$-time algorithm for \textsc{Independent Set} from \cite{BBK20} can be extended to handle the weighted version (see, e.g., \cite{BK20}), to the best of our knowledge very little is known about the weighted case in general. The only lower bound we are aware of is by de Berg and Kisfaludi{-}Bak~\cite{BK20}, hinting at the fact that the weighted versions could be considerably harder: \textsc{Min Weight Dominating Set} cannot be solved in $2^{o(n)}$ time on unit ball graphs in $\mathbb{R}^3$, unless the ETH fails. In fact, de Berg and Kisfaludi{-}Bak~\cite{BK20} asked to determine the complexity of the weighted versions of problems falling in the framework of \cite{BBK20} when restricted to intersection graphs of similarly-sized fat objects in $\mathbb{R}^2$. \Cref{subexp} partially answers this question and shows that also the weighted versions of several of these problems admit subexponential-time algorithms, at the cost of an extra $\log n$ factor in the exponent. It is then natural to ask whether this $\log n$ factor can be shaved off or not. In other words, do the complexities of the weighted and unweighted versions match in $\mathbb{R}^2$? 

Finally, it would be interesting to investigate whether the notion of fractional $\tin$-fragility could be useful in the design of subexponential-time algorithms.

\section{Concluding remarks and open problems}\label{sec:remarks}

In this work we began investigating the notion of fractional $\tin$-fragility, which allows to unify classes of sparse graphs (e.g., proper minor-closed) and dense graphs (e.g., geometric intersection graphs), and demonstrated its usefulness in the design of PTASes for maximization problems. In particular, we obtained an approximation meta-theorem for the broad family of fractionally $\tin$-fragile classes. Besides generality, we believe that an additional strength of this approach lies in the simplicity of the approximation schemes obtained. We now list some open questions whose answers allow to identify the applicability limits of our frameworks. 

The first natural direction for a possible extension is given by the following.

\begin{question} Is it true that, for each even $d\in \mathbb{N}$, \textsc{Max Weight Distance-$d$ Packing} admits a PTAS on every efficiently fractionally $\tin$-fragile class?  
\end{question}

Recall that, for each odd $d \geq 3$, we cannot hope for a positive answer to the previous question, unless $\mathsf{P} = \mathsf{NP}$ \cite{EGM14}. However, it is not clear whether the following question has still a negative answer.  

\begin{question} Let $d \in \mathbb{N}$ be odd. Does \textsc{Max Weight Distance-$d$ Packing} admit a $\mathsf{PTAS}$ for intersection graphs of $c$-fat collections of objects in $\mathbb{R}^k$, for fixed $c$ and $k$? 
\end{question}

We observed in \Cref{separators} that every fractionally $\tin$-fragile class has separators of sublinear independence number. The converse does not hold in general, as shown by the following construction. Let $\mathcal{G}$ be the class of triangle-free graphs $G$ on $n$ vertices and with $\alpha(G) = O(\sqrt{n\log n})$ (such graphs exist for every sufficiently large $n$ by \cite{Kim95}). Let now $\mathcal{G}'$ be the class consisting of those graphs obtained by taking two copies of a graph in $\mathcal{G}$ and adding all possible edges between the two copies. The class $\mathcal{G}'$ has separators of sublinear independence number but, by Ramsey's theorem, 
$\mathcal{G}'$ has unbounded induced biclique number and hence is not fractionally $\tin$-fragile thanks to \Cref{bicliqueA}. However, the following might be true.

\begin{question}\label{qsublinear} Is it true that every hereditary class with separators of sublinear independence number is fractionally $\tin$-fragile?      
\end{question}

We remark that the similar question of whether every subgraph-closed class with strongly sublinear separators is fractionally $\tw$-fragile was asked by Dvo\v{r}\'{a}k~\cite{Dvo16} and is still open. It turns out that the answer is positive under the additional assumption of bounded maximum degree (see \cite[Lemma~19]{Dvo16}).   

A \textit{clique-based separator} of a graph $G$ is a collection $\mathcal{S}$ of vertex-disjoint cliques whose union is a balanced separator of $G$. The \textit{weight} of $\mathcal{S}$ is the quantity $\sum_{C\in\mathcal{S}} \log(|C|+1)$. de Berg et al.~\cite{BBK20,BKMT23} showed that several intersection graphs of geometric objects in the plane (e.g., map graphs and intersection graphs of convex globally fat objects, similarly-sized globally fat objects, or pseudo-disks) admit clique-based separators of weight $O(\sqrt{n})$ and used this to obtain algorithms with running time $2^{O(\sqrt{n})}$ for many problems on such graphs, as remarked in \Cref{sec:subexp}. Clearly, if a graph class $\mathcal{G}$ admits clique-based separators of sublinear weight, then it also admits separators of sublinear independence number. However, the class $\mathcal{G}'$ defined above shows that the converse does not hold, as each clique-based separator of a graph in $\mathcal{G}'$ has weight $\Omega(n)$. 

\begin{question} Is it true that every class of bounded layered tree-independence number admits clique-based separators of sublinear weight?
\end{question}

It would be interesting to further investigate the notions of layered treewidth, layered tree-independence number and fractional $\tin$-fragility for classes of string graphs and intersection graphs of pseudo-disks. Recall that a string graph is the intersection graph of a set of curves in the plane, where it can be assumed that no three curves meet at a single point (see, e.g., \cite{asinowski}). For an integer $k \geq 2$, if each curve is in at most $k$ intersections with other curves, then the corresponding string graph is called a \textit{$k$-string graph}. A \textit{$(g, k)$-string graph} is defined analogously for curves on a surface of Euler genus at most $g$. In general, the class of string graphs is not fractionally $\tin$-fragile (see \Cref{bicliqueA}). However, Dujmovi\'c et al.~\cite{DJM18} showed that the class of $(g, k)$-string graphs has bounded layered treewidth. A natural question is what happens for its superclass\footnote{The maximum degree of a $k$-string graph might be much less than $k$.} of string graphs of bounded maximum degree: 

\begin{question} Does the class of string graphs of bounded maximum degree have bounded layered treewidth?
\end{question}

More generally, one can consider $K_{t,t}$-subgraph-free string graphs. Thanks to a result of Lee~\cite{Lee17}, these graphs have separators of size $O(\sqrt{n})$. However, the layered tree-independence number is unbounded, as is easily seen by considering $2$-dimensional grids with a dominating vertex.

\begin{question} Let $t \geq 2$. Is the class of $K_{t,t}$-subgraph-free string graphs fractionally $\tin$-fragile? More generally, is the class of $K_{t,t}$-free string graphs fractionally $\tin$-fragile?
\end{question}

A set of objects in the plane is a collection of \textit{pseudo-disks} if each object is bounded by a Jordan curve and, for each pair of objects, their boundaries intersect at most twice. de Berg et al.~\cite{BKMT23} showed that any intersection graph of pseudo-disks has a clique-based separator of weight $O(n^{2/3}\log n)$. \Cref{fragileunboundedlocal} implies that the class of intersection graphs of pseudo-disks has unbounded layered tree-independence number. However, we ask the following.  

\begin{question} Is the class of intersection graphs of pseudo-disks fractionally $\tin$-fragile?
\end{question}

\bibliography{references}

\end{document}